\RequirePackage{tikz}  % THIS IS NECESSARY BECAUSE WHEN USING THE SPRINGER TEMPLATE
\RequirePackage{scrlfile}
\PreventPackageFromLoading{program}

\documentclass[sn-basic]{sn-jnl}% Basic Springer Nature Reference Style/Chemistry Reference Style

\usepackage{amsmath,amssymb,amsfonts,amsthm}
\usepackage{algorithm}
\usepackage{algpseudocode}
\usepackage{algcompatible}
  \renewcommand{\algorithmicrequire}{\textbf{Input:}}
  \renewcommand{\algorithmicensure}{\textbf{Output:}}
  \renewcommand{\algorithmiccomment}[1]{/* #1 */}
\usepackage{graphicx}
\usepackage{textcomp}
\usepackage{xcolor}
\usepackage{stmaryrd}
\usepackage[bb=boondox]{mathalfa}

\usepackage{mathtools}
\usepackage{adjustbox}
\usepackage{pgfplots}
\pgfplotsset{compat=1.17}

\usepackage{xspace}
\usepackage[official]{eurosym}
\usepackage{wrapfig}
\usepackage[per-mode=symbol]{siunitx}
\usepackage{listings}
\usepackage{booktabs}
\usepackage{caption}
\usepackage{subcaption}
\usepackage{makecell}
\usepackage{adjustbox}
\usepackage{ltl}
\usepackage{csquotes}

\usepackage{microtype}

\usepackage{thmtools,thm-restate}

\hypersetup{breaklinks=true}
\usepackage{bbding}
\usepackage{cleveref}

\usepackage{tikz}
\usetikzlibrary{arrows}
\usetikzlibrary{shapes}
\usetikzlibrary{positioning}
\usetikzlibrary{fit}
\tikzset{cross/.style={cross out, draw=black, minimum size=2*(#1-\pgflinewidth), inner sep=0pt, outer sep=0pt},
cross/.default={1pt}}

\newcommand{\computationModel}[1]{\ensuremath{{\mathsf{#1}}}\xspace}
\newcommand{\seqSys}{\computationModel{P}}

\newcommand{\reactSys}{\computationModel{S}}

\newcommand{\boolSys}{\ensuremath{\hat{\reactSys}}\xspace}
\newcommand{\apSys}{\computationModel{C}}

\newcommand{\ltsSys}{\computationModel{L}}

\newcommand{\rValSys}{\computationModel{M}}
\newcommand{\SysSTL}{\rValSys}
\newcommand{\SysU}{\rValSys}

\newcommand{\mixedIOSys}{\ltsSys}
\newcommand{\tracelang}{\mixedIOSys}

\newcommand{\cleanFor}{w.r.t.\ }

\newcommand{\robustlyCleanNDet}{robustly clean\xspace}
\newcommand{\robustCleannessNDetPar}{robust cleanness\xspace}
\newcommand{\robustCleannessNDet}{robust cleanness\xspace}

\newcommand{\RobustCleannessNDet}{Robust cleanness\xspace}

\newcommand{\robustlyLowCleanNDetPar}{l\nobreakdash-rob\-ustly clean\xspace}
\newcommand{\robustlyLowCleanNDet}{l\nobreakdash-rob\-ustly clean\xspace}
\newcommand{\robustLowCleannessNDetPar}{l-robust cleanness\xspace}
\newcommand{\robustLowCleannessNDet}{l-robust cleanness\xspace}

\newcommand{\robustlyUpCleanNDetPar}{u\nobreakdash-rob\-ustly clean\xspace}
\newcommand{\robustlyUpCleanNDet}{u\nobreakdash-rob\-ustly clean\xspace}
\newcommand{\robustUpCleannessNDetPar}{u-robust cleanness\xspace}
\newcommand{\robustUpCleannessNDet}{u-robust cleanness\xspace}

\newcommand{\fCleanNDet}{func-clean\xspace}
\newcommand{\fCleannessNDetPar}{func-cleanness\xspace}
\newcommand{\fCleannessNDet}{func-cleanness\xspace}

\newcommand{\fLowCleanNDetPar}{l-func-clean\xspace}
\newcommand{\fLowCleanNDet}{l-func-clean\xspace}

\newcommand{\fLowCleannessNDet}{l-func-cleanness\xspace}

\newcommand{\fUpCleanNDetPar}{u-func-clean\xspace}
\newcommand{\fUpCleanNDet}{u-func-clean\xspace}

\newcommand{\fUpCleannessNDet}{u-func-cleanness\xspace}

\newcommand{\traceIntegral}{trace integral\xspace}

\newcommand{\NOx}{\ensuremath{\mathrm{NO}_x}\xspace}

\newcommand{\drivingCycle}[1]{\textrm{#1}\xspace}
\newcommand{\NEDC}{\drivingCycle{NEDC}}

\newcommand{\WLTC}{\drivingCycle{WLTC}}
\newcommand{\WLTP}{\drivingCycle{WLTP}}
\newcommand{\RDE}{\drivingCycle{RDE}}

\newcommand{\powerSet}[1]{\ensuremath{2^{{#1}}}}


\newcommand{\quiet}{\ensuremath{\delta}\xspace}
\newcommand{\quiescence}{\quiet}

\newcommand{\Inputs}{{\mathsf{In}}}

\newcommand{\Outputs}{{\mathsf{Out}}}

\newcommand{\Norm}{{\mathsf{StdIn}}}

\newcommand{\Std}{\computationModel{Std}}

\newcommand{\StdSet}{\Std}
\newcommand{\StdIn}{\Norm}

\newcommand{\dIn}{\ensuremath{d_\Inputs}\xspace}
\newcommand{\dOut}{\ensuremath{d_\Outputs}\xspace}

\newcommand{\mapTrace}[2]{#1 \ensuremath{{\downarrow_{#2}}}}
\newcommand{\mapInp}[1]{#1 \ensuremath{{\downarrow_\inp}}}
\newcommand{\mapOut}[1]{#1 \ensuremath{{\downarrow_\outp}}}

\newcommand{\mapAPi}[1]{#1 \ensuremath{{\downarrow_{\ap_\inp}}}}

\newcommand{\dInTL}{\dIn}
\newcommand{\dOutTL}{\dOut}
\newcommand{\dInMIO}{\dIn}
\newcommand{\dOutMIO}{\dOut}

\newcommand{\Contract}{\ensuremath{\mathcal{C}}\xspace}
\newcommand{\inp}{\mathsf{i}}

\newcommand{\outp}{\mathsf{o}}

\newcommand{\inpbound}{\kappa_\inp}
\newcommand{\outpbound}{\kappa_\outp}
\newcommand{\NoInp}{\ensuremath{\text{--}_\inp}}
\newcommand{\NoOutp}{\ensuremath{\text{--}_\outp}} 

\newcommand{\hrwoman}{\text{Unica}\xspace}
\newcommand{\john}{\text{John}\xspace}
\newcommand{\synthiaA}{\text{Synthia}\xspace}
\newcommand{\synthiaB}{\text{Synclair}\xspace}

\newcommand{\alexa}{\text{Alexa}\xspace}
\newcommand{\syntbad}{\text{Syntbad}\xspace}
\newcommand{\evgenji}{\text{Eugene}\xspace}
\newcommand{\syna}{\text{Syna}\xspace}

\newcommand{\numberofapplicants}{4096\xspace}
\newcommand{\numberofapplicantsminusthree}{4093\xspace}

\newcommand{\eduMark}{\ensuremath{\mathsf{ed}}}
\newcommand{\expMark}{\ensuremath{\mathsf{ex}}}
\newcommand{\personMark}{\ensuremath{\mathsf{pe}}}
\newcommand{\intelliMark}{\ensuremath{\mathsf{in}}}
\newcommand{\skillMark}{\ensuremath{\mathsf{sk}}}

\newcommand{\fairnessscore}{\ensuremath{\xi }}

\newcommand{\markSet}{\ensuremath{\mathcal{M}}}

\newcommand{\PS}{{\mathsf{PS}}}
\newcommand{\rob}{\fairnessscore}

\newcommand{\fWrapper}{\ensuremath{\mathsf{FairnessAwareSystem}}\xspace}

\newcommand{\fMinimiser}{\ensuremath{\mathsf{FairnessMonitor}}\xspace}
\newcommand{\robmin}{\ensuremath{\rob\text{-}\mathsf{min}}\xspace}

\newcommand{\aisystem}{system\xspace}
\newcommand{\Aisystem}{System\xspace}

\newcommand{\fContract}{\ensuremath{\mathcal{F}}\xspace}
\newcommand{\fairCon}{fairness contract\xspace}
\newcommand{\FairCon}{Fairness contract\xspace}

\newcommand{\ffairC}{func-fair\xspace}

\newcommand{\ffairnessC}{func-fairness\xspace}
\newcommand{\FfairnessC}{Func-fairness\xspace}
\newcommand{\Is}{\ensuremath{{\mathcal{I}}}}

\newcommand{\synIndex}{\mathsf{s}}
\newcommand{\synInp}{\ensuremath{\inp_\synIndex}}

\newcommand{\realIndex}{\mathsf{a}}
\newcommand{\realInp}{{\ensuremath{\inp_\realIndex}}}

\newcommand{\NOT}{\textbf{not}\xspace}
\newcommand{\stdcnt}{\ensuremath{c}}
\newcommand{\predictor}{\ensuremath{\mathcal{P}}}
\newcommand{\recording}{\ensuremath{\mathcal{D}}\xspace}
\newcommand{\w}{{w}}
\newcommand{\Reals}{\ensuremath{\mathbb{R}}}

\newcommand{\N}{\ensuremath{\mathbb{N}}}
\newcommand{\minimise}{\mathsf{minimise}}

\newcommand{\LolaDrives}{\ensuremath{{\mathsf{LolaDrives}}}\xspace}

\newcommand{\RR}{\ensuremath{\mathbb{R}}}
\newcommand{\RRpos}{\ensuremath{\RR_{\geq 0}}}
\newcommand{\RRposinf}{\ensuremath{\RRinf_{\geq 0}}}
\newcommand{\RRinf}{\ensuremath{\overline{\RR}}}

\newcommand{\NN}{\ensuremath{\mathbb{N}}}

\newcommand{\timeSet}{\ensuremath{\mathcal{T}}\xspace}

\newcommand{\varTrace}{\ensuremath{t}}
\newcommand{\varTime}{\ensuremath{t}}

\newcommand{\lImp}{\Rightarrow}
\newcommand{\lIff}{\Leftrightarrow}
\newcommand{\limp}{\rightarrow}

\newcommand{\true}{\textsf{true}}

\renewcommand{\emptyset}{\varnothing}
\newcommand{\set}[1]{\ensuremath{\{ #1 \}}}
\newcommand{\abs}[1]{\lvert {#1} \rvert}

\newcommand{\tr}{\varTrace}
\newcommand{\ap}{\mathsf{AP}}

\newcommand{\X}{\mathop{\mathsf{X}}}
\newcommand{\F}{\LTLfinally}
\newcommand{\G}{\LTLglobally}
\newcommand{\U}{\LTLuntil} 

\newcommand{\W}{\LTLweakuntil}
\newcommand{\EEE}[1]{\mathop{\exists{#1}.}}
\newcommand{\E}{\EEE}
\newcommand{\AAA}[1]{\mathop{\forall{#1}.}}
\newcommand{\A}{\AAA}

\newcommand{\RETURN}{\State \textbf{return }}
\newif\ifhideplots

\renewcommand{\sc}{\textsc}
\newcommand{\noop}[1]{}

\jyear{2023}%

\theoremstyle{thmstyleone}%
\newtheorem{theorem}{Theorem}%  meant for continuous numbers
\newtheorem{proposition}[theorem]{Proposition}% 
\newtheorem{lemma}[theorem]{Lemma}% 

\theoremstyle{thmstyletwo}%
\newtheorem{example}{Example}%
\theoremstyle{thmstylethree}%
\newtheorem{definition}{Definition}%

\newcounter{auxFootnote}

\begin{document}
\def\qedsymbol{\ensuremath{\opensquare}} % <- Need this, because we must prevent package 'program' from loading; overwrites definition in sn-jnl.cls

\title{Software Doping Analysis for~Human~Oversight}


\author*[1]{\fnm{Sebastian} \sur{Biewer}}\email{biewer@depend.uni-saarland.de}
\author*[1,2,3]{\fnm{Kevin} \sur{Baum}}\email{kevin.baum@uni-saarland.de}
\author*[1]{\fnm{Sarah} \sur{Sterz}}\email{sterz@depend.uni-saarland.de}
\author[1]{\fnm{Holger} \sur{Hermanns}}\email{hermanns@cs.uni-saarland.de}


\author[4]{\fnm{Sven} \sur{Hetmank}}\email{sven.hetmank@tu-dresden.de}

\author[5]{\fnm{Markus} \sur{Langer}}\email{markus.langer@uni-marburg.de}

\author[4]{\fnm{Anne} \sur{Lauber-Rönsberg}}\email{anne.lauber@tu-dresden.de}
\author[4]{\fnm{Franz} \sur{Lehr}}\nomail%\email{franz.lehr@tu-dresden.de}

\affil[1]{\orgdiv{Department of Computer Science}, \orgname{Saarland University}, \orgaddress{\street{Saarland Informatics Campus}, \city{Saarbrücken}, \postcode{66123}, \country{Germany}}}

\affil[2]{\orgdiv{Institute of Philosophy}, \orgname{Saarland University}, \orgaddress{\street{Campus}, \city{Saarbrücken}, \postcode{66123}, \country{Germany}}}

\affil[3]{\orgdiv{Neuro-mechanistic Modeling}, \orgname{DFKI}, \orgaddress{\street{Stuhlsatzenhausweg 3}, \city{Saarbrücken}, \postcode{66123}, \country{Germany}}}

\affil[4]{\orgdiv{IRGET}, \orgname{TU Dresden}, \orgaddress{\street{Bergstraße 53}, \city{Dresden}, \postcode{01062}, \country{Germany}}}

\affil[5]{\orgdiv{Department of Psychology}, \orgname{University of Marburg}, \orgaddress{\street{Gutenbergstraße 18}, \city{Marburg}, \postcode{35039}, \country{Germany}}}


\abstract{
  This article introduces a framework that is meant to assist in mitigating societal risks that software can pose. Concretely, this encompasses facets of software doping as well as unfairness and discrimination in high-risk decision-making systems.  The term \emph{software doping} refers to software that contains surreptitiously added functionality that is against the interest of the user.
  A prominent example of software doping are the tampered emission cleaning systems that were found in millions of cars around the world when the diesel emissions scandal surfaced.

  The first part of this article combines the formal foundations of
  software doping analysis with established probabilistic
  falsification techniques to arrive at a black-box analysis technique
  for identifying undesired effects of software. We apply this
  technique to emission cleaning systems in diesel cars but also
  to high-risk systems that evaluate humans in a possibly
  unfair or discriminating way.
  We demonstrate how our approach can assist humans-in-the-loop to make better informed and more responsible decisions. This is to promote effective human oversight, which will be a central requirement enforced by the European Union's upcoming AI Act~\cite{eu-0106-2021,ai-act-amendments}.
  We complement our technical contribution with a juridically, philosophically, and psychologically informed perspective on the potential problems caused by such systems.
}


\newcommand{\limitsAndChallenges}{Limitations \& Challenges}



\maketitle

\section{Introduction}\label{sec:intro}

Software is the main driver of innovation of our times. Software-defined systems are permeating our  communication, perception, and storage technology as well as our personal interactions with technical systems at an unprecedented pace. \emph{``Software-defined everything''} is among the hottest buzzwords in IT today~\cite{wired,SDE}.

At the same time, we are doomed to trust these systems, despite being unable to inspect or look inside the software we are facing: 
The owners of the physical hull of  \enquote*{everything} are typically not the ones owning the software defining \enquote*{everything}, nor will they have the right to look at what and how \enquote*{everything} is defined.
This is because commercial software typically is protected by intellectual property rights of the software manufacturer. This prohibits any attempt to disassemble the software or to reconstruct its inner working, albeit it is the very software that is forecasted to be defining \enquote*{everything}. The use of machine-learnt software components amplifies the problem considerably by adding opacity of its own kind. Since commercial interests of the software manufacturers seldomly are aligned with the interest of end users, the promise of \enquote*{software-defined everything} might well become a dystopia from the perspective of individual digital sovereignty.  
In this article, we address two of the most pressing incarnations of problematic software behaviour. 

\paragraph{Diesel emissions scandal}
A massive example of software-defined collective damage is the diesel emissions scandal.
Over a period of more than 10 years, millions of diesel-powered cars have been equipped with illegal software that altogether polluted the environment for the sake of commercial advantages of the car manufacturers. At its core, this was made possible by the fact that only a single, precisely defined test setup was put in place for  checking conformance with exhaust emissions regulations. This made it a trivial software engineering task to identify the test particularities and to turn off emission cleaning outside these particular conditions. This is an archetypal instance of software doping. 

Software doping can be formally characterised as a violation of a \emph{cleanness} property of a program~\cite{DBLP:conf/isola/BartheDFH16,DBLP:conf/esop/DArgenioBBFH17}.
A detailled and comparative account of meaningful cleanness definitions related to software doping is avaialable~\cite[Chapter 3]{diss:Biewer}. 
One cleanness notion that has proven  suitable to detect diesel emissions doping is \emph{robust cleanness}~\cite{DBLP:journals/tomacs/BiewerDH21,diss:Biewer}.
It is based on the assumption that there is some well-defined and agreed standard input/output behaviour of the system which the definition extends to the vicinity around the inputs and outputs close to the standard behaviour. 
The precise specification of ``vicinity'' and of ``standard behaviour'' is assumed to be part of a \emph{contract} between software manufacturer and user.  
That contract entails the standard behaviour, distance functions for input and output values, and distance thresholds to define the input and output vicinity, respectively. 
With this, a system behaviour is considered clean, if its output (is or) stays in the output vicinity of the standard, unless the input (is or) moves outside the standard's input vicinity.

\begin{example}
  \label{ex:contract}
    Every car model that is to enter the market in the European Union (and other countries) must be compliant with local regulations. As part of this homologation process, 
    common to all of these regulations is the need for executing a test  under precisely defined lab conditions, carried out on a chassis dynamometer.
    In this, the car has to follow a speed profile, which is called \emph{test cycle} in regulations.
    At the time when the diesel scandal surfaced, the \emph{New European Driving Cycle} (\NEDC)~\cite{nedc} was the single test cycle used in the European Union.
    It has by now been replaced by the \emph{Worldwide harmonized Light vehicles Test Cycle} (\WLTC)~\cite{LEX:32017R1151} in many countries.
    We refer to previous work for more details~\cite{DBLP:journals/tomacs/BiewerDH21,C3OTHER:STTT-toappear,diss:Biewer}.
    From a perspective of fraud prevention, having only a single test cycle is a major weakness of the homologation procedure.
    \RobustCleannessNDet can overcome this problem.
    It admits the consideration of driving profiles that stay in the bounded vicinity of one of several standardised test cycle (i.e., \NEDC as well as \WLTC), while enforcing bounds on the deviations regarding exhaust emission.
\end{example}

\paragraph{Discrimination mitigation}
Another set of exemplary scenarios we consider in this article are \emph{high-risk} AI systems, systems empowered by AI technology whose functioning may introduce risks to health, safety, or fundamental rights of human individuals. 
The European Union is currently developing the \emph{AI Act}~\cite{eu-0106-2021,ai-act-amendments} that sets out to mitigate many of the risks that such systems pose. 
Application areas of concern include credit approval (\cite{6677}), decisions on visa applications (\cite{visa}), admissions to higher education (\cite{Waters_Miikkulainen_2014,texasuniversity}), screening of individuals in predictive policing (\cite{predpol1}), selection in HR (\cite{jobs, 10.5555/3002861, oraclehr}), juridicial decisions (as with COMPAS \cite{DBLP:journals/bigdata/Chouldechova17,dressel2018accuracy,compaspropubstory,compaspropubanalysis}),  tenant screening (\cite{tenantscreen}), and more.
In many of these areas, there are legitimate interests and valid reasons for using well-understood AI technology, although the risks associated with their use to date is manifold. 

It is widely recognised that discrimination by unfair classification and regression models is one particularly important risk. As a result, a colourful zoo of different operationalisations of unfairness has emerged \cite{wachter2020bias,10.1145/3494672}, which should be seen less as a set of competing approaches and more as mutually complementary  \cite{10.1145/3433949}.
At the same time, a consensus is emerging that human oversight is an important piece of the puzzle for mitigating and minimising societal risks of AI \cite{aihleg,10.3389/frai.2021.737072,aiunesco}.
Accordingly, that principle made it into recent drafts of legislation including the European  AI Act \cite{eu-0106-2021,ai-act-amendments} or certain US state laws \cite{washstatefacerecbill}. 

The generic approach we develop for software-doping analysis turns out to be powerful enough to provide automated assistance for human overseers of high-risk AI systems. Apart from spelling out the necessary refocusing  
we illustrate the challenge that our work helps to overcome by an exemplary, albeit hypothetical  admission system for higher education (inspired by \cite{Waters_Miikkulainen_2014,texasuniversity}).

\begin{example}% 
\label{ex:hrWomanIntro}
A large university assigns scores to applicants aiming to enter their computer science PhD program. The sores are computed using an automated, model-based procedure $\seqSys$ which is based on three data points: the position of the applicant's last graduate institution in an official, subject-specific ranking, the applicant's most recent grade point average (GPA), and their score in a subject-specific standardised test taken as part of the application procedure. The system then automatically computes a score for the candidate based on an estimation of how successful it expects them to be as students.
A dedicated university employee, \hrwoman is in charge of overseeing the individual outcomes of $\seqSys$ and is supposed to detect cases where the output of $\seqSys$ is or appears flawed. The university pays especial attention to fairness in the scoring procedure, so \hrwoman has to watch out to any signs of potential unfairness. \hrwoman is supposed to desk-reject candidates whose scores are below a certain, predefined threshold -- unless she finds problems with $\seqSys$'s scoring. 
Without any additional support,
\hrwoman, as human overseer in the loop, must manually check all cases for signs of unfairness as they are processed.
This can be a tedious, complicated, and error-prone task
and as such constitutes an impediment for the assumed scalability of the automated scoring process to high numbers of applicants. 
Therefore, she at least requires tool support that assists her in detecting when something is off about the scoring of individual applicants. 
 \end{example}

This support can be made real by exploiting the technical contributions of this article, in terms of a runtime monitor that provides automated assistance to the human oversight and itself is based on the probabilistic falsification technique we develop. As we will explain, func-cleanness, a variant of cleanness, is a suitable basis for rolling out runtime monitors for such high-risk systems, that are able to detect and flag discrimination or unfair treatment of humans.

\paragraph{} 
The contributions made by this article are threefold.
\begin{description}
  \item[Detecting software doping using probabilistic falsification.] The paper starts off by developing the theory of \robustCleannessNDet and \fCleannessNDet. 
  We provide characterisations in the temporal logics HyperSTL and STL, that are then used for an adaptation of existing probabilistic falsification techniques~\cite{DBLP:journals/tecs/AbbasFSIG13,DBLP:journals/tcs/FainekosP09}. Altogether, this reduces the problem of software doping detection to the problem of falsifying the logical characterisation of the respective cleanness definition. 
  \item[Falsification-based test input generation.]
  Recent work~\cite{DBLP:journals/tomacs/BiewerDH21} proposes a formal framework for \robustCleannessNDet testing, with the ambition of making it usable in practice, namely for emissions tests conducted with a real diesel car on a chassis dynamometer.
  However, that approach leaves open how to perform test input selection in a meaningful manner.
  The probabilistic falsification technique presented in this article attacks this shortcoming. It supports the testing procedure by guiding it towards test inputs that make the \robustCleannessNDet tests likely to fail.
  \item[Promoting effective human oversight.]
   We discuss and demonstrate how the technical contributions of this paper contribute to  
  effective human oversight of high-risk systems, as required by the current proposal of the AI act.
  The hypothetical university admission scenario introduced above will serve as a demonstrator for shedding light on the applicability of our approach as well as the the principles behind it. On a technical level, we provide a runtime monitor for individual fairness based on probabilistic falsification of \fCleannessNDet. On a conceptual level, we consider it important to clarify which duties come with the usage of such a system; from a \textit{legal} perspective, particularly considering the AI Act, substantiated by considering the \textit{ethical} dimension from a philosophical perspective, and from a \textit{psychological} perspective, particularly deliberating on how the overseeing can become \emph{effective}.
\end{description}

This paper is based on a conference publication~\cite{DBLP:conf/fase/BiewerH22}. Relative to that paper, the development of the theory here is more complete
and now includes temporal logic characterisations for \fCleannessNDet. 
On the conceptual side, this article adds a principled analysis of the applicability of \fCleannessNDet to effective human oversight, spelled out in the setting of admission to higher education.
We live up to the societal complexity of this new example and provide an interdisciplinary situation analysis and an interdisciplinary assessment of our proposed solution.
Accordingly, although the technical realisation is based on the probabilistic falsification approach outlined in this article, our solution is substantially more thoughtful than a naive instantiation of the falsification framework.

This article is structured as follows.
\Cref{sec:background} provides the preliminaries for the contributions in this article.
\Cref{sec:logicalCharacterisation} develops the theoretical foundations necessary to use the concept of probabilistic falsification with \robustCleannessNDet and \fCleannessNDet.
\Cref{sec:dieselSupervision} demonstrates how the probabilistic falsification approach can be combined with the previously proposed testing approach~\cite{DBLP:journals/tomacs/BiewerDH21} for \robustCleannessNDet, with a focus on tampered emission cleaning systems of diesel cars.
\Cref{sec:fairness} develops the technical realisation of a fairness monitor based on \fCleannessNDet for high-risk systems.
\Cref{sec:assessment} evaluates the fairness monitor from the perspective of the disciplines philosophy, psychology, and law.
Finally, \Cref{sec:conclusion} summarises the contributions of this article and discusses limitations of our approaches. 
The appendix of this article contains additional technical details, proofs, and further philosophical and juridical explanations.

\section{Background}\label{sec:background}

\subsection{Software Doping}\label{sec:background:doping}
After early informal characterisations of \emph{software doping}~\cite{Baum16:isola,DBLP:conf/isola/BartheDFH16},
D'Argenio et al.~\cite{DBLP:conf/esop/DArgenioBBFH17} propose a collection of formal definitions that specify when a software is \emph{clean}.
The authors call a software \emph{doped} (w.r.t. a cleanness definition) whenever it does not satisfy such cleanness definition.
We focus on \emph{\robustCleannessNDet} and \emph{\fCleannessNDet} in this article~\cite{DBLP:conf/esop/DArgenioBBFH17}.

We define by $\RRpos \coloneqq \set{x\in\RR \mid x \geq 0}$ the set of non-negative real numbers, by $\RRinf \coloneqq \RR \cup \set{-\infty, \infty}$ the set of \emph{extended reals}~\cite{rockafellar2009variational}, and by $\RRposinf \coloneqq \RRpos \cup \set{\infty}$ the set of the non-negative extended real numbers.
We say that a
function $d: X \times X \to \RRposinf$ is a \emph{distance function} if and only if it satisfies $d(x,x)=0$ and $d(x,y)=d(y,x)$ %and $d(x,y)\leq d(x,z)+d(z,y)$ 
for all $x$, $y \in X$. 
We let $\sigma[k]$ denote the $k$th literal of the finite or infinite word $\sigma$. 

\paragraph{Reactive Execution Model}
We can view a nondeterministic reactive program as a function $\reactSys:\Inputs^\omega\to 2^{(\Outputs^\omega)}$ perpetually mapping  inputs $\Inputs$ to sets of outputs $\Outputs$~\cite{DBLP:conf/esop/DArgenioBBFH17}. 
To formally model contracts that specify the concrete configuration of \robustCleannessNDet or \fCleannessNDet, we denote by $\StdIn \subseteq \Inputs^\omega$ the input space of the system designated to define the standard behaviour, and by $d_\Inputs:(\Inputs\times\Inputs)\to\RRposinf$ and  $d_\Outputs:(\Outputs\times\Outputs)\to\RRposinf$ distance functions on 
inputs, respectively outputs.

For \robustCleannessNDet, we additionally consider two constants $\inpbound, \outpbound \in \RRposinf$.
$\inpbound$ defines the maximum distance that a non-standard input must have to a standard input to be considered in the cleanness evaluation.
For this evaluation, $\outpbound$ defines the maximum distance between two outputs such that they are still considered sufficiently close. 
Intuitively, the contract defines tubes around standard inputs and there outputs.
For example, in \Cref{fig:background:robCleannessIntuition}, $\inp$ is a standard input and
$\dIn$ and $\inpbound$ implicitly define a $2\inpbound$ wide tube around $\inp$.
Every input $\inp'$ that is within this tube will be evaluated on its outputs.
Similarly, $\dOut$ and $\outpbound$ define a tube around each of the outputs of $\inp$.
An output for $\inp'$ that is within this tube satisfies the robust cleanness condition.
Together, the above objects constitute a formal contract $\Contract = \langle \StdIn, d_\Inputs, d_\Outputs, \inpbound, \outpbound \rangle$.
\RobustCleannessNDet is composed of two separate definitions called \robustLowCleannessNDet and \robustUpCleannessNDet.
Assuming a fixed standard behaviour of a system, \robustLowCleannessNDet imposes a lower bound on the non-standard outputs that a system must exhibit, while \robustUpCleannessNDet imposes an upper bound.
Such lower and upper bound considerations are necessary because of the potential nondeterministic behaviour of the system; for deterministic systems the two notions coincide.
We remark that in this article we are using past-forgetful distance functions and the \traceIntegral variants of \robustCleannessNDet and \fCleannessNDet (see Biewer~\cite[Chapter 3]{diss:Biewer} for details).

\begin{definition}\label{def:robclean:react}
  A nondeterministic reactive program $\reactSys:\Inputs^\omega\to 2^{(\Outputs^\omega)}$ is \emph{ \robustlyCleanNDet} w.r.t.\ contract $\Contract = \langle \StdIn, d_\Inputs, d_\Outputs, \inpbound, \outpbound \rangle$ if for every standard input
  $\inp\in\Norm$ and input sequence $\inp'\in\Inputs^\omega$ it is the case that 
  \begin{enumerate}
    \item for every $\outp\in \reactSys(\inp)$, there exists $\outp'\in \reactSys(\inp')$, such
    that for every index $k\in\NN$, if
    $d_\Inputs(\inp[j],\inp'[j])\leq\inpbound$ for all $j\leq k$,
    then it holds that $d_\Outputs(\outp[k],\outp'[k]) \leq \outpbound$,\\
    \phantom{X}\hfill \emph{(\robustLowCleannessNDet)}
    \label{def:robclean:react:L}
    \item for every $\outp'\in \reactSys(\inp')$, there exists $\outp\in \reactSys(\inp)$, such
    that for every index $k\in\NN$, if
    $d_\Inputs(\inp[j],\inp'[j])\leq\inpbound$ for all $j\leq k$,
    then it holds that $d_\Outputs(\outp[k],\outp'[k]) \leq \outpbound$.\\
    \phantom{X}\hfill \emph{(\robustUpCleannessNDet)}
    \label{def:robclean:react:U}
  \end{enumerate}

\end{definition}

We will in the following refer to \Cref{def:robclean:react}.\ref{def:robclean:react:L} for \robustLowCleannessNDet and \Cref{def:robclean:react}.\ref{def:robclean:react:U} for \robustUpCleannessNDet.
Intuitively, \robustLowCleannessNDet enforces that whenever an input $\inp'$ remains within $\inpbound$ vicinity around the standard input $\inp$, then for every standard output $\outp \in \reactSys(\inp)$, there must be a non-standard output $\outp' \in \reactSys(\inp')$ that is in $\outpbound$ proximity of $\outp$.
Referring to \Cref{fig:background:robCleannessIntuition}, every $\inp'$ in the tube around $\inp$ must produce for every standard output $\outp \in \reactSys(\inp)$ at least one output $\outp' \in \reactSys(\inp')$ that resides in the $\outpbound$-tube around $\outp$.
In other words, for non-standard inputs the system must not lose behaviour that it can exhibit for a standard input in $\inpbound$ proximity. 

For \robustUpCleannessNDet the standard and non-standard output switch roles.
It enforces that whenever an input $\inp'$ remains within $\inpbound$ vicinity around the standard input $\inp$, then for every output $\outp' \in \reactSys(\inp')$ the system can exhibit for this non-standard input, there must be a standard output $\outp \in \reactSys(\inp)$ that is in $\outpbound$ proximity of $\outp'$.
Referring to \Cref{fig:background:robCleannessIntuition}, every $\inp'$ in the tube around $\inp$ must only produce outputs $\outp' \in \reactSys(\inp')$ that are in the $\outpbound$-tube of at least one $\outp \in \reactSys(\inp)$.
In other words, for non-standard inputs within $\inpbound$ proximity of a standard input the system must not introduce new behaviour, i.e., it must not exhibit an output that is further than $\outpbound$ away from the set of standard outputs.

\begin{figure}
  \includegraphics[width=\textwidth]{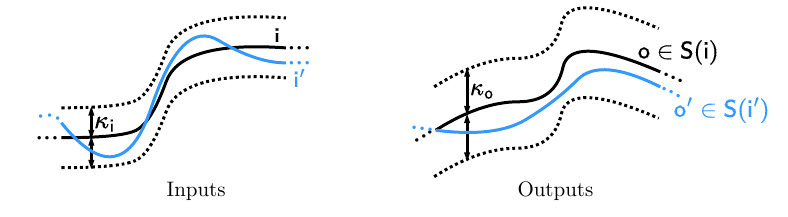}
  \caption{Robust Cleanness Intuition}
  \label{fig:background:robCleannessIntuition}
\end{figure}

A generalisation of \robustCleannessNDet is \fCleannessNDet.
A cleanness contract for \fCleannessNDetPar replaces the constants $\inpbound$ and $\outpbound$ by a function $f: \RRposinf \to \RRposinf$ inducing a dynamic threshold for output distances based on the distance between the inputs producing such outputs.

\begin{definition}\label{def:clean:react:f}
  A nondeterministic reactive system $\reactSys$ is \emph{ \fCleanNDet} \cleanFor contract $\Contract = \langle \Norm, \dIn, \dOut, f \rangle$ if for every standard input
  $\inp\in\Norm$ and input sequence $\inp'\in\Inputs^\omega$ it is the case that
  \begin{enumerate}
    \item for every $\outp\in \reactSys(\inp)$, there exists $\outp'\in
    \reactSys(\inp')$, such that for every index $k\in\NN$, $\dOut(\outp[k],\outp'[k]) \leq
    f(\dIn(\inp[k],\inp'[k]))$,
    \hfill \emph{(\fLowCleannessNDet)}
    \label{def:clean:react:f:L}
    \item for every $\outp'\in \reactSys(\inp')$, there exists $\outp\in
    \reactSys(\inp)$, such that for every index $k \in \NN$, $\dOut(\outp[k],\outp'[k]) \leq
    f(\dIn(\inp[k],\inp'[k]))$.
    \hfill \emph{(\fUpCleannessNDet)}
    \label{def:clean:react:f:U}
  \end{enumerate}
\end{definition}

We will in the following refer to \Cref{def:clean:react:f}.\ref{def:clean:react:f:L} for \fLowCleannessNDet and \Cref{def:clean:react:f}.\ref{def:clean:react:f:U} for \fUpCleannessNDet.

For the fairness monitor in \Cref{sec:fairness} we will use a simpler variant of \fCleannessNDet for deterministic sequential programs.
Since \seqSys is deterministic, the lower and upper bound requirements coincide, yielding the following simplified definition.

\begin{definition}\label{def:clean:seq:f:det}
    A deterministic sequential program $\seqSys$ is \emph{\fCleanNDet} \cleanFor contract $\Contract = \langle \Norm, \dIn, \dOut, f \rangle$ if for every standard input $\inp \in \Norm$ and input $\inp' \in \Inputs$, it holds that 
    $\dOut(\seqSys(\inp),\seqSys(\inp')) \leq
    f(\dIn(\inp,\inp'))$. 
\end{definition}

\paragraph{Mixed-IO System Model}
The reactive execution model above has the strict requirement that for every input, the system produces exactly one output.
Recent work~\cite{DBLP:conf/qest/BiewerDH19,DBLP:journals/tomacs/BiewerDH21} instead considers mixed-IO models, where a program $\tracelang \subseteq (\Inputs \cup \Outputs)^\omega$ is a subset of traces containing both inputs and outputs, but without any restriction on the order or frequency in which inputs and outputs appear in the trace. In particular, they are not required to strictly alternate (but they may, and in this way the reactive execution model can be considered a special case~\cite{diss:Biewer}). 
A particularity of this model is the distinct output symbol $\quiescence$ for quiescence, i.e., the absence of an output.
For example, finite behaviour can be expressed by adding infinitely many $\quiescence$ symbols to a finite trace.

The new system model induces consequences regarding cleanness contracts.
Every mixed-IO trace is projected into an input, respectively output domain.
The set of input symbols contains one additional element $\NoInp$, that indicates that in the respective steps an output was produced, but masking the concrete output.
Similarly, the set of output symbols contains the additional element $\NoOutp$ to mask a concrete input symbol.
\emph{Projection on inputs} $\mapInp{} \colon$ $(\Inputs \cup \Outputs)^\omega \rightarrow
(\Inputs \cup \{\NoInp\})^\omega$ and \emph{projection on outputs}
$\mapOut{} \colon (\Inputs \cup \Outputs)^\omega \rightarrow (\Outputs \cup \{\NoOutp\})^\omega$
are defined for all traces $\sigma \in (\Inputs \cup \Outputs)^\omega$ and $k\in\N$ as follows:
 $\mapInp{\sigma} [k] \coloneqq
  \textbf{ if } \sigma[k] \in \Inputs \textbf{ then } \sigma[k]
  \textbf{ else } \NoInp$ and similarly
 $\mapOut{\sigma} [k] \coloneqq
  \textbf{ if } \sigma[k] \in \Outputs \textbf{ then } \sigma[k]
  \textbf{ else } \NoOutp$.
The distance functions $\dInMIO$ and $\dOutMIO$ apply on input and output symbols or their respective masks, i.e., they are functions
$(\Inputs\cup\{\NoInp\}) \times (\Inputs\cup\{\NoInp\}) \rightarrow \RRposinf$
and, respectively,
$(\Outputs\cup\{\NoOutp\}) \times (\Outputs\cup\{\NoOutp\}) \rightarrow \RRposinf$.
Finally, instead of a set of standard inputs $\Norm$, we evaluate mixed-IO system cleanness w.r.t. to a set of standard behaviour $\Std \subseteq \tracelang$.
Thus, not only inputs, but also outputs can be defined as standard behaviour and for an input, one of its outputs can be considered in $\Std$ while a different output can be excluded from $\Std$.
As a consequence, the set $\Std$ is specific for some mixed-IO system $\tracelang$, because $\Std$ is useful only if $\Std \subseteq \tracelang$.
To emphasise this difference we will call the tuple $\Contract = \langle \StdSet, {d}_\Inputs, {d}_\Outputs, \inpbound, \outpbound \rangle$ (cleanness) \emph{context} (instead of cleanness contract).
\RobustCleannessNDet of mixed-IO systems w.r.t.\ such a context is defined below~\cite{DBLP:journals/tomacs/BiewerDH21}.

\begin{definition}\label{def:robclean:LTS}
  A mixed-IO system $\tracelang \subseteq (\Inputs \cup \Outputs)^\omega$ is \emph{ \robustlyCleanNDet} w.r.t.\ context $\Contract = \langle \StdSet, {d}_\Inputs, {d}_\Outputs, \inpbound, \outpbound \rangle$ if and only if $\StdSet \subseteq \tracelang$ and for all
     $\sigma\in\StdSet$ and $\sigma'\in\tracelang$, 
     \begin{enumerate}
      \item there exists $\sigma''\in\tracelang$ with
      $\mapInp{\sigma'} = \mapInp{\sigma''}$, such that
     for every index ${k\in \NN}$ it holds that whenever 
     ${d}_\Inputs(\mapInp{\sigma[j]},\mapInp{\sigma'[j]})\leq\inpbound$
     for all $j\leq k$, then
       ${d}_\Outputs(\mapOut{\sigma[k]},\mapOut{\sigma''[k]})\leq\outpbound$,\hfill \emph{(\robustLowCleannessNDet)}
       \label{def:robclean:LTS:L}
      \item there exists $\sigma''\in\StdSet$ with
      $\mapInp{\sigma}=\mapInp{\sigma''}$, such that for every
     index $k\in\NN$ it holds that whenever 
     ${d}_\Inputs(\mapInp{\sigma[j]},\mapInp{\sigma'[j]})\leq\inpbound$
     for all $j\leq k$, then
     ${d}_\Outputs(\mapOut{\sigma'[k]},\mapOut{\sigma''[k]})\leq\outpbound$.\hfill \emph{(\robustUpCleannessNDet)}
     \label{def:robclean:LTS:U}
     \end{enumerate}
\end{definition}

We will in the following refer to \Cref{def:robclean:LTS}.\ref{def:robclean:LTS:L} for \robustLowCleannessNDet and \Cref{def:robclean:LTS}.\ref{def:robclean:LTS:U} for \robustUpCleannessNDet.
\Cref{def:robclean:LTS} universally quantifies a standard trace $\sigma$.
For \robustLowCleannessNDet, the universal quantification of $\sigma'$ effectively only quantifies an input sequence; the input projection for the existentially quantified $\sigma''$ must match the projection for $\sigma'$.
The remaining parts of the definition are conceptually identical to their reactive systems counterpart in \Cref{def:robclean:react}.\ref{def:robclean:react:L}.
For \robustUpCleannessNDet, 
the existentially quantified trace $\sigma''$ is obtained from set $\Std$ in contrast to \robustLowCleannessNDet, where $\sigma''$ can be any arbitrary trace of \tracelang.
This is necessary, because \robustUpCleannessNDet is defined w.r.t. a cleanness context; from knowing that $\sigma\in\Std$ is a standard trace and by enforcing that $\mapInp{\sigma} = \mapInp{\sigma''}$ we cannot conclude that also $\sigma''\in\Std$.

\Cref{def:fclean:LTS} shows the definition \fCleannessNDet of mixed-IO systems.

\begin{definition}\label{def:fclean:LTS}
  A mixed-IO system $\tracelang \subseteq (\Inputs \cup \Outputs)^\omega$ is \emph{\fCleanNDet} w.r.t.\ context $\Contract = \langle \StdSet, {d}_\Inputs, {d}_\Outputs, f \rangle$ if and only if $\StdSet \subseteq \tracelang$ and for all $\sigma\in\StdSet$ and $\sigma'\in\tracelang$, 
  \begin{enumerate}
    \item there exists $\sigma''\in\tracelang$ with
    $\mapInp{\sigma'} = \mapInp{\sigma''}$, such that 
    for every index $k\in\NN$, it holds that
    ${d}_\Outputs(\mapOut{\sigma[k]},\mapOut{\sigma''[k]})\leq f({d}_\Inputs(\mapInp{\sigma[k]},\mapInp{\sigma'[k]}))$,
    \hfill \emph{(\fLowCleannessNDet)}
    \label{def:fclean:LTS:L}
    \item there exists $\sigma''\in\StdSet$ with
    $\mapInp{\sigma}=\mapInp{\sigma''}$, such that for every
    index $k\in\NN$, it holds that
    ${d}_\Outputs(\mapOut{\sigma'[k]},\mapOut{\sigma''[k]})\leq f({d}_\Inputs(\mapInp{\sigma[k]},\mapInp{\sigma'[k]}))$.
    \hfill \emph{(\fUpCleannessNDet)}
    \label{def:fclean:LTS:U}
  \end{enumerate}
      
  \end{definition}

  We will in the following refer to \Cref{def:fclean:LTS}.\ref{def:fclean:LTS:L} for \fLowCleannessNDet and \Cref{def:fclean:LTS}.\ref{def:fclean:LTS:U} for \fUpCleannessNDet.

\subsection{Temporal Logics}\label{sec:background:tl}
\subsubsection{HyperLTL}\label{sec:background:tl:hyperltl}
Linear Temporal Logic (LTL)~\cite{DBLP:conf/focs/Pnueli77} is a popular formalism to reason about properties of traces. 
A trace is an infinite word where each literal is a subset of $\ap$, the set of atomic propositions. 
We interpret programs as circuits encoded as sets $\apSys \subseteq(2^{\ap})^\omega$ of such traces. 
LTL provides expressive means to characterise sets of traces, often  called \emph{trace properties}. 
% \mysubsubsec{Temporal Logics for Hyperproperties}
For some set of traces $T$, a trace property defines a subset of $T$ (for which the property holds), whereas a \emph{hyperproperty} defines a \emph{set of} subsets of $T$ (constituting combinations of traces for which the property holds). In this way it  specifies which traces are valid in combination with one another. Many temporal logics have been extended to corresponding hyperlogics supporting the specification of hyperproperties.  

HyperLTL~\cite{ClarksonFKMRS14:post} is such a temporal logic for the specification of hyperproperties of reactive systems.  
It extends LTL with trace quantifiers and trace variables that make it possible to refer to multiple traces within a logical formula.
  A \emph{HyperLTL formula} is defined by the following
grammar, where $\pi$ is drawn from a set $\mathcal{V}$ of \emph{trace
variables} and $a$ from the set $\ap$:
\begin{equation*}
  \begin{array}{c!{\ {::{=}}\ }c!{\ \mid\ }c!{\ \mid\ }c!{\ \mid\ }c!{\ \mid\ }c}
    \psi & \EEE{\pi} \psi & \AAA{\pi} \psi & \multicolumn{3}{l}{\ \ \phi } \\
    \phi & a_\pi & \neg\phi & \phi\land\phi & \ \X\phi \ \, & \phi\U\phi \\
  \end{array}\label{eq:hyperltl:syntax}
\end{equation*}
The quantifiers $\exists$ and $\forall$ quantify existentially and universally, respectively, over the set of 
traces.  For example, the formula
$\AAA{\pi}\EEE{\pi'}\phi$ means that for every trace $\pi$ there exists
another trace $\pi'$ such that $\phi$ holds over the pair of traces.
To account for distinct valuations of atomic propositions across distinct traces, the atomic propositions are indexed with trace variables:  for some atomic proposition $a\in\ap$ and some trace variable $\pi\in \mathcal{V}$, $a_\pi$ states that $a$ holds in the initial position of trace $\pi$.  
The temporal operators and Boolean connectives are interpreted as usual for LTL.
Further operators are derivable:
${\F\phi}\equiv{\true\U\phi}$ enforces $\phi$ to eventually hold in the future,
${\G\phi}\equiv{\neg\F\neg\phi}$ enforces $\phi$ to always hold, and
the weak-until operator  
${\phi\W\phi'}\equiv{{\phi\U\phi'}\lor{\G\phi}}$ 
allows $\phi$ to always hold as an alternative to the obligation for $\phi'$ to eventually hold.

\paragraph{HyperLTL Characterisations of Cleanness}
D'Argenio et al.~\cite{DBLP:conf/esop/DArgenioBBFH17} assume distinct sets of atomic propositions to encode inputs and outputs.
That is, they assume that $\ap = {\ap_\inp\cup\ap_\outp}$ of atomic propositions, where
$\ap_\inp$ and $\ap_\outp$ are the atomic propositions
that define the the input values and, respectively, the output
values.
Thus, in the context of Boolean circuit encodings of programs, we take $\Inputs=2^{\ap_\inp}$ and
$\Outputs=2^{\ap_\outp}$.
We capture the following natural correspondence between reactive programs and Boolean circuits; 
a circuit $\apSys$ can be interpreted as a function $\boolSys:\Inputs^\omega\to 2^{(\Outputs^\omega)}$, where
\begin{equation}
  t\in \apSys \quad \text{ if and only if } \quad
  (\mapTrace{\tr}{\ap_\outp})\ \in\ \boolSys(\mapTrace{\tr}{\ap_\inp}),
  \label{eq:prog:tr-to-fn}
\end{equation}
with $\mapTrace{\tr}{A}$ defined by $(\mapTrace{\tr}{A})[k]=\tr[k]\cap A$ for all
$k\in\NN$.

In the HyperLTL formulas below occur, for convenience, non-atomic propositions.
Their semantics is encoded by atomic propositions and Boolean connectives according to a Boolean encoding of inputs and outputs.
We refer to the original work for the details~\cite[Table 1]{DBLP:conf/esop/DArgenioBBFH17}. 
Further, we assume that there is 
a quantifier-free HyperLTL formula $\Norm_\pi$ that can check whether the trace represented by trace variable $\pi$ is in the set of standard inputs $\Norm \subseteq \Inputs^\omega$.
That is, $\Norm_\pi$ should be defined such that for every trace $\varTrace\in\apSys$ it holds that $\set{\pi \coloneqq \varTrace} \models_\apSys \Norm_\pi$ if and only if $(\mapAPi{\varTrace}) \in \Norm$.

\Cref{prop:mc:hyperltl:robClean} shows HyperLTL formulas for \robustLowCleannessNDetPar and \robustUpCleannessNDetPar, respectively.\footnote{All HyperLTL formulas from D'Argenio et al.~\cite{DBLP:conf/esop/DArgenioBBFH17} are adapted for non-parametrised systems.} 

% ROBUST CLEANNESS (LTL)
\begin{proposition}\label{prop:mc:hyperltl:robClean}%
  Let $\apSys$ be a set of infinite traces over $\powerSet{\ap}$, let $\boolSys$ be the reactive system constructed from $\apSys$ according to \Cref{eq:prog:tr-to-fn}, and let $\Contract = \langle \Norm, \dIn, \dOut, \inpbound, \outpbound \rangle$ be a contract for \robustCleannessNDetPar. 
  Then $\boolSys$ is \robustlyLowCleanNDetPar w.r.t. \Contract 
  if and only if $\apSys$ satisfies the
  HyperLTL formula
  \begin{align*}\small
    & \AAA{\pi_1}\AAA{\pi_2}\EEE{\pi'_2}
    {\Norm_{\pi_1}} \\
    & \qquad\qquad\qquad\qquad \limp \Big( {\G({\inp_{\pi_2}=\inp_{\pi'_2}})}\land 
    \big((\dOutTL(\outp_{\pi_1},\outp_{\pi'_2})\leq\outpbound)\W(\dInTL(\inp_{\pi_1},\inp_{\pi'_2})>\inpbound) \big)\Big), \notag
  \end{align*}
  and $\boolSys$ is \robustlyUpCleanNDetPar w.r.t. \Contract 
  if and only if $\apSys$ satisfies the
  HyperLTL formula
  {\begin{align*}
    &  \AAA{\pi_1}\AAA{\pi_2}\EEE{\pi'_1} 
     {\Norm_{\pi_1}}  \\
    & \qquad\qquad\qquad\qquad \limp \Big({\G({\inp_{\pi_1}=\inp_{\pi'_1}})}\land
    \big((\dOutTL(\outp_{\pi'_1},\outp_{\pi_2})\leq\outpbound)\W(\dInTL(\inp_{\pi'_1},\inp_{\pi_2})>\inpbound)\big)\Big).
  \end{align*}}
\end{proposition}

The first quantifier (for $\pi_1$) in both formulas implicitly quantifies the standard input $\inp$ and 
the second quantifier (for $\pi_2$) implicitly quantifies the second input $\inp'$. 
Due to the potential nondeterminism in the behaviour of the system, the third, existential, quantifier for $\pi_1'$, respectively $\pi_2'$ is necessary.
While the formula for \robustLowCleannessNDetPar has the universal
quantification on the outputs of the program that takes the standard input $\inp$
and the existential quantification on the output for $\inp'$,
the formula for \robustUpCleannessNDetPar works in the other way around.   
Thus, the formulas capture the $\forall\exists$ alternation in  \Cref{def:robclean:react}.
The weak until operator $\W$ has exactly the behaviour necessary to represent the interaction between the distances of inputs and
the distances of outputs.  

The HyperLTL formulas for \fCleannessNDet are given below.
\begin{proposition}\label{prop:mc:hyperltl:f-clean}%
    Let $\apSys$ be a set of infinite traces over $\powerSet{\ap}$, let $\boolSys$ be the reactive system constructed from $\apSys$ according to \Cref{eq:prog:tr-to-fn}, and let $\Contract = \langle  \Norm, \dIn, \dOut, f \rangle$ be a contract for \fCleannessNDetPar.
    Then $\boolSys$ is \fLowCleanNDetPar w.r.t. \Contract 
    if and only if $\apSys$ satisfies the
    HyperLTL formula
    \begin{align*}
      & \AAA{\pi_1}\AAA{\pi_2}\EEE{\pi'_2} {\Norm_{\pi_1}} \limp \left({\G({\inp_{\pi_2}=\inp_{\pi'_2}})}\land{\G\left({\dOutTL(\outp_{\pi_1},\outp_{\pi'_2})}\leq{f(\dInTL(\inp_{\pi_1},\inp_{\pi'_2}))}\right)}\right), 
    \end{align*}
    and 
    $\boolSys$ is \fUpCleanNDetPar w.r.t. \Contract 
    if and only if $\apSys$ satisfies the
    HyperLTL formula
    \begin{align*}
      & \AAA{\pi_1}\AAA{\pi_2}\EEE{\pi'_1} {\Norm_{\pi_1}}  \limp \left({\G({\inp_{\pi_1}=\inp_{\pi'_1}})}\land{\G\left({\dOutTL(\outp_{\pi'_1},\outp_{\pi_2})}\leq{f(\dInTL(\inp_{\pi'_1},\inp_{\pi_2}))}\right)}\right).
    \end{align*}
  \end{proposition}

\subsubsection{Signal Temporal Logic}
LTL enables reasoning over traces $\sigma \in {(2^\ap)}^\omega$ 
for which it is necessary to encode values using the atomic propositions in $\ap$.
Each literal in a trace represents a discrete time step of an underlying model.
Thus, $\sigma$ can equivalently be viewed as a function $\N \to 2^\ap$.
One extension of LTL is \emph{Signal Temporal Logic} (STL)~\cite{DBLP:conf/cav/DonzeFM13,DBLP:conf/formats/MalerN04}, which instead is used for reasoning over real-valued signals that may change in value along an underlying continuous time domain.
In this article, we generalise the original work and use \emph{generalised timed traces} (GTTs)~\cite{Gazda2019}, which, for some value domain $X$ and time domain $\mathcal{T}$ define traces as functions $\mathcal{T} \rightarrow X$.
The time domain $\mathcal{T}$ can be either  $\NN$ (\emph{discrete-time}),  or $\RRpos$ (\emph{continuous-time}).
For the value domain we will use vectors of real values $X = \RR^n$ for some $n > 0$ or, to express mixed-IO traces, the set $X = \Inputs \cup \Outputs$.

STL formulas can express properties 
of systems modelled as sets $\SysSTL \subseteq (\mathcal{T}\rightarrow X)$ of traces 
by making the atomic properties refer to booleanisations of the signal values.  The syntax of the variant of STL that we use in this article is as follows, where $f \in X \rightarrow \Reals$: 
\begin{equation*}
  \begin{array}{c!{\ {::{=}}\ }c!{\ \mid\ }c!{\ \mid\ }c!{\ \mid\ }c!{\ \mid\ }c}
    \phi & \top & f > 0 & \neg\phi & \phi\land\phi & \phi\U\phi \\
  \end{array}.\label{eq:stl:syntax}
\end{equation*}

STL replaces atomic propositions by \emph{threshold predicates} of the form $f > 0$, which hold if and only if function $f$ applied to the trace value at the current
 time returns a positive value.
The Boolean operators and the Until operator $\U$ are very similar to those of HyperLTL.
The Next operator $\X$ is not part of STL, because ``next'' is without precise meaning in  continuous time.
The definitions of the derived operators $\F$, $\G$ and $\W$ are the same as for HyperLTL.
Formally, the \emph{Boolean semantics} of an STL formula $\phi$ at time $t \in \mathcal{T}$ for a trace 
$\w \in \mathcal{T} \to X$
is defined inductively: 
\begin{align*}
  & w, t \models \top  && \\
  & w, t \models f > 0 && \quad\text{iff}\quad f(\w(t)) > 0 \\
  & w, t \models \neg\phi && \quad\text{iff}\quad w, t \not\models \phi \\
  & w, t \models \phi \land \psi && \quad\text{iff}\quad w, t \models \phi \text{ and } w, t \models \psi \\
  & w, t \models \phi\U\psi && \quad\text{iff}\quad \text{exists }t' \geq t \text{ s.t. } w, t' \models \psi \text{ and} \\
  &&&\qquad\qquad\qquad\qquad\qquad\qquad\qquad\text{ for all } t'' \in [t, t'),\ w,t'' \models \phi
\end{align*}
A system \SysSTL satisfies a formula $\phi$, denoted $\SysSTL \models \phi$, if and only if for every $\w\in\SysSTL$ it holds that $\w,0 \models \phi$.

\paragraph{Quantitative Interpretation}
STL has been extended by a \emph{quantitative semantics}~\cite{DBLP:conf/cav/DonzeFM13,DBLP:journals/tecs/AbbasFSIG13,DBLP:journals/tcs/FainekosP09}. 
This semantics is designed in such a way that whenever $\rho(\phi, w, t) \neq 0$, its sign indicates whether $w, t \models \phi$ holds in the Boolean semantics.
For any STL formula $\phi$, trace $w$ and time $t$, if $\rho(\phi, w, t) > 0$, then $w, t \models \phi$ holds, and if $\rho(\phi, w, t) < 0$, then $w, t \models \phi$ does not hold.
The quantitative semantics for an STL formula $\phi$, trace $w$, and time $t$ the quantitative semantics is defined inductively:
 \begin{align*}
  \rho(\top, w, t) &\ =\ \infty \\
  \rho(f > 0, w, t) &\ =\ f(\w(t)) \\
  \rho(\neg\phi, w, t) &\ =\ -\rho(\phi, w, t) \\
  \rho(\phi \land \psi, w, t) &\ =\ \min(\rho(\phi, w, t), \rho(\psi, w, t)) \\
  \rho(\phi\U\psi, w, t) &\ =\ \sup_{t' \geq t} \min \set{\rho(\psi, w, t'), \inf_{t'' \in [t, t')} \rho(\phi, w, t'')}
\end{align*}

\paragraph{Robustness and Falsification}
The value of the quantitative semantics can serve as a \emph{robustness estimate} and as such be used to search for a violation of the property 
    \begin{algorithm}[t]
      \renewcommand{\algorithmicrequire}{\textbf{Input:}}
      \renewcommand{\algorithmicensure}{\textbf{Output:}}
      \renewcommand{\algorithmiccomment}[1]{/* #1 */}
    \begin{algorithmic}[1]
      \REQUIRE $w$: Initial trace, $\mathcal{R}$: Robustness function, $\PS$: Proposal Scheme
      \ENSURE $w\in\SysSTL$
      \WHILE{$\mathcal{R}(w) > 0$}
      \STATE $w' \leftarrow \PS(w)$
      \STATE $\alpha \leftarrow \exp(-\beta(\mathcal{R}(w') - \mathcal{R}(w)))$
      \STATE $r \leftarrow \mathsf{UniformRandomReal}(0,1)$
      \IF{$r \leq \alpha$}
      \STATE $w \leftarrow w'$
      \ENDIF
      \ENDWHILE
    \end{algorithmic}
    \caption{Monte-Carlo falsification}
    \label{algo:falsification}
    \end{algorithm}
at hand, i.e.,\ to falsify it.
The robustness of STL formula $\phi$ is its quantitative  value at time $0$, that is, $\mathcal{R}_\phi(w) \coloneqq \rho(\phi, w, 0)$. 
So,  falsifying a formula $\phi$ for a system $\SysSTL$ 
boils down to a search problem with the goal condition $\mathcal{R}_\phi(w) < 0$.
Successful falsification algorithms solve this  problem by understanding it as the
optimisation
 problem $\minimise_{w \in \SysSTL} \mathcal{R}_\phi(w)$.
Algorithm~\ref{algo:falsification}~\cite{DBLP:journals/tecs/AbbasFSIG13,DBLP:conf/hybrid/NghiemSFIGP10} sketches an algorithm for Monte-Carlo Markov Chain falsification, which is based on acceptance-rejection sampling~\cite{chib1995understanding}.

An input to the algorithm is an initial trace $w$ and a computable robustness function $\mathcal{R}$.
Robustness computation for STL formulas 
has been addressed in the literature~\cite{DBLP:conf/cav/DonzeFM13,DBLP:journals/tcs/FainekosP09}; we omit this discussion here.
The third input $\PS$ is a proposal scheme that proposes a new trace to the algorithm based on the previous one (line 2).
The parameter $\beta$ (used in line 3) can be adjusted during the search and is a means to avoid being trapped in local minima, preventing to find a global minimum.

Notably, there exists prior work by Nguyen et al.~\cite{DBLP:conf/memocode/NguyenKJDJ17} that discusses an extension of STL to HyperSTL though using a non-standard semantic underpinning.  In this context, 
they present a falsification approach restricted to the fragment ``t-HyperSTL'' where, according to the authors, ``a nesting structure of temporal logic formulas involving different traces is not allowed''.
Therefore, none of our cleanness definitions belongs to this fragment.

\section{Logical Characterisation of Mixed-IO Cleanness}\label{sec:logicalCharacterisation}

In this section we provide a temporal logic characterisation for \robustCleannessNDet and \fCleannessNDet for mixed-IO systems.
For this, we propose a HyperSTL semantics (different to that of \cite{DBLP:conf/memocode/NguyenKJDJ17}) and propose HyperSTL formulas for \robustCleannessNDet and \fCleannessNDet.
We explain how these formulas 
can be applied to mixed-IO traces and prove that the characterisation is correct.
Furthermore, for the special case that $\Std$ is a finite set, we reformulate the HyperSTL formulas characterising the u-cleannesses as equivalent STL formulas. 

\paragraph{Hyperlogics over Continuous Domains}
Previous work~\cite{DBLP:conf/memocode/NguyenKJDJ17} extends STL to HyperSTL echoing the extension of LTL to HyperLTL. 
We use a similar HyperSTL syntax in this article: 
\begin{equation*}
  \begin{array}{c!{\ {::{=}}\ }c!{\ \mid\ }c!{\ \mid\ }c!{\ \mid\ }c!{\ \mid\ }c}
    \psi & \E{\pi} \psi & \A{\pi} \psi & \multicolumn{2}{l}{\ \ \phi } \\
    \phi & \top & f > 0 & \neg\phi & \phi\land\phi & \ \phi\U\phi \ .\\
  \end{array}\label{eq:hyperstl:syntax}
\end{equation*}
The meaning of the universal and existential quantifier is as for HyperLTL.
In contrast to HyperLTL (and to the existing definition of HyperSTL), we consider it insufficient to allow propositions to refer to only a single trace.
In HyperLTL 
atomic propositions of individual traces can be compared by means of the Boolean connectives.
To formulate thresholds for real values, however, we feel the need to allow real values from multiple traces to be combined in the function $f$, and thus to appear as arguments of $f$.
Hence, in our semantics of HyperSTL, $f>0$ holds if and only if the result of $f$, applied to all traces quantified over, is greater than 0. For this to work formally, the arity of function $f$ is 
the number $m$ of traces quantified over at the occurrence of $f>0$ in the formula, so 
$f: X^m \rightarrow \Reals$.

A trace assignment~\cite{ClarksonFKMRS14:post} $\Pi: \mathcal{V} \rightarrow \SysSTL$ is a partial function assigning traces of $\SysSTL$ to variables.
Let $\Pi[\pi \coloneqq w]$ denote the same function as $\Pi$, except that $\pi$ is mapped to trace $w$. 
The Boolean semantics of HyperSTL is defined below. 
\begin{definition}\label{def:hyperstl:boolSem}
  Let $\psi$ be a HyperSTL formula, $t \in \mathcal{T}$ a time point, $\SysU \subseteq (\mathcal{T}\rightarrow X)$ a set of GTTs, and $\Pi$ a trace assignment.
  Then, the Boolean semantics for $\SysU, \Pi, \varTime \models \psi$ is defined inductively:  
  \begin{align*}
    \SysU, \Pi, \varTime &\models \exists \pi. \psi && \lIff\ \exists \w \in \SysU.\ \SysU, \Pi[\pi \coloneqq w], \varTime \models \psi\\
    \SysU, \Pi, \varTime &\models \forall \pi. \psi && \lIff\ \forall \w \in \SysU.\ \SysU, \Pi[\pi \coloneqq w], \varTime \models \psi\\
    \SysU, \Pi, \varTime &\models \top &&\\
    \SysU, \Pi, \varTime &\models f > 0 && \lIff\ f(\Pi(\pi_1)(t), \dots, \Pi(\pi_m)(t)) > 0 \text{ for } dom(\Pi) = \set{\pi_1, \dots, \pi_m}\footnotemark\setcounter{auxFootnote}{\value{footnote}}\\
    \SysU, \Pi, \varTime &\models \neg \phi && \lIff\ \SysU, \Pi, \varTime \not\models \phi\\
    \SysU, \Pi, \varTime &\models \phi_1 \land \phi_2 && \lIff\ \SysU, \Pi, \varTime \models \phi_1 \text{ and } \SysU, \Pi, \varTime \models \phi_2\\
    \SysU, \Pi, \varTime &\models \phi_1 \U \phi_2 && \lIff\ \exists \varTime' \geq \varTime.\  \SysU, \Pi, \varTime' \models \phi_2 \text{ and } \forall \varTime''\in[\varTime, \varTime').\ \SysU, \Pi, \varTime'' \models \phi_1
  \end{align*}
\end{definition}

A system $\SysU$ satisfies a formula $\psi$ if and only if $\SysU, \emptyset, 0 \models \psi$.
The quantitative semantics for HyperSTL
is defined below:
\begin{definition}\label{def:hyperstl:quantiSem}
  Let $\psi$ be a HyperSTL formula, $t \in \mathcal{T}$ a time point, $\SysU \subseteq (\mathcal{T}\rightarrow X)$ a set of GTTs, and $\Pi$ a trace assignment.
  Then, the quantitative semantics for $\rho(\psi, \SysU, \Pi, t)$ is defined inductively:
  \begin{align*}
    \rho(\exists\pi.\ \psi, \SysU, \Pi, t) &\ =\ \sup_{w \in \SysU} \rho(\psi, \SysU, \Pi[\pi \coloneqq w], t)\\
    \rho(\forall\pi.\ \psi, \SysU, \Pi, t) &\ =\ \inf_{w \in \SysU} \rho(\psi, \SysU, \Pi[\pi \coloneqq w], t)\\
    \rho(\top, \SysU, \Pi, t) &\ =\ \infty \\
    \rho(f > 0, \SysU, \Pi, t) &\ =\ f(\Pi(\pi_1)(t), \dots, \Pi(\pi_m)(t))
    \text{ for } dom(\Pi) = \set{\pi_1, \dots, \pi_m}\footnotemark[\value{footnote}]\\
    \rho(\neg\phi, \SysU, \Pi, t) &\ =\ -\rho(\phi, \SysU, \Pi, t) \\
    \rho(\phi_1 \land \phi_2, \SysU, \Pi, t) &\ =\ \min(\rho(\phi_1, \SysU, \Pi, t), \rho(\phi_2, \SysU, \Pi, t)) \\
    \rho(\phi_1\U\phi_2, \SysU, \Pi, t) &\ =\ \sup_{t' \geq t} \min \{\rho(\phi_2, \SysU, \Pi, t'), \inf_{t'' \in [t, t')} \rho(\phi_1, \SysU, \Pi, t'')\} 
  \end{align*}
\end{definition}
\footnotetext{We admit some sloppiness; the set $dom(\Pi)$ should have a fixed order.}

\paragraph{HyperSTL Characterisation}
The HyperLTL characterisations in \Cref{sec:background:tl:hyperltl} assume the system to be a subset of $(2^\ap)^\omega$ and works with distances between traces by means of a Boolean encoding into atomic propositions.
By using HyperSTL, we can characterise cleanness for systems that are representable as subsets of $(\mathcal{T} \rightarrow X)$. 

We can take the HyperLTL formulas from \Cref{prop:mc:hyperltl:robClean,prop:mc:hyperltl:f-clean} and transform them into HyperSTL formulas by applying simple syntactic changes.
We get for \robustLowCleannessNDet the formula
{\small\begin{align}
  &\psi_{\textsf{l-rob}} \coloneqq  \AAA{\pi_1}\AAA{\pi_2}\EEE{\pi'_2}
  {\Std_{\pi_1} > 0} \label{eq:hyperstl:lrob} \\
  & \qquad\qquad\ \  \limp \Big({\G(\mathsf{eq}(\mapInp{\pi_2}, \mapInp{\pi'_2}) \leq 0)}\land {} \notag \\
  & \qquad\qquad\qquad\qquad\qquad\ \ 
  \big((\dOut(\mapOut{\pi_1},\mapOut{\pi'_2})-\outpbound\leq 0)\W(\dIn(\mapInp{\pi_1},\mapInp{\pi'_2})-\inpbound > 0) \big)\Big),\notag
\end{align}}
\robustUpCleannessNDet is characterised by
{\small\begin{align}
  &\psi_{\textsf{u-rob}} \coloneqq \AAA{\pi_1}\AAA{\pi_2}\EEE{\pi'_1}
  {\Std_{\pi_1} > 0} \label{eq:hyperstl:urob} \\
  & \qquad\qquad\ \   \limp \Big({\Std_{\pi'_1} > 0} \land {\G(\mathsf{eq}(\mapInp{\pi_1}, \mapInp{\pi'_1}) \leq 0)}\land {}\notag \\
  & \qquad\qquad\qquad\qquad\qquad\ \ 
  \big((\dOut(\mapOut{\pi'_1},\mapOut{\pi_2})-\outpbound\leq 0)\W(\dIn(\mapInp{\pi'_1},\mapInp{\pi_2})-\inpbound > 0) \big)\Big),\notag
\end{align}}
for \fLowCleannessNDet we get the formula
{\small\begin{align}
  &\psi_{\textsf{l-fun}} \coloneqq  \AAA{\pi_1}\AAA{\pi_2}\EEE{\pi'_2}
  {\Std_{\pi_1} > 0} \label{eq:hyperstl:lfunc} \\
  & \qquad\quad    \limp \Big( {\G(\mathsf{eq}({\mapInp{\pi_2}, \mapInp{\pi'_2}}) \leq 0)}\land 
  \big(\G(\dOut(\mapOut{\pi_1},\mapOut{\pi'_2}) - f(\dIn(\mapInp{\pi_1},\mapInp{\pi'_2})) \leq 0) \big)\Big),\notag
\end{align}}
and, finally, \fUpCleannessNDet is encoded by
{\small\begin{align}
  &\psi_{\textsf{u-fun}} \coloneqq  \AAA{\pi_1}\AAA{\pi_2}\EEE{\pi'_1}
  {\Std_{\pi_1} > 0} \label{eq:hyperstl:ufunc} \\
  & \qquad\qquad\ \    \limp \Big({\Std_{\pi'_1} > 0} \land  {\G(\mathsf{eq}({\mapInp{\pi_1}, \mapInp{\pi'_1}}) \leq 0)}\land {} \notag \\
  & \qquad\qquad\qquad\qquad\qquad\qquad\qquad\quad\ 
  \big(\G(\dOut(\mapOut{\pi'_1},\mapOut{\pi_2}) - f(\dIn(\mapInp{\pi'_1},\mapInp{\pi_2})) \leq 0) \big)\Big).\notag
\end{align}}
The quantifiers remain unchanged relative to the formulas in \Cref{prop:mc:hyperltl:robClean,prop:mc:hyperltl:f-clean}.
The formulas use generic projection functions $\mapInp{}: X \to \Inputs$ and $\mapOut{}: X \to \Outputs$ to extract the input values, respectively output values from a trace.
To apply the formulas, these functions must be instantiated with functions for the concrete instantiation of the value domain $X$ of the traces to be analysed.
For example, for $\Inputs = \RR^m$, $\Outputs=\RR^l$, and $\SysSTL \subseteq (\timeSet \to \RR^{m+l})$, the projections could be defined for every $\w = (s_1, \dots, s_m, s_{m+1}, \dots, s_{m+l})$ as $\mapInp{\w} = (s_1, \dots, s_m)$ and $\mapOut{\w} = (s_{m+1}, \dots, s_{m+l})$.
The input equality requirement for two traces $\pi$ and $\pi'$ is ensured by globally enforcing
$\mathsf{eq}(\mapInp{\pi}, \mapInp{\pi'}) \leq 0$.
$\mathsf{eq}$ is a generic function that returns zero if its arguments are identical and a positive value otherwise.
It must be instantiated for concrete value domains.
For example, $\mathsf{eq}((s_1, \dots, s_m), (s'_1, \dots, s'_m))$ could be defined as the sum of the component-wise distances $\sum_{1 \leq i \leq m} \abs{s_i - s'_i}$.
Finally, in the above formulas we 
perform simple arithmetic operations to match the syntactic requirements of HyperSTL.

Formulas (\ref{eq:hyperstl:urob}) and (\ref{eq:hyperstl:ufunc}) are prepared to express \robustUpCleannessNDet, respectively \fUpCleannessNDet w.r.t. both cleanness \emph{contracts} or cleanness \emph{contexts}.
That is, we assume the existence of a function $\Std_\pi$ that returns a positive value if and only if the trace assigned to $\pi$ encodes a standard input (when considering cleanness \emph{contracts}) or encodes an input and output that constitute a standard behaviour (when considering cleanness \emph{contexts}).
Explicitly requiring that $\pi'_1$ represents a standard behaviour echoes the setup in Definitions~\ref{def:robclean:LTS}.\ref{def:robclean:LTS:U} and \ref{def:fclean:LTS}.\ref{def:fclean:LTS:U}. 

We remark that for encoding $\StdSet_{\pi}$, due to the absence of the Next-operator in HyperSTL, it might be necessary to add a clock signal $s(t) = t$ to traces in a preprocessing step. 

\begin{example}
  Let $\Inputs = \Outputs = \RR$ be the sets representing real-valued inputs and outputs, $\mathcal{T} = \NN$ be the discrete time domain, and $X = \Inputs \times \Outputs$ the value domain that considers pairs of inputs and outputs as values.
  We consider the \robustCleannessNDet context $\Contract = \langle \StdSet, \dInMIO, \dOutMIO, \inpbound, \outpbound \rangle$,
  where $\StdSet = \set{w_0, w_1}$ contains the two standard traces 
  %$\w_1 = (1,1)\,(3,3)\,(5,5)\,(7,7)\,\cdots$ and  $\w_2 = (1,0)\,(3,0)\,(5,0)\,(7,0)\,\cdots$.
  $\w_0 = (1;0)\,(2;0)\,(3;0)\,(4;0)\,\cdots$ and
  $\w_1 = (1;1)\,(2;2)\,(3;3)\,(4;4)\,\cdots$.
  % Intuitively, systems are supposed to either compute the identity function (which for every input responds with that input) or to always return zero.
  For the distance functions we use the absolute differences, i.e., $\dIn(v_1, v_2) = \dOut(v_1, v_2) = \abs{v_1 - v_2}$.
  Let the value thresholds be $\inpbound = 1$ and $\outpbound = 2$, and let $\mapInp{}, \mapOut{}, \mathsf{eq}$ and $\Std_\pi$ be defined as explained above.
  We consider the non-standard traces 
  $\w_A = (1.3;0)\,(2.6;0)\,(3.9;0)\,(5.2;0)\,\cdots$, $\w_B = (1.3;1.3)\,(2.6;2.6)\,(3.9;3.9)\,(5.2;5.2)\,\cdots$,  and $\w_\text{\tiny\XSolidBold} = (1.5;1.5)\,(2.5;3.2)\,(3.5;4.9)\,(4.5;6.6)\,\cdots$.

  The HyperSTL formulas $\psi_{\textsf{l-rob}}$ and $\psi_{\textsf{u-rob}}$ reason about sets of traces. 
  For example, the set $\SysSTL = \set{\w_0, \w_1, \w_A, \w_B}$ satisfies both formulas.
  If both $\pi_1$ and $\pi_2$ represent standard traces, then $\mapInp{\pi_1} = \mapInp{\pi_2}$, because $\mapInp{\w_0} = \mapInp{\w_1}$, and the formulas hold for $\pi'_2 = \pi_1$, respectively $\pi'_1 = \pi_2$.
  Otherwise,
  assume that $\pi_1$ represents $\w_0$ and $\pi_2$ represents $\w_B$ (the reasoning for other combinations of traces is similar). 

  First considering $\psi_{\textsf{l-rob}}$, we pick $\w_A$ for $\pi'_2$.
  We get that $\mapInp{\pi_2} = \mapInp{\pi'_2}$, because $\mapInp{\w_B} = \mapInp{\w_A}$. 
  Hence, we globally have 
  $\abs{\mapInp{\pi_2} - \mapInp{\pi'_2}} = 0$ and, thus,
  $\mathsf{eq}(\mapInp{\pi_2}, \mapInp{\pi'_2}) = 0$.
  At time steps $0 \leq t \leq 3$, the distance between the outputs $\abs{\mapOut{\w_0}(t) - \mapOut{\w_A}(t)}$ is at most $\outpbound$.
  Hence, the left operand of $\W$ 
  holds and the formula is satisfied for $t \leq 3$.
  At time $t = 3$ we have that $\abs{\mapInp{\w_0}(t) - \mapInp{\w_A}(t)} = \abs{4.0 - 5.2} > \inpbound$.
  Hence, the right operand of the $\W$ operator holds and $\psi_{\textsf{l-rob}}$ is satisfied also for $t \geq 3$.
  Notice that if 
  we would remove $\w_A$ from $\SysSTL$, then it would violate $\psi_{\textsf{l-rob}}$, because there is no possible choice for $\pi'_2$ that has the same inputs as $\w_B$ and where the output distances to $\w_0$ are below the $\outpbound$ threshold.

  To satisfy $\psi_{\textsf{u-rob}}$, we pick $\w_1$ for $\pi'_1$.
  The reasoning why the formula holds for this choice is analogue to $\psi_{\textsf{l-rob}}$.
  Notice that if we add the trace $\w_\text{\tiny\XSolidBold}$ to $\SysSTL$, then $\psi_{\textsf{u-rob}}$ is violated.
  Concretely, $\pi_2$ could represent $\w_\text{\tiny\XSolidBold}$; then, whether we pick $\w_0$ or $\w_1$ for $\pi'_1$, we eventually get outputs that violate the $\outpbound$ constraint, while the $\inpbound$ constraint is always satisfied.
  For example, if we compare $\w_\text{\tiny\XSolidBold}$ and $\w_1$, then we have for all time steps $t \leq 3$ that $\abs{\mapInp{\w_1}(t) - \mapInp{\w_\text{\tiny\XSolidBold}}(t)} = 0.5 \leq \inpbound$, but at time $t = 3$ we get  $\abs{\mapOut{\w_1}(t) - \mapOut{\w_\text{\tiny\XSolidBold}}(t)} = 2.6 > \outpbound$.
  Hence, at $t = 3$ the left and right operand of $\W$ are false, so $\psi_{\textsf{u-rob}}$ is violated.
\end{example}

\paragraph{Correctness under Mixed-IO Interpretation}
Mixed-IO signals are defined in the discrete time domain $\NN$ and value domain $\Inputs \cup \Outputs$.
The abstract functions 
$\mapInp{}$ and $\mapOut{}$ can be defined equally to the syntactically identical projection functions for mixed-IO models defined in \Cref{sec:background:doping}.
The function $\textsf{eq}(\inp_1, \inp_2)$ can be defined 
using the distance function $\dIn$ and some arbitrary small $\varepsilon > 0$:
\begin{align}
  \textsf{eq}(\inp_1, \inp_2) &\coloneqq \begin{cases}
    0, & \text{if } \inp_1 = \inp_2\\
    \dInMIO(\inp_1, \inp_2) + \varepsilon, & \text{if } \inp_1 \neq \inp_2 \land  \inp_1, \inp_2 \in \Inputs \\
    \infty, & \text{otherwise. }
  \end{cases} \label{eq:hyperstl:mixedIO:superD}
\end{align}
 In the second 
 clause of the above definition we add some positive value $\varepsilon$ to the result of $\dInMIO$, 
 because $\dInMIO(\inp_1, \inp_2)$ could be 0 even if $\inp_1 \neq \inp_2$.
For the correctness of 
the above HyperSTL formulas,
however, it is crucial that $\textsf{eq}(\inp_1, \inp_2) = 0$ if and only if $\inp_1= \inp_2$.
For a good performance of the falsification algorithm, we will nevertheless want to make use of $\dInMIO$ 
if $\inp_1 \neq \inp_2$.

\Cref{prop:hyperstl:robCorrectness} shows that HyperSTL formulas (\ref{eq:hyperstl:lrob}) and (\ref{eq:hyperstl:urob}) under the mixed-IO interpretation outlined above indeed characterise \robustLowCleannessNDet and \robustUpCleannessNDet.
\Cref{prop:hyperstl:fCorrectness} shows the same for \fCleannessNDet.

\begin{restatable}{proposition}{propHyperstlRobCorrectness}
  \label{prop:hyperstl:robCorrectness}
  Let $\mixedIOSys \subseteq \NN \to (\Inputs \cup \Outputs)$ be a mixed-IO system and $\Contract = \langle \Std,\allowbreak \dIn,\allowbreak \dOut,\allowbreak \inpbound,\allowbreak \outpbound \rangle$ a contract or context for \robustCleannessNDet with $\Std \subseteq \mixedIOSys$.
  Further, let $\Std_\pi$ be a quantifier-free HyperSTL subformula, such that $\mixedIOSys, \set{\pi \coloneqq w}, 0 \models \Std_\pi$ if and only if $\w \in \Std$. 
  Then, $\mixedIOSys$ is \robustlyLowCleanNDet \cleanFor \Contract  if and only if $\mixedIOSys, \emptyset, 0 \models \psi_{\text{l-rob}}$, and 
  $\mixedIOSys$ is \robustlyUpCleanNDet \cleanFor \Contract  if and only if $\mixedIOSys, \emptyset, 0 \models \psi_{\text{u-rob}}$. 
\end{restatable}

\begin{restatable}{proposition}{propHyperstlFCorrectness}
  \label{prop:hyperstl:fCorrectness}
  Let $\mixedIOSys \subseteq \NN \to (\Inputs \cup \Outputs)$ be a mixed-IO system and $\Contract = \langle \Std,\allowbreak \dIn,\allowbreak \dOut,\allowbreak f \rangle$ a contract or context for \fCleannessNDet with $\Std \subseteq \mixedIOSys$.
  Further, let $\Std_\pi$ be a quantifier-free HyperSTL subformula, such that $\mixedIOSys, \set{\pi \coloneqq w}, 0 \models \Std_\pi$ if and only if $\w \in \Std$. %, and let $\textsf{eq}$ be defined as in \Cref{eq:hyperstl:mixedIO:superD}.
  Then, $\mixedIOSys$ is \fLowCleanNDet \cleanFor \Contract  if and only if $\mixedIOSys, \emptyset, 0 \models \psi_{\text{l-fun}}$,
  and $\mixedIOSys$ is \fUpCleanNDet \cleanFor \Contract  if and only if $\mixedIOSys, \emptyset, 0 \models \psi_{\text{u-fun}}$.
\end{restatable}

\paragraph{STL Characterisation for Finite Standard Behaviour}
In many practical settings -- when the different standard behaviours are spelled out upfront explicitly, as in \NEDC and \WLTC{} -- it can be assumed that the number of distinct standard behaviours $\StdSet$ is finite (while there are infinitely many possible behaviours in $\SysSTL$). 
Finiteness of $\StdSet$ makes it possible to remove by enumeration the quantifiers from the \robustUpCleannessNDet and \fUpCleannessNDet HyperSTL formulas. 
This opens the way to work with the STL fragment of HyperSTL, after proper adjustments.
In the following, we assume that the set $\StdSet  = \set{w_1, \dots, w_\stdcnt}$ is an arbitrary standard set with $\stdcnt$ unique standard traces, where every $\w_k: \mathcal{T} \to X$ uses the same time domain $\mathcal{T}$ and value domain $X$.

To encode the HyperSTL formulas (\ref{eq:hyperstl:urob}) and (\ref{eq:hyperstl:ufunc}) in STL, we use the concept of \emph{self-composition}, which has proven useful for the analysis of hyperproperties~\cite{FinkbeinerRS15:cav,BartheDR11:mscs}.
We concatenate a trace under analysis $\w: \mathcal{T} \to X$ and the standard traces $\w_1$ to $\w_\stdcnt$ to the composed trace $w_+ = (\w, \w_1, \dots, \w_\stdcnt) \subseteq (\mathcal{T} \to X^{\stdcnt+1})$.
Given a system $\SysSTL \subseteq (\timeSet \to X)$ and a set $\StdSet  = \set{\w_1, \dots, \w_\stdcnt} \subseteq \SysSTL$, 
we denote by $\SysSTL \circ \StdSet \coloneqq \set{(\w, \w_1, \dots, \w_\stdcnt) \mid \w \in \SysSTL}$ the system in which every trace in \SysSTL is composed with the standard traces in \StdSet.
For every $\w_+ \in \SysSTL \circ \StdSet$, we will in the following STL formula write $\w$ to mean the projection on $\w_+$ to the trace $\w$, and we write $\w_k$, for $1 \leq k \leq \stdcnt$, to mean the projection  on $\w_+$ to the 
$k$th standard trace. 

\begin{restatable}{theorem}{thmHyperstlUrobCorrectnessSTL}
  \label{thm:hyperstl:urobCorrectnessSTL}
  Let $\mixedIOSys \subseteq \NN \to (\Inputs \cup \Outputs)$ be a mixed-IO system and $\Contract = \langle \Std,\allowbreak \dIn,\allowbreak \dOut,\allowbreak \inpbound,\allowbreak \outpbound \rangle$ a context for \robustCleannessNDet with finite standard behaviour $\Std = \set{w_1, \dots, w_\stdcnt} \subseteq \mixedIOSys$.
  Then, $\mixedIOSys$ is \robustlyUpCleanNDet \cleanFor \Contract  if and only if $(\mixedIOSys \circ \StdSet) \models \varphi_{\textsf{u-rob}}$, where
  \begin{align*}
    &\varphi_{\textsf{u-rob}} \coloneqq \bigwedge_{1 \leq a \leq \stdcnt} \ \bigvee_{1 \leq b \leq \stdcnt} 
    \Big({\G(\mathsf{eq}(\mapInp{\w_a}, \mapInp{\w_b}) \leq 0)}\ \land \\
    &\hspace{36.5mm} \big((d_\Outputs(\mapOut{\w_b}, \mapOut{\w})-\outpbound \leq 0)\W(d_\Inputs(\mapInp{\w_b}, \mapInp{\w}) - \inpbound > 0) \big)\Big).
  \end{align*}
\end{restatable}

The theorem for \fUpCleannessNDet is analogue to \Cref{thm:hyperstl:urobCorrectnessSTL}.

\begin{restatable}{theorem}{thmHyperstlUfCorrectnessSTL}
  \label{thm:hyperstl:ufCorrectnessSTL}
  Let $\mixedIOSys \subseteq \NN \to (\Inputs \cup \Outputs)$ be a mixed-IO system and $\Contract = \langle \Std,\allowbreak \dIn,\allowbreak \dOut,\allowbreak f \rangle$ a context for \fCleannessNDet with finite standard behaviour $\Std = \set{w_1, \dots, w_\stdcnt} \subseteq \mixedIOSys$.
  Then, $\mixedIOSys$ is \fUpCleanNDet \cleanFor \Contract  if and only if $(\mixedIOSys \circ \StdSet) \models \varphi_{\textsf{u-fun}}$, where 
  \begin{align*}
    &\varphi_{\textsf{u-fun}} \coloneqq \bigwedge_{1 \leq a \leq \stdcnt} \ \bigvee_{1 \leq b \leq \stdcnt} 
    \Big({\G(\mathsf{eq}(\mapInp{\w_a}, \mapInp{\w_b}) \leq 0)}\ \land \\
    &\hspace{52mm} \big(\G(\dOut(\mapOut{\w_b},\mapOut{\w}) - f(\dIn(\mapInp{\w_b},\mapInp{\w})) \leq 0) \big)\Big).
  \end{align*}
\end{restatable}

\begin{example}
  \label{ex:mixedIOstl}
  We consider the \robustCleannessNDet context $\Contract = \langle \StdSet, \dInMIO, \dOutMIO, \inpbound, \outpbound \rangle$
  where $\StdSet = \set{w_1, w_2}$ contains the two standard traces $w_1 = 1_\inp\, 2_\inp\, 3_\inp\, 7_\outp\, 0_\inp \, \quiescence^\omega$ and  $w_2 = 0_\inp\, 1_\inp\, 2_\inp\, 3_\inp\, 6_\outp \, \quiescence^\omega$.
  We here decorate inputs with index $\inp$ and outputs with index $\outp$, i.e., $w_1$ describes a system receiving the three inputs $1$, $2$, and $3$, then producing the output $7$, and finally receiving input $0$ before entering quiescence.
  We take 
  \begin{align*}
    \!    \dInMIO(\inp_1, \inp_2) = 
        \begin{cases}
        \abs{\inp_1 - \inp_2},\!\!\!\! & \text{ if } \inp_1, \inp_2\in\Inputs \\
        0, & \text{ if } \inp_1\! =\! \inp_2\! =\! \NoInp \\ 
        \infty, & \text{ otherwise,}
      \end{cases}
      \end{align*}
      and
  \begin{align*}
\!    \dOutMIO(\outp_1, \outp_2) = 
    \begin{cases}
    \abs{\outp_1 - \outp_2},\!\!\!\! & \text{ if } \outp_1, \outp_2\in\Outputs \backslash \set{\quiescence} \\
    0, & \text{ if } \outp_1\! =\! \outp_2\! =\! \NoOutp \text{ or } \outp_1\! =\! \outp_2\! =\! \quiescence\\
    \infty, & \text{ otherwise.}
  \end{cases}
  \end{align*} 
  The contractual value thresholds are assumed to be $\inpbound = 1$ and $\outpbound = 6$.

  Assume we are observing the trace $w = 0_\inp\, 1_\inp\, 2_\inp\, 6_\outp\, 0_\inp \, \quiescence^\omega$ to be monitored with 
  STL formula $\varphi_{\textsf{u-rob}}$ (from \Cref{prop:hyperstl:STLrobCleannessCorrectness}).
  First notice, that for combinations of $a$ and $b$ 
  in $\varphi_{\textsf{u-rob}}$, where $a \neq b$, the subformula $\G(\mathsf{eq}(\mapInp{\w_a}, \mapInp{\w_b}) \leq 0)$ is always false, because $\w_1$ and $\w_2$ 
  have different (input) values at time point 0. 
  Hence, it remains to show that
  \[ ({d}_\Outputs(\mapOut{w_1}, \mapOut{w})-\outpbound \leq 0)\W({d}_\Inputs(\mapInp{w_1}, \mapInp{w}) - \inpbound > 0)\hphantom{.\,} \land \]
  \[ ({d}_\Outputs(\mapOut{w_2}, \mapOut{w})-\outpbound \leq 0)\W({d}_\Inputs(\mapInp{w_2}, \mapInp{w}) - \inpbound > 0).\,\hphantom{\land} \]
  For the first conjunct, the input distance between inputs in $\w$ and $\w_1$ is always 1 at positions 1 to 3, it is 0 at position 4 (because $\NoInp$ is compared to $\NoInp$), and remains 0 in position 5 and beyond.
  Thus, ${d}_\Inputs(\mapInp{\w_1}, \mapInp{\w}) - \inpbound$ is always at most 0, and the right hand-side of the $\W$ operator is always false.
  Consequently, by definition of $\W$, the left operand of $\W$ must always hold, i.e., ${d}_\Outputs(\mapOut{\w_1}, \mapOut{\w})$ must always be less or equal to 6.
  This is the case for $\w_1$ and $\w$: at all positions except for 4, $\NoOutp$ is compared to $\NoOutp$ (or $\quiescence$ to $\quiescence$), so the difference is 0, and at position 4, the distance of 6 and~7~is~1.

  For the second $\W$-formula, $\w$ is compared to $\w_2$.
  These two traces are comparable only to a limited extent: the order of input and output is altered at the last two positions of the signals before quiescence.
  Hence, the right operand of $\W$ is true at position 4, and the formula holds for the remaining trace.
  For positions 1 to 3, the input distances are 0, because the input values are identical.
  At these positions, the left operand must hold.
  The values are input values, so $\NoOutp$ is compared to $\NoOutp$ at each position.
  This distance is defined to be 0, so it holds that $-6 \leq 0$, and the formula is satisfied.
  Since both formulas hold, the conjunction of both holds, too, and trace $\w$ is qualified as robustly clean.
  There could however be other system traces not considered in this example, that overall could violate robust cleanness of the system.
\end{example}

\paragraph{Restriction of input space}
\label{sec:restrictInput}
Robust cleanness puts semantic requirements on fragments of a system's input space, outside of which the system's behaviour remains unspecified. 
Typically, the fragment of the input space covered is rather small.
To falsify the STL formula $\varphi_{\textsf{u-rob}}$ from \Cref{prop:hyperstl:STLrobCleannessCorrectness}, the falsifier has two challenging tasks.
First, it has to find a way to stay in the relevant input space, i.e., select inputs with a distance of at most $\inpbound$ from the standard behaviour. 
Only if this is assured it can search for an output large enough to violate the $\outpbound$ requirement. In this, a large robustness estimate provided by the quantitative semantics of STL cannot serve as an indicator for deciding whether an input is too far off or whether an output stays too close to the standard behaviour.
We can improve the efficiency of the falsification process significantly by narrowing upfront the input space the falsifier uses.

In practice, test execution traces will always be finite.
In previous real-life doping tests, test execution lengths have been bounded by some constant $B \in \N$~\cite{DBLP:journals/tomacs/BiewerDH21}, i.e., systems are represented as sets of finite traces $\SysU \subseteq (\Inputs\cup\Outputs)^{B}$ (which for formality reasons each can be considered suffixed with $\quiescence^\omega$).
In this bounded horizon, we can provide a predicate discriminating between relevant and irrelevant input sequences.
Formally, the restriction to the relevant  input space fragment of a system $\SysU \subseteq (\Inputs\cup\Outputs)^{B}$  is given by the set
$\Inputs_{\StdSet, \inpbound} = \set{w \in \SysU \mid \exists w' \in \StdSet.\ \bigwedge_{k=0}^{B\!-\!1} (\dInMIO(\mapInp{w[k]}, \mapInp{w'[k]}) \leq \inpbound)}$.
Since $\StdSet$ and $B$ are finite, membership is computable.

There are rare cases in which this optimisation may prevent the falsifier from finding a counterexample. 
This is only the case if there is an input prefix leading to a violation of the formula  for which there is no suffix such that the whole trace satisfies the $\inpbound$ constraint. 
Below is a pathological example in which this could make a difference.

\begin{example}
  Apart from \NOx{} emissions, \NEDC (and \WLTC) tests are used to measure fuel consumption.
  Consider a contract similar to the contracts above, but with fuel rate as the output quantity.
  Assuming a ``normal'' fuel rate behaviour during the standard test, there might be a test within a reasonable $\inpbound$ distance, where the fuel is wasted insanely.
  Then, the fuel tank might run empty before the intended end of the test, which therefore could not be finished within the $\inpbound$ distance, because speed would be constantly 0 at the end.
  The actually driven test is not in set $\Inputs_{\StdSet, \inpbound}$, but there is a prefix within $\inpbound$ distance that violates the robust cleanness property.
\end{example}

Notably, there may be additional techniques to reduce the size of the input space. For example, if the next input symbol depends on the history of inputs, this constraint could be considered in the proposal scheme.

\section{Supervision of Diesel Emission Cleaning Systems}\label{sec:dieselSupervision}

The severity of the diesel emissions scandal showed that the regulations alone are insufficient to prevent car manufacturers from implementing tampered -- or doped -- emission cleaning systems.
Recent works~\cite{DBLP:journals/tomacs/BiewerDH21} shows that \robustCleannessNDet is a suitable means to extend the precisely defined behaviour of cars for the \NEDC to test cycles within a $\inpbound$ range around the \NEDC.
To demonstrate the usefulness of \robustCleannessNDet, the essential details of the emission testing scenario were modelled: the set of inputs is the set of speed values, an output value represents the amount of emissions -- in particular, the nitric oxide (\NOx) emissions -- measured at the exhaust pipe of a car.
The distance functions are the absolute differences of speed, respectively \NOx, values, and the standard behaviour is the singleton set that contains a trace that consists of the inputs that define the test cycle followed by the average amount of \NOx gas measured during the test.
Thus, formally, we 
get $\Inputs = \RR$, $\Outputs = \RR$, $\Std = \set{\NEDC\cdot \outp}$,\footnote{$\NEDC$ is the sequence of $1180$ inputs with the $k$th input defining the speed of the car after $k$ seconds from the beginning of the \NEDC} and $\dIn$ and $\dOut$ as defined in \Cref{ex:mixedIOstl}~\cite{DBLP:journals/tomacs/BiewerDH21}.

The STL formulas developed in the previous section, combined with the probabilistic falsification approach, give rise to further improvements to the existing testing-based work~\cite{DBLP:journals/tomacs/BiewerDH21} on diesel doping detection.

To use the falsification algorithm in \Cref{algo:falsification}, we implement the
restriction of the input space  to $\Inputs_{\set{\NEDC\cdot \outp},\inpbound}$ as explained in Section~\ref{sec:restrictInput}. 
With this restriction the STL formula $\varphi_{\textsf{u-rob}}$ from \Cref{prop:hyperstl:STLrobCleannessCorrectness} can be simplified to
\begin{align}
  \G(\dOutMIO(\mapOut{(\NEDC\cdot \outp)}, \mapOut{\w})-\outpbound \leq 0).
  \label{eq:stl:nedc-clean}
\end{align}

This is because the conjunction and disjunction over standard traces becomes obsolete for only a single standard trace.
For the same reason, the requirement $\G(\mathsf{eq}(\mapInp{\w_a}, \mapInp{\w_b}) \leq 0)$ becomes obsolete, as the compared traces are always identical.
In the $\W$ subformula, the right proposition is always false, because of the restricted input space. 
We implemented Algorithm~\ref{algo:falsification} for the robustness computation according to formula~(\ref{eq:stl:nedc-clean}).

In practice, running tests like \NEDC with real cars is a time consuming and expensive endeavour. Furthermore, tests on chassis dynamometers are usually prohibited to be carried out with rented cars by the rental companies.
On the other hand, car emission models for simulation are not available to the public -- and models provided by the manufacturer cannot be considered trustworthy. 
To carry out our experiments, we instead use an approximation technique that estimates the  amount of \NOx{} emissions of a car along a certain trajectory based on data recorded during previous trips with the same car, sampled at a frequency of \SI{1}{\hertz} (one sample per second).
Notably, these trips do not need to have much in common with the trajectory to be approximated.
A trip is represented as a finite sequence $\vartheta \in (\Reals \times \Reals \times\Reals)^*$ of triples, where 
each such triple $(v,a,n)$ represents the speed, the acceleration, and the (absolute) amount of \NOx{} emitted at a particular time instant in the sample.
Speed and acceleration can be considered as the main parameters influencing the instant emission of \NOx{}. This is, for instance, reflected in the regulation~\cite{DBLP:conf/rv/KohlHB18,LEX:32017R1151} where the decisive quantities to validate test routes for real-world driving emissions tests on public roads are speed and acceleration.

A recording \recording is the union of finitely many trips $\vartheta$. We can turn such a recording  into a predictor $\predictor$ of the \NOx{} values given pairs of speed and acceleration as follows:
\[
\predictor(v,a) = \mathsf{average} [n \ | \ (\exists v', a'.\ (\abs{v-v'} \leq 2 \ \land \abs{a-a'} \leq 2 \ \land \ (v', a', n) \in \recording))].
\]
The amount of \NOx{} assigned to a pair $(v,a)$ here is the average of all \NOx{} values seen in the recording \recording for $v \pm \ell$ and $a \pm \ell$, with $0 \leq \ell \leq 2$.
To overcome measurement inaccuracies and to increase the robustness of the approximated emissions, the speed and acceleration may deviate up to $\SI{2}{\kilo\meter\per\hour}$, and $\SI{2}{\meter\per\second\squared}$, respectively.
This tolerance is adopted from the official \NEDC regulation~\cite{nedc}, which allows up to $\SI{2}{\kilo\meter\per\hour}$ of deviations while driving the \NEDC.

\paragraph{}
To demonstrate the practical applicability of our implementation of Algorithm~\ref{algo:falsification} and our \NOx{} approximation, we report here on 
two
experiments with
an Audi A6 Avant Diesel admitted in June 2020 and with its successor admitted in 2021.
We will refer to the former as car \emph{A20} and to the latter as car \emph{A21}.
We used the app \LolaDrives to perform in total 
six
low-cost \RDE tests -- two with A20 and four\footnote{We do not consider test A21.3 in this article, see~\cite[Section 5]{C3OTHER:STTT-toappear} for details} with A21 -- and recorded the data received from the cars' diagnosis ports.
The raw data is available on Zenodo~\cite{biewer_sebastian_2023_8058770}.
Using the emissions predictor proposed above we estimate that for an \NEDC test A20 emits $\SI{86}{\milli\gram\per\kilo\meter}$ of \NOx and that A21 emits $\SI{9}{\milli\gram\per\kilo\meter}$.
Car A20 has previously been falsified w.r.t. the RDE specification. 
Neither A20 nor A21 has been falsified w.r.t. robust cleanness.

Before turning to falsification, we spell out meaningful contexts for \robustCleannessNDet.  
We identified suitable $\Inputs$, $\Outputs$, $\Std$, $\dIn$, and $\dOut$ at the beginning of the section.
For $\inpbound$, it turned out that $\inpbound = \SI{15}{\kilo\meter\per\hour}$ is a reasonable choice, as it leaves enough flexibility for human-caused driving mistakes and intended deviations~\cite{DBLP:journals/tomacs/BiewerDH21}.
The threshold for \NOx{} emissions under lab conditions is $\SI{80}{\milli\gram\per\kilo\meter}$.
The emission limits for \RDE tests depend on the admission date of the car.
Cars admitted in 2020 or earlier, must  emit $\SI{168}{\milli\gram\per\kilo\meter}$ at most, and cars admitted later must adhere to the limit of $\SI{120}{\milli\gram\per\kilo\meter}$.
For our experiments, we use $\outpbound = \SI{88}{\milli\gram\per\kilo\meter}$ for A20 and $\outpbound = \SI{40}{\milli\gram\per\kilo\meter}$ for A21 to have the same tolerances as for \RDE tests.
Effectively, the upper threshold for A20 is $84 + 88 = \SI{172}{\milli\gram\per\kilo\meter}$, and for A21 the limit is $9 + 40 = \SI{49}{\milli\gram\per\kilo\meter}$.
Notice that for software doping analysis, the output observed for a certain standard behaviour and the constant $\outpbound$ define the effective threshold; this threshold is typically different from the thresholds defined by the regulation.

We modified Algorithm~\ref{algo:falsification} by adding a timeout condition: if the algorithm is not able to find a falsifying counterexample within 3,000 iterations, it terminates and returns both the trace for which the smallest robustness has been observed and its corresponding robustness value.
Hence, if falsification of robust cleanness for a system is not possible, 
the algorithm outputs an upper bound on how robust the system satisfies robust cleanness.
\begin{figure}[tb]
  \begin{center}
  \hspace{1cm}    
  \pgfplotsset{
      axis line style={white}
    }
    \begin{tikzpicture}[trim axis left, scale=1]
      \begin{axis}[
          width=.97\columnwidth,
          height=0.42\columnwidth,
          xlabel={Time [s]},
          ylabel={Speed [$\frac{\mathit{km}}{h}$]},
          xmin=-2, xmax=1180,
          ymin=-2, ymax=125,
          xtick={0,200, 400, 600, 800, 1000, 1180},
          ytick={0, 32, 70, 100, 120},
          legend pos=north west,
          ymajorgrids=true,
          grid style=dotted,draw=gray!10,
      ]
       
      \addplot[
          color=blue,
          line width=.5pt,
          dash pattern=on 2pt off 1pt,
          ]
          coordinates {
          (0,0)(6,0)(11,0)(15,15)(23,15)(25,10)(28,0)(44,0)(49,0)(54,15)(56,15)(61,32)(85,32)(93,10)(96,0)(112,0)(117,0)(122,15)(124,15)(133,35)(135,35)(143,50)(155,50)(163,35)(176,35)(178,35)(185,10)(188,0)(195,0)(201,0)(206,0)(210,15)(218,15)(220,10)(223,0)(239,0)(244,0)(249,15)(251,15)(256,32)(280,32)(288,10)(291,0)(307,0)(312,0)(317,15)(319,15)(328,35)(330,35)(338,50)(350,50)(358,35)(371,35)(373,35)(380,10)(383,0)(390,0)(396,0)(401,0)(405,15)(413,15)(415,10)(418,0)(434,0)(439,0)(444,15)(446,15)(451,32)(475,32)(483,10)(486,0)(502,0)(507,0)(512,15)(514,15)(523,35)(525,35)(533,50)(545,50)(553,35)(566,35)(568,35)(575,10)(578,0)(585,0)(591,0)(596,0)(600,15)(608,15)(610,10)(613,0)(629,0)(634,0)(639,15)(641,15)(646,32)(670,32)(678,10)(681,0)(697,0)(702,0)(707,15)(709,15)(718,35)(720,35)(728,50)(740,50)(748,35)(761,35)(763,35)(770,10)(773,0)(780,0)(800,0)(805,15)(807,15)(816,35)(818,35)(826,50)(828,50)(841,70)(891,70)(895,60)(899,50)(968,50)(981,70)(1031,70)(1066,100)(1096,100)(1116,120)(1126,120)(1142,80)(1150,50)(1160,0)(1180,0)
          };
      
      \addplot [color=red,style=solid,line width=.2pt ]
        coordinates {(0,0)(1,0)(2,0)(3,0)(4,0)(5,0)(6,0)(7,0)(8,0)(9,0)(10,0)(11,0)(12,0)(13,0)(14,0)(15,0)(16,0)(17,2)(18,0)(19,0)(20,1)(21,3)(22,5)(23,3)(24,5)(25,6)(26,5)(27,7)(28,8)(29,6)(30,5)(31,4)(32,2)(33,1)(34,1)(35,0)(36,0)(37,0)(38,0)(39,0)(40,0)(41,0)(42,0)(43,0)(44,0)(45,0)(46,0)(47,0)(48,0)(49,0)(50,0)(51,0)(52,0)(53,1)(54,0)(55,0)(56,0)(57,1)(58,2)(59,4)(60,5)(61,7)(62,9)(63,10)(64,10)(65,12)(66,14)(67,15)(68,17)(69,17)(70,19)(71,20)(72,21)(73,23)(74,21)(75,21)(76,23)(77,24)(78,26)(79,27)(80,26)(81,26)(82,26)(83,26)(84,26)(85,25)(86,27)(87,28)(88,30)(89,29)(90,28)(91,27)(92,25)(93,23)(94,21)(95,19)(96,17)(97,15)(98,13)(99,13)(100,12)(101,10)(102,8)(103,6)(104,4)(105,2)(106,0)(107,0)(108,0)(109,0)(110,0)(111,0)(112,0)(113,0)(114,0)(115,0)(116,0)(117,0)(118,0)(119,0)(120,0)(121,0)(122,0)(123,0)(124,0)(125,1)(126,1)(127,2)(128,4)(129,6)(130,8)(131,9)(132,11)(133,12)(134,14)(135,15)(136,17)(137,19)(138,17)(139,19)(140,21)(141,23)(142,25)(143,27)(144,28)(145,30)(146,32)(147,33)(148,34)(149,36)(150,38)(151,39)(152,41)(153,43)(154,45)(155,44)(156,44)(157,44)(158,44)(159,45)(160,43)(161,41)(162,40)(163,38)(164,37)(165,36)(166,34)(167,32)(168,32)(169,30)(170,28)(171,26)(172,28)(173,30)(174,28)(175,30)(176,28)(177,30)(178,29)(179,28)(180,30)(181,28)(182,27)(183,27)(184,25)(185,25)(186,23)(187,21)(188,23)(189,21)(190,19)(191,20)(192,18)(193,16)(194,14)(195,12)(196,11)(197,9)(198,7)(199,5)(200,3)(201,1)(202,0)(203,0)(204,0)(205,0)(206,0)(207,0)(208,0)(209,0)(210,0)(211,0)(212,0)(213,0)(214,0)(215,2)(216,0)(217,1)(218,3)(219,4)(220,6)(221,7)(222,8)(223,6)(224,6)(225,4)(226,2)(227,0)(228,0)(229,0)(230,0)(231,0)(232,0)(233,0)(234,0)(235,0)(236,0)(237,0)(238,0)(239,1)(240,0)(241,1)(242,0)(243,0)(244,0)(245,0)(246,0)(247,0)(248,0)(249,0)(250,0)(251,0)(252,1)(253,2)(254,1)(255,2)(256,3)(257,5)(258,7)(259,9)(260,10)(261,12)(262,13)(263,15)(264,17)(265,19)(266,21)(267,22)(268,24)(269,26)(270,27)(271,27)(272,25)(273,26)(274,25)(275,24)(276,23)(277,25)(278,26)(279,28)(280,27)(281,25)(282,27)(283,27)(284,25)(285,23)(286,25)(287,23)(288,22)(289,20)(290,20)(291,18)(292,16)(293,14)(294,13)(295,12)(296,10)(297,8)(298,6)(299,4)(300,2)(301,0)(302,0)(303,0)(304,0)(305,0)(306,0)(307,0)(308,0)(309,0)(310,1)(311,1)(312,1)(313,0)(314,0)(315,0)(316,0)(317,1)(318,0)(319,1)(320,2)(321,3)(322,5)(323,6)(324,7)(325,7)(326,9)(327,10)(328,11)(329,11)(330,13)(331,15)(332,17)(333,17)(334,19)(335,20)(336,22)(337,24)(338,26)(339,28)(340,29)(341,31)(342,33)(343,35)(344,37)(345,39)(346,41)(347,42)(348,44)(349,42)(350,43)(351,43)(352,44)(353,46)(354,46)(355,44)(356,42)(357,43)(358,41)(359,39)(360,37)(361,35)(362,34)(363,32)(364,31)(365,29)(366,31)(367,29)(368,28)(369,30)(370,30)(371,29)(372,29)(373,31)(374,33)(375,32)(376,30)(377,28)(378,28)(379,27)(380,26)(381,24)(382,22)(383,20)(384,18)(385,16)(386,14)(387,12)(388,12)(389,11)(390,9)(391,7)(392,5)(393,5)(394,3)(395,1)(396,0)(397,0)(398,0)(399,0)(400,0)(401,0)(402,0)(403,0)(404,0)(405,0)(406,0)(407,0)(408,0)(409,0)(410,2)(411,3)(412,4)(413,5)(414,6)(415,8)(416,8)(417,9)(418,7)(419,5)(420,3)(421,3)(422,2)(423,0)(424,0)(425,0)(426,1)(427,0)(428,0)(429,0)(430,0)(431,0)(432,0)(433,1)(434,0)(435,0)(436,0)(437,0)(438,0)(439,0)(440,0)(441,0)(442,0)(443,1)(444,0)(445,0)(446,0)(447,0)(448,2)(449,4)(450,5)(451,4)(452,6)(453,8)(454,10)(455,11)(456,12)(457,11)(458,13)(459,15)(460,17)(461,19)(462,19)(463,21)(464,23)(465,24)(466,26)(467,25)(468,26)(469,28)(470,28)(471,27)(472,27)(473,26)(474,28)(475,26)(476,26)(477,28)(478,27)(479,26)(480,25)(481,24)(482,22)(483,22)(484,21)(485,20)(486,18)(487,18)(488,16)(489,14)(490,13)(491,12)(492,10)(493,8)(494,6)(495,4)(496,2)(497,0)(498,1)(499,1)(500,0)(501,0)(502,0)(503,0)(504,1)(505,0)(506,0)(507,0)(508,0)(509,0)(510,0)(511,0)(512,2)(513,0)(514,1)(515,3)(516,4)(517,3)(518,5)(519,6)(520,7)(521,7)(522,8)(523,10)(524,10)(525,12)(526,13)(527,15)(528,16)(529,18)(530,20)(531,20)(532,22)(533,24)(534,26)(535,28)(536,30)(537,32)(538,34)(539,36)(540,38)(541,39)(542,40)(543,41)(544,43)(545,44)(546,43)(547,41)(548,40)(549,40)(550,41)(551,40)(552,40)(553,38)(554,37)(555,36)(556,35)(557,34)(558,33)(559,32)(560,30)(561,31)(562,31)(563,30)(564,29)(565,29)(566,27)(567,28)(568,28)(569,30)(570,32)(571,30)(572,29)(573,27)(574,26)(575,27)(576,25)(577,24)(578,22)(579,20)(580,19)(581,17)(582,15)(583,13)(584,12)(585,11)(586,9)(587,8)(588,6)(589,5)(590,3)(591,2)(592,0)(593,0)(594,0)(595,0)(596,0)(597,0)(598,0)(599,0)(600,0)(601,0)(602,0)(603,0)(604,0)(605,2)(606,3)(607,4)(608,4)(609,5)(610,7)(611,5)(612,3)(613,4)(614,2)(615,1)(616,0)(617,0)(618,0)(619,0)(620,0)(621,0)(622,0)(623,0)(624,0)(625,0)(626,1)(627,0)(628,0)(629,1)(630,0)(631,0)(632,0)(633,0)(634,0)(635,0)(636,0)(637,0)(638,0)(639,0)(640,0)(641,1)(642,0)(643,1)(644,3)(645,5)(646,7)(647,6)(648,8)(649,9)(650,10)(651,12)(652,14)(653,15)(654,17)(655,19)(656,20)(657,22)(658,24)(659,26)(660,28)(661,28)(662,28)(663,26)(664,26)(665,26)(666,27)(667,26)(668,28)(669,26)(670,27)(671,27)(672,26)(673,25)(674,25)(675,23)(676,22)(677,22)(678,20)(679,19)(680,17)(681,18)(682,16)(683,14)(684,12)(685,10)(686,8)(687,6)(688,4)(689,3)(690,1)(691,0)(692,0)(693,0)(694,0)(695,0)(696,0)(697,0)(698,0)(699,0)(700,0)(701,0)(702,0)(703,0)(704,0)(705,0)(706,0)(707,2)(708,0)(709,0)(710,0)(711,2)(712,4)(713,5)(714,7)(715,7)(716,9)(717,11)(718,13)(719,14)(720,16)(721,16)(722,18)(723,20)(724,22)(725,24)(726,26)(727,27)(728,29)(729,30)(730,29)(731,31)(732,32)(733,34)(734,36)(735,38)(736,40)(737,42)(738,43)(739,45)(740,44)(741,43)(742,44)(743,44)(744,42)(745,41)(746,39)(747,39)(748,40)(749,38)(750,36)(751,35)(752,34)(753,32)(754,30)(755,30)(756,28)(757,29)(758,27)(759,29)(760,28)(761,28)(762,29)(763,27)(764,28)(765,28)(766,28)(767,28)(768,29)(769,27)(770,26)(771,25)(772,23)(773,21)(774,19)(775,18)(776,16)(777,14)(778,12)(779,10)(780,10)(781,8)(782,6)(783,4)(784,4)(785,2)(786,0)(787,0)(788,0)(789,0)(790,0)(791,0)(792,0)(793,0)(794,0)(795,0)(796,0)(797,0)(798,0)(799,0)(800,0)(801,0)(802,0)(803,0)(804,0)(805,0)(806,0)(807,0)(808,0)(809,1)(810,3)(811,3)(812,5)(813,6)(814,8)(815,10)(816,12)(817,14)(818,16)(819,17)(820,19)(821,21)(822,23)(823,25)(824,27)(825,29)(826,30)(827,31)(828,33)(829,34)(830,36)(831,36)(832,37)(833,39)(834,41)(835,43)(836,45)(837,47)(838,47)(839,49)(840,51)(841,53)(842,55)(843,57)(844,59)(845,60)(846,60)(847,62)(848,62)(849,64)(850,63)(851,64)(852,65)(853,63)(854,64)(855,62)(856,62)(857,63)(858,62)(859,64)(860,63)(861,64)(862,63)(863,62)(864,63)(865,65)(866,65)(867,67)(868,65)(869,67)(870,68)(871,66)(872,66)(873,64)(874,66)(875,67)(876,65)(877,63)(878,63)(879,63)(880,64)(881,64)(882,66)(883,64)(884,63)(885,61)(886,63)(887,63)(888,61)(889,63)(890,65)(891,64)(892,62)(893,60)(894,60)(895,60)(896,59)(897,60)(898,60)(899,59)(900,57)(901,55)(902,53)(903,51)(904,51)(905,49)(906,47)(907,45)(908,45)(909,45)(910,46)(911,44)(912,44)(913,42)(914,43)(915,45)(916,47)(917,45)(918,43)(919,44)(920,45)(921,45)(922,46)(923,45)(924,47)(925,46)(926,44)(927,46)(928,48)(929,46)(930,47)(931,46)(932,45)(933,45)(934,47)(935,45)(936,46)(937,44)(938,46)(939,44)(940,43)(941,43)(942,42)(943,44)(944,44)(945,45)(946,46)(947,46)(948,46)(949,44)(950,46)(951,44)(952,43)(953,44)(954,44)(955,43)(956,44)(957,46)(958,45)(959,43)(960,43)(961,45)(962,47)(963,45)(964,46)(965,47)(966,46)(967,48)(968,47)(969,49)(970,48)(971,47)(972,48)(973,47)(974,48)(975,50)(976,49)(977,50)(978,51)(979,52)(980,54)(981,55)(982,57)(983,59)(984,61)(985,63)(986,64)(987,66)(988,64)(989,66)(990,67)(991,65)(992,65)(993,63)(994,63)(995,65)(996,65)(997,64)(998,64)(999,65)(1000,64)(1001,65)(1002,66)(1003,66)(1004,67)(1005,66)(1006,64)(1007,66)(1008,64)(1009,66)(1010,65)(1011,66)(1012,64)(1013,64)(1014,64)(1015,65)(1016,63)(1017,61)(1018,62)(1019,63)(1020,64)(1021,63)(1022,64)(1023,64)(1024,64)(1025,65)(1026,66)(1027,64)(1028,63)(1029,62)(1030,64)(1031,66)(1032,64)(1033,65)(1034,63)(1035,64)(1036,64)(1037,65)(1038,65)(1039,67)(1040,69)(1041,70)(1042,71)(1043,73)(1044,72)(1045,72)(1046,73)(1047,74)(1048,73)(1049,74)(1050,74)(1051,76)(1052,78)(1053,78)(1054,80)(1055,82)(1056,84)(1057,83)(1058,85)(1059,86)(1060,87)(1061,87)(1062,86)(1063,88)(1064,89)(1065,90)(1066,88)(1067,90)(1068,92)(1069,94)(1070,96)(1071,95)(1072,95)(1073,93)(1074,94)(1075,93)(1076,93)(1077,95)(1078,95)(1079,93)(1080,95)(1081,94)(1082,93)(1083,94)(1084,92)(1085,93)(1086,94)(1087,93)(1088,94)(1089,94)(1090,95)(1091,93)(1092,92)(1093,93)(1094,95)(1095,96)(1096,95)(1097,94)(1098,95)(1099,94)(1100,96)(1101,97)(1102,97)(1103,99)(1104,99)(1105,97)(1106,98)(1107,100)(1108,99)(1109,100)(1110,102)(1111,104)(1112,106)(1113,107)(1114,109)(1115,111)(1116,111)(1117,110)(1118,110)(1119,112)(1120,114)(1121,114)(1122,115)(1123,113)(1124,115)(1125,116)(1126,116)(1127,114)(1128,115)(1129,113)(1130,113)(1131,111)(1132,110)(1133,109)(1134,108)(1135,106)(1136,104)(1137,102)(1138,100)(1139,100)(1140,98)(1141,96)(1142,96)(1143,94)(1144,92)(1145,90)(1146,88)(1147,86)(1148,86)(1149,84)(1150,83)(1151,81)(1152,80)(1153,78)(1154,76)(1155,74)(1156,72)(1157,70)(1158,68)(1159,65)(1160,60)(1161,55)(1162,50)(1163,45)(1164,40)(1165,35)(1166,30)(1167,25)(1168,20)(1169,15)(1170,13)(1171,11)(1172,9)(1173,7)(1174,7)(1175,6)(1176,5)(1177,4)(1178,4)(1179,2)(1180,1)}; % Output=181.64568720064392

      %\addplot[color=blue,mark=*,mark size=1.3pt,only marks]
      %	coordinates {
      %    (61,32)(143,50)(280,32)(374,35)(841,70)(899, 50)(1031,70)(1066,100)
      %    };
      %    
      %\node at (axis cs:61,32) [anchor=south] {$i_{61}$};
      %\node at (axis cs:143,50) [anchor=south] {$i_{143}$};
      %\node at (axis cs:280,32) [anchor=south] {$i_{280}$};
      %\node at (axis cs:374,35) [anchor=south west] {$i_{374}$};
      %\node at (axis cs:841,70) [anchor=south] {$i_{841}$};
      %\node at (axis cs:899,50) [anchor=north] {$i_{899}$};
      %\node at (axis cs:1031,70) [anchor=south east] {$i_{1031}$};
      %\node at (axis cs:1066,100) [anchor=south east] {$i_{1066}$};
       %
      \end{axis}
      \end{tikzpicture}
      
  \caption{\NEDC speed profile (blue, dashed) and input falsifying $\Contract$ for $\outpbound = \SI{88}{\milli\gram\per\kilo\meter}$ (red) with $\SI{182}{\milli\gram\per\kilo\meter}$ of emitted \NOx{}.}
  \label{fig:nox-falsification}
  \end{center}
\end{figure}

For the concrete case of the diesel emissions, the robustness value during the first 1180 inputs (sampled from the restricted input space $\Inputs_{\StdSet, \inpbound}$) is always $\outpbound$.
When the \NEDC output $o_\NEDC$ and the non-standard output $o$ are compared, the robustness value is $\outpbound - \abs{o_\NEDC - o}$ (cf., eq.~(\ref{eq:stl:nedc-clean}), the quantitative semantics of STL, and definition of $\dOutMIO$).
Hence, for test cycles with small robustness values, we get \NOx{} emissions $o$ that are either very small or very large compared to $o_\NEDC$.
We ran the modified Algorithm~\ref{algo:falsification}  
on A20 and A21 for the contexts defined above.
For A20, it found a robustness value of $-8$, i.e., it was able to falsify robust cleanness relative to the assumed contract and found a test cycle for which \NOx emissions of $\SI{182}{\milli\gram\per\kilo\meter}$ are predicted.
The test cycle is shown in \Cref{fig:nox-falsification}.
For A21, the smallest robustness estimate found -- even after 100 independent executions of the algorithm -- was $38$, i.e., 
A21 is predicted to satisfy robust cleanness with a 
very high robustness estimate.
The corresponding test cycle is shown in \Cref{fig:nox-optimisation}.
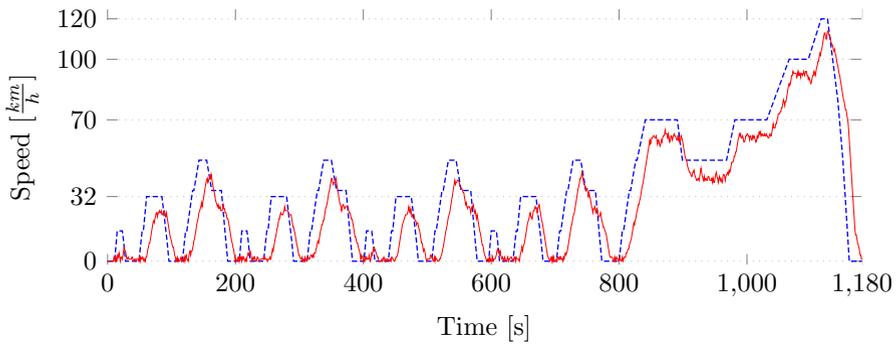
\begin{figure}[tb]
  \begin{center}
  \hspace{1cm}    
  \pgfplotsset{
      axis line style={white}
    }
    \begin{tikzpicture}[trim axis left, scale=1]
      \begin{axis}[
          width=.97\columnwidth,
          height=0.42\columnwidth,
          xlabel={Time [s]},
          ylabel={Speed [$\frac{\mathit{km}}{h}$]},
          xmin=-2, xmax=1180,
          ymin=-2, ymax=125,
          xtick={0,200, 400, 600, 800, 1000, 1180},
          ytick={0, 32, 70, 100, 120},
          legend pos=north west,
          ymajorgrids=true,
          grid style=dotted,draw=gray!10,
      ]
       
      \addplot[
          color=blue,
          line width=.5pt,
          dash pattern=on 2pt off 1pt,
          ]
          coordinates {
          (0,0)(6,0)(11,0)(15,15)(23,15)(25,10)(28,0)(44,0)(49,0)(54,15)(56,15)(61,32)(85,32)(93,10)(96,0)(112,0)(117,0)(122,15)(124,15)(133,35)(135,35)(143,50)(155,50)(163,35)(176,35)(178,35)(185,10)(188,0)(195,0)(201,0)(206,0)(210,15)(218,15)(220,10)(223,0)(239,0)(244,0)(249,15)(251,15)(256,32)(280,32)(288,10)(291,0)(307,0)(312,0)(317,15)(319,15)(328,35)(330,35)(338,50)(350,50)(358,35)(371,35)(373,35)(380,10)(383,0)(390,0)(396,0)(401,0)(405,15)(413,15)(415,10)(418,0)(434,0)(439,0)(444,15)(446,15)(451,32)(475,32)(483,10)(486,0)(502,0)(507,0)(512,15)(514,15)(523,35)(525,35)(533,50)(545,50)(553,35)(566,35)(568,35)(575,10)(578,0)(585,0)(591,0)(596,0)(600,15)(608,15)(610,10)(613,0)(629,0)(634,0)(639,15)(641,15)(646,32)(670,32)(678,10)(681,0)(697,0)(702,0)(707,15)(709,15)(718,35)(720,35)(728,50)(740,50)(748,35)(761,35)(763,35)(770,10)(773,0)(780,0)(800,0)(805,15)(807,15)(816,35)(818,35)(826,50)(828,50)(841,70)(891,70)(895,60)(899,50)(968,50)(981,70)(1031,70)(1066,100)(1096,100)(1116,120)(1126,120)(1142,80)(1150,50)(1160,0)(1180,0)
          };
      
      \addplot [color=red,style=solid,line width=.2pt ]
        coordinates {(0,0)(1,0)(2,0)(3,0)(4,0)(5,0)(6,0)(7,0)(8,0)(9,0)(10,2)(11,1)(12,0)(13,2)(14,0)(15,1)(16,3)(17,2)(18,4)(19,2)(20,0)(21,2)(22,2)(23,1)(24,3)(25,5)(26,7)(27,5)(28,5)(29,3)(30,1)(31,1)(32,2)(33,1)(34,1)(35,1)(36,1)(37,1)(38,2)(39,2)(40,0)(41,0)(42,1)(43,1)(44,0)(45,1)(46,1)(47,1)(48,1)(49,0)(50,2)(51,0)(52,2)(53,1)(54,1)(55,2)(56,1)(57,0)(58,1)(59,0)(60,2)(61,2)(62,4)(63,6)(64,7)(65,9)(66,9)(67,11)(68,13)(69,15)(70,17)(71,19)(72,20)(73,19)(74,21)(75,22)(76,24)(77,24)(78,22)(79,24)(80,24)(81,25)(82,24)(83,24)(84,25)(85,23)(86,23)(87,24)(88,25)(89,25)(90,23)(91,22)(92,24)(93,22)(94,20)(95,18)(96,16)(97,14)(98,12)(99,11)(100,10)(101,8)(102,7)(103,7)(104,5)(105,3)(106,2)(107,0)(108,1)(109,2)(110,0)(111,1)(112,0)(113,0)(114,1)(115,2)(116,0)(117,1)(118,1)(119,0)(120,2)(121,0)(122,2)(123,3)(124,1)(125,1)(126,0)(127,2)(128,4)(129,5)(130,3)(131,5)(132,5)(133,6)(134,8)(135,10)(136,11)(137,13)(138,15)(139,17)(140,19)(141,20)(142,22)(143,24)(144,25)(145,26)(146,27)(147,29)(148,27)(149,29)(150,31)(151,33)(152,35)(153,37)(154,39)(155,40)(156,41)(157,39)(158,40)(159,41)(160,42)(161,41)(162,43)(163,41)(164,39)(165,37)(166,35)(167,34)(168,32)(169,30)(170,29)(171,28)(172,27)(173,28)(174,26)(175,26)(176,27)(177,25)(178,25)(179,25)(180,26)(181,28)(182,26)(183,26)(184,24)(185,24)(186,23)(187,21)(188,21)(189,20)(190,18)(191,16)(192,14)(193,12)(194,10)(195,9)(196,7)(197,6)(198,4)(199,2)(200,0)(201,1)(202,0)(203,0)(204,2)(205,2)(206,0)(207,2)(208,1)(209,2)(210,0)(211,2)(212,1)(213,2)(214,1)(215,0)(216,1)(217,0)(218,2)(219,3)(220,4)(221,3)(222,4)(223,5)(224,4)(225,2)(226,0)(227,1)(228,2)(229,1)(230,2)(231,1)(232,0)(233,0)(234,2)(235,2)(236,1)(237,0)(238,0)(239,1)(240,1)(241,1)(242,1)(243,1)(244,2)(245,1)(246,2)(247,1)(248,1)(249,1)(250,0)(251,1)(252,1)(253,1)(254,2)(255,1)(256,2)(257,4)(258,6)(259,5)(260,7)(261,9)(262,11)(263,12)(264,14)(265,16)(266,17)(267,19)(268,21)(269,21)(270,23)(271,24)(272,23)(273,23)(274,23)(275,24)(276,25)(277,27)(278,25)(279,25)(280,25)(281,25)(282,25)(283,24)(284,23)(285,23)(286,23)(287,23)(288,21)(289,19)(290,17)(291,17)(292,15)(293,14)(294,12)(295,10)(296,8)(297,7)(298,5)(299,3)(300,3)(301,2)(302,1)(303,1)(304,2)(305,0)(306,1)(307,1)(308,0)(309,2)(310,1)(311,0)(312,0)(313,2)(314,2)(315,1)(316,1)(317,2)(318,3)(319,3)(320,2)(321,3)(322,2)(323,4)(324,4)(325,4)(326,6)(327,7)(328,9)(329,10)(330,9)(331,11)(332,13)(333,15)(334,16)(335,18)(336,18)(337,20)(338,20)(339,22)(340,24)(341,26)(342,27)(343,28)(344,30)(345,32)(346,34)(347,35)(348,37)(349,39)(350,41)(351,41)(352,39)(353,38)(354,39)(355,39)(356,39)(357,41)(358,39)(359,37)(360,35)(361,33)(362,31)(363,29)(364,28)(365,27)(366,25)(367,27)(368,28)(369,28)(370,26)(371,26)(372,27)(373,27)(374,27)(375,27)(376,27)(377,25)(378,24)(379,24)(380,23)(381,21)(382,21)(383,20)(384,18)(385,16)(386,14)(387,12)(388,12)(389,10)(390,9)(391,9)(392,7)(393,6)(394,5)(395,3)(396,1)(397,3)(398,2)(399,2)(400,0)(401,1)(402,0)(403,2)(404,1)(405,2)(406,1)(407,1)(408,1)(409,3)(410,1)(411,0)(412,2)(413,4)(414,3)(415,5)(416,7)(417,7)(418,6)(419,5)(420,4)(421,3)(422,2)(423,2)(424,0)(425,0)(426,0)(427,0)(428,0)(429,1)(430,0)(431,1)(432,2)(433,1)(434,0)(435,1)(436,0)(437,1)(438,1)(439,1)(440,0)(441,0)(442,1)(443,0)(444,2)(445,2)(446,1)(447,1)(448,0)(449,0)(450,1)(451,2)(452,2)(453,3)(454,4)(455,6)(456,8)(457,10)(458,12)(459,14)(460,16)(461,18)(462,19)(463,20)(464,21)(465,22)(466,23)(467,23)(468,23)(469,24)(470,25)(471,27)(472,25)(473,24)(474,26)(475,24)(476,22)(477,23)(478,23)(479,23)(480,21)(481,20)(482,20)(483,18)(484,19)(485,19)(486,17)(487,16)(488,15)(489,13)(490,11)(491,9)(492,9)(493,7)(494,5)(495,6)(496,5)(497,3)(498,1)(499,0)(500,1)(501,1)(502,1)(503,1)(504,1)(505,3)(506,3)(507,2)(508,2)(509,0)(510,0)(511,2)(512,0)(513,2)(514,1)(515,1)(516,0)(517,0)(518,0)(519,1)(520,1)(521,2)(522,3)(523,4)(524,6)(525,7)(526,9)(527,11)(528,13)(529,15)(530,16)(531,18)(532,20)(533,22)(534,23)(535,24)(536,26)(537,28)(538,30)(539,32)(540,34)(541,36)(542,38)(543,40)(544,38)(545,39)(546,39)(547,40)(548,40)(549,40)(550,39)(551,37)(552,37)(553,36)(554,36)(555,35)(556,34)(557,34)(558,32)(559,30)(560,28)(561,29)(562,27)(563,26)(564,27)(565,27)(566,27)(567,26)(568,24)(569,26)(570,27)(571,27)(572,28)(573,26)(574,24)(575,25)(576,23)(577,21)(578,23)(579,21)(580,19)(581,17)(582,15)(583,14)(584,12)(585,10)(586,8)(587,8)(588,6)(589,4)(590,3)(591,2)(592,1)(593,0)(594,1)(595,1)(596,0)(597,2)(598,0)(599,1)(600,0)(601,0)(602,2)(603,0)(604,1)(605,1)(606,1)(607,2)(608,3)(609,5)(610,6)(611,5)(612,5)(613,3)(614,5)(615,3)(616,2)(617,1)(618,2)(619,1)(620,0)(621,2)(622,1)(623,0)(624,0)(625,1)(626,2)(627,0)(628,1)(629,1)(630,2)(631,2)(632,2)(633,0)(634,0)(635,0)(636,2)(637,2)(638,0)(639,2)(640,1)(641,1)(642,1)(643,1)(644,0)(645,2)(646,4)(647,5)(648,7)(649,5)(650,7)(651,8)(652,9)(653,11)(654,13)(655,15)(656,17)(657,19)(658,20)(659,21)(660,23)(661,24)(662,22)(663,22)(664,21)(665,23)(666,25)(667,25)(668,26)(669,27)(670,26)(671,26)(672,24)(673,26)(674,27)(675,28)(676,26)(677,24)(678,22)(679,20)(680,19)(681,17)(682,15)(683,13)(684,11)(685,9)(686,7)(687,5)(688,3)(689,2)(690,1)(691,1)(692,0)(693,1)(694,2)(695,0)(696,2)(697,0)(698,2)(699,0)(700,0)(701,0)(702,1)(703,0)(704,2)(705,2)(706,1)(707,2)(708,1)(709,2)(710,1)(711,1)(712,0)(713,2)(714,1)(715,3)(716,4)(717,3)(718,4)(719,5)(720,7)(721,9)(722,11)(723,13)(724,15)(725,17)(726,18)(727,20)(728,22)(729,23)(730,24)(731,26)(732,28)(733,30)(734,32)(735,34)(736,36)(737,37)(738,39)(739,41)(740,41)(741,42)(742,44)(743,42)(744,40)(745,38)(746,38)(747,36)(748,35)(749,36)(750,36)(751,34)(752,32)(753,34)(754,32)(755,31)(756,29)(757,27)(758,26)(759,24)(760,26)(761,26)(762,27)(763,25)(764,24)(765,26)(766,25)(767,23)(768,22)(769,22)(770,23)(771,24)(772,22)(773,20)(774,19)(775,17)(776,17)(777,16)(778,15)(779,14)(780,13)(781,11)(782,9)(783,7)(784,5)(785,3)(786,2)(787,2)(788,2)(789,2)(790,1)(791,0)(792,0)(793,2)(794,0)(795,1)(796,2)(797,1)(798,1)(799,0)(800,0)(801,2)(802,1)(803,0)(804,0)(805,2)(806,0)(807,1)(808,2)(809,2)(810,2)(811,4)(812,4)(813,5)(814,6)(815,6)(816,8)(817,9)(818,10)(819,11)(820,13)(821,15)(822,17)(823,19)(824,18)(825,20)(826,21)(827,23)(828,25)(829,26)(830,28)(831,28)(832,30)(833,31)(834,33)(835,35)(836,37)(837,38)(838,40)(839,42)(840,44)(841,46)(842,48)(843,50)(844,52)(845,54)(846,56)(847,58)(848,60)(849,58)(850,59)(851,61)(852,61)(853,62)(854,62)(855,62)(856,61)(857,61)(858,62)(859,60)(860,58)(861,60)(862,60)(863,58)(864,57)(865,58)(866,60)(867,59)(868,61)(869,62)(870,62)(871,63)(872,65)(873,63)(874,63)(875,63)(876,61)(877,63)(878,61)(879,59)(880,60)(881,62)(882,62)(883,60)(884,61)(885,62)(886,63)(887,62)(888,62)(889,61)(890,60)(891,60)(892,62)(893,60)(894,61)(895,59)(896,60)(897,62)(898,61)(899,60)(900,58)(901,56)(902,54)(903,52)(904,51)(905,50)(906,49)(907,47)(908,46)(909,45)(910,44)(911,44)(912,43)(913,42)(914,42)(915,44)(916,44)(917,45)(918,43)(919,42)(920,41)(921,40)(922,40)(923,41)(924,41)(925,40)(926,42)(927,44)(928,46)(929,44)(930,42)(931,41)(932,40)(933,42)(934,40)(935,41)(936,41)(937,40)(938,40)(939,42)(940,42)(941,41)(942,40)(943,40)(944,40)(945,41)(946,39)(947,41)(948,39)(949,39)(950,41)(951,43)(952,44)(953,42)(954,41)(955,39)(956,41)(957,42)(958,40)(959,41)(960,40)(961,39)(962,41)(963,43)(964,43)(965,42)(966,40)(967,42)(968,43)(969,43)(970,44)(971,43)(972,45)(973,44)(974,45)(975,47)(976,49)(977,49)(978,49)(979,50)(980,52)(981,52)(982,54)(983,54)(984,56)(985,58)(986,60)(987,61)(988,63)(989,62)(990,63)(991,61)(992,61)(993,62)(994,61)(995,59)(996,59)(997,60)(998,62)(999,62)(1000,63)(1001,63)(1002,62)(1003,62)(1004,60)(1005,62)(1006,62)(1007,60)(1008,60)(1009,62)(1010,61)(1011,62)(1012,61)(1013,60)(1014,59)(1015,60)(1016,62)(1017,61)(1018,62)(1019,64)(1020,63)(1021,64)(1022,62)(1023,63)(1024,63)(1025,62)(1026,62)(1027,63)(1028,62)(1029,62)(1030,63)(1031,62)(1032,62)(1033,61)(1034,61)(1035,62)(1036,63)(1037,62)(1038,62)(1039,64)(1040,66)(1041,68)(1042,66)(1043,66)(1044,66)(1045,68)(1046,68)(1047,70)(1048,71)(1049,73)(1050,73)(1051,73)(1052,74)(1053,74)(1054,76)(1055,76)(1056,78)(1057,80)(1058,79)(1059,77)(1060,79)(1061,79)(1062,81)(1063,83)(1064,85)(1065,87)(1066,88)(1067,87)(1068,88)(1069,90)(1070,92)(1071,92)(1072,94)(1073,93)(1074,93)(1075,93)(1076,93)(1077,93)(1078,91)(1079,92)(1080,94)(1081,92)(1082,92)(1083,93)(1084,91)(1085,93)(1086,93)(1087,92)(1088,93)(1089,94)(1090,92)(1091,93)(1092,94)(1093,94)(1094,92)(1095,92)(1096,90)(1097,91)(1098,90)(1099,89)(1100,88)(1101,89)(1102,91)(1103,92)(1104,92)(1105,94)(1106,93)(1107,92)(1108,94)(1109,96)(1110,97)(1111,99)(1112,99)(1113,101)(1114,103)(1115,102)(1116,104)(1117,105)(1118,107)(1119,109)(1120,111)(1121,113)(1122,113)(1123,111)(1124,110)(1125,112)(1126,114)(1127,114)(1128,112)(1129,112)(1130,111)(1131,109)(1132,108)(1133,107)(1134,107)(1135,105)(1136,104)(1137,105)(1138,103)(1139,102)(1140,100)(1141,98)(1142,97)(1143,95)(1144,93)(1145,92)(1146,90)(1147,88)(1148,86)(1149,84)(1150,82)(1151,80)(1152,79)(1153,77)(1154,75)(1155,75)(1156,73)(1157,71)(1158,69)(1159,65)(1160,60)(1161,55)(1162,50)(1163,45)(1164,40)(1165,35)(1166,30)(1167,25)(1168,20)(1169,15)(1170,13)(1171,13)(1172,11)(1173,9)(1174,8)(1175,7)(1176,5)(1177,3)(1178,3)(1179,1)(1180,1)}; % Output=11.23740074732794

      %\addplot[color=blue,mark=*,mark size=1.3pt,only marks]
      %	coordinates {
      %    (61,32)(143,50)(280,32)(374,35)(841,70)(899, 50)(1031,70)(1066,100)
      %    };
      %    
      %\node at (axis cs:61,32) [anchor=south] {$i_{61}$};
      %\node at (axis cs:143,50) [anchor=south] {$i_{143}$};
      %\node at (axis cs:280,32) [anchor=south] {$i_{280}$};
      %\node at (axis cs:374,35) [anchor=south west] {$i_{374}$};
      %\node at (axis cs:841,70) [anchor=south] {$i_{841}$};
      %\node at (axis cs:899,50) [anchor=north] {$i_{899}$};
      %\node at (axis cs:1031,70) [anchor=south east] {$i_{1031}$};
      %\node at (axis cs:1066,100) [anchor=south east] {$i_{1066}$};
       %
      \end{axis}
      \end{tikzpicture}
      
  \caption{\NEDC speed profile (blue, dashed) and input maximising \NOx{} emissions to $\SI{11}{\milli\gram\per\kilo\meter}$ (red).}
  \label{fig:nox-optimisation}
  \end{center}
\end{figure}

\paragraph{On Doping Tests for Cyber-physical Systems}
\label{sec:testing:integrated}
The proposed probabilistic falsification approach to find instances of software doping needs several hundreds of iterations.
This is problematic for testing real-world cyber-physical systems (CPS) to which inputs cannot be passed in an automated way.
To conduct a test with a car, for example, the input to the system is a test cycle that is passed to the vehicle by driving it.
Notably, we consider here the scenario that the CPS is tested by an entity that is different from the manufacturer. 
While the latter might have tools to overcome these technical challenges, the former typically does not have access to them.

We propose the following \emph{integrated testing approach} for effective doping tests of cyber-physical systems. The big picture is provided in \Cref{fig:testing:integrated}.
In a first step, the CPS is used under real-world conditions without enforcing any specific constraints on the inputs to the system.
For all executions, the inputs and outputs are recorded.
So, essentially, the system can be used as it is needed by the user, but all interactions with it are recorded.
From these recordings, a \emph{model} can be learned that for arbitrary inputs (whether they were covered in the recorded data or not) predicts the output of the system.
Such learning can be as simple as using statistics as we did for the emissions example above, or as complex as using deep neural nets.
For the learned model, the probabilistic falsification algorithm computes a test input that falsifies it -- inputs to this model can be passed automatically and an output is produced almost instantly.
The resulting input serves as an input for the real CPS. 
If the prediction was correct, also the real system is falsified.
If it was incorrect, the learned model can be refined and the process starts again.
\begin{figure}
  \centering
  \includegraphics[width=8.5cm]{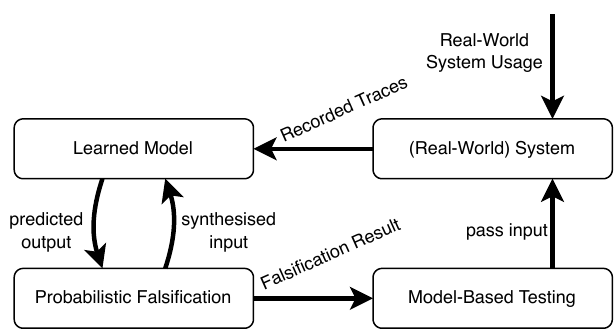}
  \caption{Integrated testing approach}
  \label{fig:testing:integrated}
\end{figure}

For diesel emissions, the first part of this integrated testing approach has been carried out as part of the work reported in this article.
We leave the second part -- evaluating the generated test traces from \Cref{fig:nox-falsification,fig:nox-optimisation} with a real car -- for future work.

\paragraph{Technical Context}
Software doping theory provides a formal basis for  enlarging the requirements on vehicle exhaust emissions beyond too narrow lab test conditions. That conceptual limitation has by now been addressed by the official authorities responsible for car type approval~\cite{LEX:32017R1151,TUTUIANU201561}:
The old \NEDC-based test procedure is replaced by the newer \emph{Worldwide Harmonised Light Vehicles Test Procedure} (\WLTP), which is deemed to be more realistic. \WLTP replaces the \NEDC test by a new \WLTC test, but \WLTC still is just a single test scenario.
In addition, \WLTP embraces so called  \emph{Real Driving Emissions} (\RDE) tests to be conducted on public roads. A recently launched mobile phone app~\cite{DBLP:conf/tacas/BiewerFHKSS21,C3OTHER:STTT-toappear}, \LolaDrives, harvests runtime monitoring technology for making low-cost \RDE tests accessible to everyone.

Learning or approximating the behaviour of a system under test has been studied intensively.
Meinke and Sindhu~\cite{DBLP:conf/tap/MeinkeS11} were among the first to present a testing approach incrementally learning a Kripke structure representing a reactive system.
Volpato and Tretmans~\cite{DBLP:journals/eceasst/VolpatoT15} propose a learning approach which gradually refines an under- and over-approximation of an input-output transition system representing the system under test.
The correctness of this approach needs several assumptions, e.g., an oracle indicating when, for some trace, all outputs, which extend the trace to a valid system trace, have been observed.

\section{Individual Fairness of Systems Evaluating Humans}\label{sec:fairness}
\Cref{ex:hrWomanIntro} introduces a new application domain for cleanness definitions.
\hrwoman uses an AI system that is supposed to assist her with the selection of applicants for a hypothetical university.
Cleanness of such a system can be related to the fair treatment of the humans that are evaluated by it.
A usable fairness analysis can happen no later than at runtime, since \hrwoman needs to make a timely decision on whether to include the applicant in further considerations.
We describe technical measures that help in mitigating this challenge by providing her with information from an individual fairness analysis in a suitable, purposeful, expedient way.
To this end, we propose a formal definition for individual fairness extending the one by~\cite{dwork2012fairness} and 
based on \fCleannessNDet. We develop a runtime monitor that analyses every output of $\seqSys$ immediately after $\seqSys$'s decision, which strategically searches for unfair treatment of a particular individual by comparing them to relevant hypothetical alternative individuals so as to provide a fairness assessment in a timely manner.

Much like $\seqSys$ is to support \hrwoman, AI systems -- in the broadest sense of the word -- more and more often support human decision makers.
Undoubtedly, such systems should be compliant with applicable law (such as the future European AI Act \cite{eu-0106-2021,ai-act-amendments} or the Washington State facial recognition law \cite{washstatefacerecbill}) and ought to minimise any risks to health, safety or fundamental rights. Sometimes, we cannot mitigate all these risks in advance by technical measures and also some risk-mitigation requires trade-off decisions involving features that are either impossible or difficult to operationalise and formalise. This is why it is essential that a human effectively oversees the system (which is also emphasised by several institutions such as UNESCO \cite{aiunesco} and the European High Level Expert Group \cite{aihleg}). 
\emph{Effective} human oversight, however, is only possible with the appropriate technical measures that allow human overseers to better understand the system at runtime~\cite{DBLP:journals/ai/LangerOSHKSSB21}. 
From a technical point of view, this raises the pressing question of what such technical measures can and ought to look like to actually enable humans to live up to these responsibilities. 
Our contribution is intended to bridge the gap between the normative expectations of law and society and the current reality of technological design.

\subsection{Positioning within Related Research Topics}

Our contribution draws on and adds to three vibrant topics of current research, namely
Explainable AI (XAI), AI fairness, and discrimination.

\paragraph{XAI}
Many of the most successful AI systems today are some kind of black boxes \cite{bathaee2017artificial}.
Accordingly, the field of \enquote{Explainable AI}~\cite{DarpaXAI2016} focuses on the question of how to provide users (and possibly other stakeholders) with more information via several key perspicuity properties~\cite{DBLP:conf/re/SterzBLH21} of these systems and their outputs to make them understand these systems and their outputs in ways necessary to meet various desiderata~\cite{doi:10.1177/2053951716679679,e23010018,DBLP:conf/re/LangerBHHSW21,arrieta2020explainable,nunes2017systematic,9604587}.  
The concrete expectations and promises associated with various XAI methods are manifold. 
Among them are enabling warranted trust in systems~\cite{10.1145/2939672.2939778, DBLP:journals/chb/SchlickerLOBKW21, jacovi2021formalizing, DBLP:conf/re/KastnerLLSSS21}, increasing human-system decision-making performance~\cite{lai2019human} for instance through increasing human situation awareness when operating systems~\cite{sanneman2020situation}, enabling responsible decision-making and effective human oversight~\cite{Baumetal-BAUFRT-2, Mecacci2020-MECMHC, 10.3389/frobt.2018.00015}, as well as identifying and reducing discrimination~\cite{e23010018}. 
It often remains unclear what kind of explanations are generated by the various  explainability methods  and how they are meant to contribute to the fulfilment of the desiderata, even though these questions have become the subject of systematic and interdisciplinary research~\cite{DBLP:journals/ai/LangerOSHKSSB21,DBLP:conf/re/LangerBHHSW21}.

Our approach can be taxonomised along at least two different distinctions~\cite{10.1145/2939672.2939778, DBLP:journals/corr/RibeiroSG16a, DBLP:conf/pkdd/MolnarCB20, DBLP:journals/ai/LangerOSHKSSB21, 10.1145/3531146.3534639}: First, it is \textit{model-agnostic} (not \textit{model-specific}), i.e., it is not tailored to a particular class of models but operates on observable behaviour -- the inputs and outputs of the model. Second, our method is a \textit{local method} (not \textit{global}), i.e., it is meant to shed light on certain outputs rather than the system as a whole. \vspace*{-0.5em}

\paragraph{(Un-)Fair Models}

Fairness, discrimination, justice, equal opportunity, bias, prejudice, and many more such concepts are part of a meaningfully interrelated cluster that has been analysed and dissected for millennia~\cite{Aristot_ne, Aristot_pol}. Many fields are traditionally concerned with the concepts of fairness and discrimination, ranging from philosophy \cite{Aristot_ne, Aristot_pol, 10.2307/2265047,10.1145/3433949,rawls1999theory,10.2307/2265349, rawls2001justice} 
to legal sciences~\cite{Borgesius2020, wachter2020bias, hartmann, Thuesing}, 
to psychology~\cite{hough2001determinants, ziegert2005employment}, 
to sociology~\cite{alves1978should,jewson1986modes}, 
to political theory~\cite{rawls1999theory}, 
to economics~\cite{10.2307/42919257}. 
Nowadays, it has also become a technological topic that calls for cross-disciplinary perspectives~\cite{9445793}.

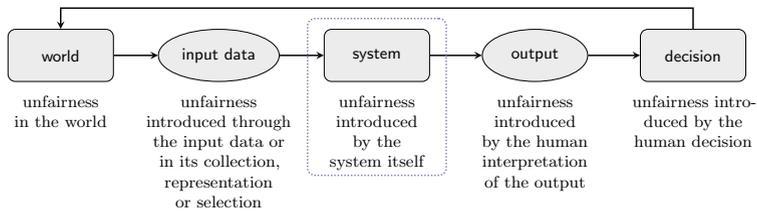
\begin{figure}[t]
    \centering
    \resizebox{0.9\linewidth}{!}{
      \tikzstyle{rectang} = [rectangle, rounded corners, minimum width=2cm, minimum height=1cm,text centered, draw=black, fill=gray!15, font=\sffamily]
\tikzstyle{round} = [ellipse, minimum width=2cm, minimum height=1cm,text centered, draw=black, fill=gray!15, font=\sffamily]
\tikzstyle{myplain} = [draw=none, fill=none, minimum width=2cm, text width=2.5cm, text depth=1.2cm, text centered]
\tikzstyle{arrow} = [thick,->,>=stealth]

\begin{tikzpicture}[node distance=3cm] \small
\node (world)       [rectang]                   {world};
\node (input)       [round, right of=world]     {input data};
\node (system)      [rectang, right of=input]   {system};
%\node (training)    [round, dashed, above=0.4cm of system, text width=1cm]    {training data};
\node (output)      [round, right of=system]    {output};
\node (decision)    [rectang, right of=output]  {decision};

\draw [arrow] (world) -- (input);
\draw [arrow] (input) -- (system);
\draw [arrow] (system) -- (output);
\draw [arrow] (output) -- (decision);

%\draw [arrow, dashed] (training) -- (system);
%\draw [arrow, dashed] (world) -- (training.west);

\draw [arrow] (decision.north) -- ++(0,+0.4) -|  (world.north);

\node (worldEx)     [myplain, below=0.65cm of world.center] {unfairness in the world};
\node (inputEx)     [myplain, text width=2.75cm, below=0.65cm of input.center] {unfairness introduced through the input data or in its collection, representation or selection};
\node (systemEx)    [myplain, text width=2cm, text depth=1.0cm, below=0.65cm of system.center] {unfairness introduced by the system itself};
\node (outputEx)    [myplain, text width=2.75cm, below=0.65cm of output.center] {unfairness introduced by the human interpretation of the output};
\node (decisionEx)  [myplain, below=0.65cm of decision.center] {unfairness introduced by the human decision};

\node (rfmBox) [draw, thick, blue!40!black!50, densely dotted, rounded corners, fit= (system) (systemEx), inner sep = 2mm] {};
%\node (rfmLabel) [align=right, text=blue!40!black!50, font=\sffamily\footnotesize, rotate=90, text width=4cm, anchor=north, below=1cm of rfmBox.north east] {our fairness monitoring};

\end{tikzpicture}
    }
    \vspace{2em}
    \caption{Sketch of different origins of unfairness in a decision process supported by a system; the dotted box indicates which unfairness our monitoring targets.}
    \label{fig:pipeline}
\end{figure}

With regard to fairness, there are two distinctions that are especially relevant to our work. 
First, one distinction is made between \emph{individual fairness}, i.e., that similar individuals are treated similarly~\cite{dwork2012fairness}, and \emph{group fairness}, i.e., that there is adequate group parity~\cite{10.1145/3351095.3372864}. 
Measures of individual fairness 
are often
close to the Aristotelian dictum to treat like cases alike \cite{Aristot_ne, Aristot_pol}. 
In a sense, operationalisations of individual fairness are robustness measures~\cite{DBLP:conf/emsoft/TabuadaBCSM12,DBLP:journals/acta/BloemCGHHJKK14}, but 
instead of requiring robustness with respect to noise or adversarial attacks, measures of individual fairness, such as the one by Dwork et al.~\cite{dwork2012fairness}, call for robustness with respect to highly context-dependent differences between representations of human individuals.
Second, 
recent work from the field of law~\cite{wachter2020bias}
suggests to differentiate between \emph{bias preserving} and \emph{bias transforming} fairness metrics.
Bias preserving fairness metrics seek to avoid adding new bias. For such metrics, historic performances are the benchmarks for models, with equivalent error rates for each group being a constraint.
In contrast, bias transforming metrics do not accept existing bias as a given or neutral starting point, but aim at adjustment. Therefore, 
they require to make a \enquote{positive normative choice}~\cite{wachter2020bias}, 
i.e. to actively decide which biases the system is allowed to exhibit, and which it must not exhibit.

Over the years, many concrete approaches have been suggested to foster different kinds of fairness in artificial systems, especially in AI-based ones~\cite{mehrabi, e23010018, wachter2020bias, zehlike,10.1145/3494672}. 
Yet, to the best of our knowledge, an approach like ours is still missing.
One of the approaches that is closest to ours,
namely that by John et al.~\cite{john}, is not local and therefore not suitable for runtime monitoring. Also, it is not model-agnostic. So, to the best of our knowledge, our approach  
provides a new contribution to the debate on unfairness detection. 

It is important to note/recognise that our approach can only be understood as part of a more holistic approach to preventing or reducing unfairness. After all, there are many sources of unfairness \cite{barocas2016big} (also see \autoref{fig:pipeline} and \autoref{app:pipeline_full}). 
Therefore, not every technical measure is able to detect every kind of unfairness and eliminating one source of unfairness might not be sufficient to eliminate all unfairness. Our approach tackles only unfairness introduced by the system, but not other kinds of unfairness.  

\paragraph{Discrimination}
We understand discrimination as dissimilar treatment of similar cases 
or similar treatment of dissimilar cases without justifying reason.
This is a definition that can also be found in the law \cite[§43]{ecj-c-356-12}.
Our work is exclusively focused on 
discrimination \emph{qua} dissimilar treatment of similar cases.
Discrimination requires a thoughtful and largely not formalisable consideration of \enquote{justifying reason}. 
However, we will exploit the relation of discrimination and fairness: 
Unfairness in a system can arguably be a good proxy of discrimination --
even though not every unfair treatment by a system necessarily constitutes discrimination (especially not in the legal sense).
Thus, a tool that highlights cases of unfairness in a system can be highly instrumental in detecting discriminatory
features of a system. It is not viable, though, to let such a tool rule out unfair treatment fully automatically without human oversight, since there could be justifying reason to treat two similar inputs in a dissimilar way.

\subsection{Individual Fairness}

\hrwoman from Example \ref{ex:hrWomanIntro} should be able to detect individual unfairness. 
An operationalisation thereof by Dwork et al.~\cite{dwork2012fairness} 
is based on
the Lipschitz condition 
to enforce that similar individuals are treated similarly. 
To measure similarity, they
assume the existence of an input distance function $\dIn$ and an output distance function $\dOut$.
This assumption is very similar to the one that we implicitly made in the previous sections for \robustCleannessNDet and \fCleannessNDet.
However, in the case of the fair treatment of humans finding reasonable distance functions is more challenging than it was for the examples in the previous chapters.
Dwork et al.\ assume that both distance functions perfectly measure distances between individuals\footnote{For easier readability, we will not distinguish between \emph{individuals} and their \emph{representations} unless this distinction is relevant in the specific context. It is nevertheless important to note that inputs are not individuals, but only representations of individuals, since an input could inadequately represent an individual and therefore be unfair (also see \autoref{app:pipeline_full}).} and between outputs of the system, respectively, but admit that in practice these distance functions are only approximations of a ground truth at best.
They suggest that distance measures might be learned, 
but there is no one-size-fits-all approach to selecting distance measures. 
Indeed, obtaining such distance metrics is a topic of active research~\cite{zemel2013learning,pmlr-v119-mukherjee20a,ilvento2019metric}.
Additionally, the Lipschitz condition assumes a Lipschitz constant $L$ to establish a linear constraint between input and output distances.

\begin{definition}
    A deterministic sequential program $\seqSys: \Inputs \to \Outputs$ is \emph{Lipschitz-fair} w.r.t. 
    %distance functions 
    $\dIn: \Inputs \times \Inputs \to \RR$, $\dOut: \Outputs \times \Outputs \to \RR$, and a Lipschitz constant $L$,
  if and only if for all $\inp_1, \inp_2 \in \Inputs$, $\dOut(\seqSys(\inp_1), P(\inp_2)) \leq L \cdot \dIn(\inp_1, \inp_2)$.
\end{definition}

Lipschitz-fairness comes with some restrictions that limit its suitability for practical application:

\begin{description}
  \item[$\dIn$-$\dOut$-relation:] High-risk systems are typically complex systems and ask for more complex fairness constraints than the linearly bounded output distances provided by the Lipschitz condition. 
  For example, using the Lipschitz condition prevents us from allowing small local 
  jumps in the output and at the same time forbidding jumps of the same rate of increase over larger ranges of the input space (also see supplementary material in \Cref{app:fbetterL}).

  \item[Input relevance:] 
  The condition quantifies over the entire input domain of a program. This overlooks two things: first, it is questionable whether each input in such a domain is plausible as a representation for a real-world individual. But whether a system is unfair for two implausible and purely hypothetical inputs is largely irrelevant in practice. 
Secondly, it also ignores that mere potential unfair treatment is at most a threat, not necessarily already a harm~\cite{Rowe2022-ROWCAR-4}. Therefore, even with a restriction to only plausible applicants, the analysis might take into account more inputs than needed for many real-world applications. 
What is important in practice is the ability to determine whether \emph{actual} applicants are treated unfairly -- and for this it is often not needed to look at the entire input domain.
  
 \item[Monitorability:]  In a monitoring scenario with the Lipschitz condition in place, a fixed input $\inp_1$ must be compared to potentially all other inputs $\inp_2$.
  Since the input domain of the system can be arbitrarily large, the Lipschitz condition is not yet suitable for monitoring in practice
  (for a related point see John et al.~\cite{john}). 
\end{description}

\paragraph{}
We propose a notion of individual fairness that is based on 
\Cref{def:clean:seq:f:det}.
Instead of cleanness contracts we consider here \emph{fairness contracts}, which are tuples $\fContract = \langle d_\Inputs, d_\Outputs, f \rangle$ containing input and output distance functions and the function $f$ relating input distances and output distances.
Notably, the set of standard inputs $\Norm$ known from cleanness contracts is not part of a fairness contract; it is unknown what qualifies an input to be `standard' in the context of fairness analyses.
Still, our fairness definition evaluates fairness for a set  of individuals $\Is\subseteq \Inputs$ (e.g., a set of applicants), which has conceptual similarities to the set $\Norm$.
A \fairCon specifies certain fairness parameters for a concrete context or situation. 
Such parameters should generally not already include $\Is$ to avoid introducing new unfairness through the monitor by tailoring it to specific inputs individually or by treating certain inputs differently from others. \FfairnessC can thus be defined as follows:

\begin{definition}\label{def:ffairC}
    A deterministic sequential program $\seqSys: \Inputs \rightarrow \Outputs$ is \emph{\ffairC} for a set $\Is \subseteq \Inputs$ of actual inputs w.r.t.
    a  \fairCon $\fContract = \langle d_\Inputs, d_\Outputs, f \rangle$, 
    if and only if 
    for every $\inp \in \Is$ and $\inp' \in \Inputs$, $\dOut(\seqSys(\inp), \seqSys(\inp')) \leq f(\dIn(\inp,\inp'))$.
\end{definition}

The idea behind \ffairnessC is 
that every individual in set $\Is$ is compared to potential other inputs in the domain of $\seqSys$.
These other inputs do not necessarily need to be in $\Is$, nor do these inputs need to have ``physical counterparts'' in the real world.
Driven by the insights of the 
\textit{Input relevance}
restriction of Lipschitz-fairness,
we explicitly distinguish 
inputs in the following
and will call inputs that are given to $\seqSys$ by a user \emph{actual inputs}, denoted
$\realInp$, and call inputs 
to which such $\realInp$ are compared to \emph{synthetic inputs},
denoted $\synInp$.
Actual inputs are typically\footnote{A case where actual inputs might not have real-world counterparts is testing.} inputs that have a real-world counterpart, 
while this might or might not be true for synthetic inputs. 
On first glance, an alternative to using synthetic inputs is to use only actual inputs, e.g., to compare every actual input with every other actual input in $\Is$. 
For example, for a university admission, all applicants could be compared to every other applicant.
However, this would heavily rely on contingencies: the detection of unfair treatment of an applicant depends on whether they were lucky enough that, coincidentally, another candidate has also applied who aids in unveiling the system's unfairness towards them.
Instead, \ffairnessC prefers to over-approximate the set of plausible inputs that actual inputs are compared to rather than under-approximating it by comparing only to other inputs in $\Is$. 
This way, the attention of the human exercising oversight of the system might be drawn to cases that are actually not unfair, but as a competent human in the loop, they will most likely be able to judge that the input was compared to an implausible counterpart. This will usually enable more effective human oversight than an under-approximation that misses to alert the human to unfair cases.

Notice that \ffairnessC is a conservative extension of Lipschitz-fairness. 
With $\Is = \Inputs$ and $f(x) = L \cdot x$, \ffairnessC mimics Lipschitz-fairness. 
Wachter et al.~\cite{wachter2020bias} classify the Lipschitz-fairness of Dwork et al.~\cite{dwork2012fairness} as bias-transforming. As we generalise this and introduce no element that has to be regarded as bias-preserving, our approach arguably is bias-transforming, too.

\FfairnessC, with its function $f$, provides a powerful tool to model complex fairness constraints.
How such an $f$ is defined has profound impact on the quality of the fairness analysis.
A full discussion about which types of functions make a good $f$ go beyond the scope of this article.
A suitable choice for $f$ and the distance functions $\dIn$ and $\dOut$ heavily depends on the context in which fairness is analysed -- there is no one-fits-it-all solution. 
\FfairnessC makes this explicit with the formal \fairCon $\fContract = \langle d_\Inputs, d_\Outputs, f \rangle$.

\subsection{Fairness Monitoring}

\newcommand{\AlgInpSet}{\Is}
\begin{algorithm}[t]
  \caption{\fMinimiser, \\with $\robmin\ S = (\rob, \inp_1, \inp_2)$ only if $(\rob, \inp_1, \inp_2) \in S$ and for all $(\rob', \inp'_1, \inp'_2) \in S$, $\rob' \geq \rob$}
\label{algo:fMinimiser}
\begin{algorithmic}[1]
\item[\textbf{Falsification Parameters:}] $\PS$: Proposal scheme, $\beta$: Temperature parameter
\REQUIRE \Aisystem $P: \Inputs \rightarrow \Outputs$, \FairCon $\fContract = \langle d_\Inputs, d_\Outputs, f \rangle$, and 
set of actual inputs $\AlgInpSet$
\ENSURE A minimal fairness score triple from 
$\RR \times \AlgInpSet \times \Inputs$.
\STATE $\synInp \leftarrow$ any input $\realInp \in \AlgInpSet$
\STATE $(\rob, \inp_\mathsf{min}, \synInp) \leftarrow \robmin \set{(F(\realInp,\synInp), \realInp,\synInp) \ | \ \realInp\in\AlgInpSet}$
\STATE $(\rob_\mathsf{min}, \inp_1, \inp_2) \leftarrow (\rob, \inp_\mathsf{min}, \synInp)$
\WHILE {\NOT timeout}
  \STATE $\synInp' \leftarrow \PS(\synInp,\seqSys(\synInp))$
  \STATE $(\rob', \inp'_\mathsf{min}, \synInp') \leftarrow \robmin \set{(F(\realInp,\synInp'), \realInp, \synInp') \ | \ \realInp\in\AlgInpSet}$  
  \STATE $(\rob_\mathsf{min}, \inp_1, \inp_2) \leftarrow \robmin \set{ (\rob_\mathsf{min}, \inp_1, \inp_2), (\rob', \inp'_\mathsf{min}, \synInp') }$
  \STATE $\alpha \leftarrow \exp(-\beta(\rob' - \rob))$
  \STATE $r \leftarrow \mathsf{UniformRandomReal}(0,1)$
  \IF{$r \leq \alpha$}
    \STATE $\synInp \leftarrow \synInp'$
    \STATE $\rob \leftarrow \rob'$
  \ENDIF
\ENDWHILE
\RETURN $(\rob_\mathsf{min}, \inp_1, \inp_2)$ 
\end{algorithmic}
\end{algorithm}

We develop a probabilistic-falsification-based fairness monitor that, given a set of actual inputs, searches for a synthetic counterexample to falsify a system $\seqSys$ w.r.t. a \fairCon \fContract.
To this end, it is necessary to provide a quantitative description of \ffairnessC that satisfies the characteristics of a robustness estimate.
We call this description \emph{fairness score}.
For an actual input $\realInp$ and a synthetic input $\synInp$ we define the fairness score
as $F(\realInp, \synInp) \coloneqq f(d_\Inputs(\realInp, \synInp)) - d_\Outputs(\seqSys(\realInp),\seqSys(\synInp))$.
$F$ is indeed a robustness estimate function: if $F(\realInp, \synInp)$ is non-negative, then $d_\Outputs(\seqSys(\realInp),\seqSys(\synInp)) \leq f(d_\Inputs(\realInp, \synInp))$, and if it is negative, then $d_\Outputs(\seqSys(\realInp),\seqSys(\synInp)) \not\leq f(d_\Inputs(\realInp, \synInp))$.
For a set of actual inputs $\Is$, the definition generalises to 
$F(\Is, \synInp) \coloneqq \min \set{ F(\realInp, \synInp) \mid \realInp\in\Is }$, i.e., the overall fairness score is the minimum of the concrete fairness scores of the inputs in $\Is$.
Notice that $\mathcal{R}_\Is(\synInp) \coloneqq F(\Is, \synInp)$ is essentially the quantitative interpretation of $\varphi_{\textsf{u-func}}$ (from \Cref{prop:hyperstl:STLfCleannessCorrectness}) after simplifications attributed to the fact that \seqSys is a sequential and deterministic program (cf. \Cref{def:clean:react:f}.\ref{def:clean:react:f:U} vs. \Cref{def:clean:seq:f:det}).

Algorithm~\ref{algo:fMinimiser} shows \fMinimiser, which builds on  Algorithm~\ref{algo:falsification} to search for the minimal fairness score in a \aisystem $\seqSys$ for \fairCon $\fContract$.
The algorithm stores fairness scores in triples that also contain the two inputs for which the fairness score was computed. 
The minimum in a set of such triples is defined by the function $\robmin$ that returns the triple with the smallest fairness score of all triples in the set.
The first line of \fMinimiser initialises the variable $\synInp$ with an arbitrary actual input from $\AlgInpSet$.
For this value of $\synInp$, the algorithm checks the corresponding fairness scores for all actual inputs $\realInp \in \Is$ and stores the smallest one.
In line 3, the globally smallest fairness score triple is initialised.
In line 5 it uses the proposal scheme to get the next synthetic input $\synInp'$.
Line 6 is similar to line 2: for the newly proposed $\synInp'$ it finds the smallest fairness score, stores it, and updates the global minimum if it found a smaller fairness score (line 7).
Lines 8-13 come from Algorithm~\ref{algo:falsification}.
The only difference is that in addition to $\synInp$ we also store the fairness score $\rob$.
Line~4 of \Cref{algo:fMinimiser} differs from \Cref{algo:falsification} by terminating the falsification process after a timeout occurs (similar to the adaptation of \Cref{algo:falsification} in \Cref{sec:dieselSupervision}).
Hence, the algorithm does not (exclusively) aim to falsify the fairness property, but aims at minimising the fairness score;  even if the fair treatment of the inputs in $\Is$ cannot be falsified in a reasonable amount of time, we still learn how robustly they are treated fairly, i.e., how far the least fairly treated individual in $\Is$ is away from being treated unfairly.
After the timeout occurs, the algorithm returns the triple with the overall smallest seen fairness score $\rob_\mathsf{min}$, together with the actual input $\inp_1$ and the synthetic input $\inp_2$ for which  $\rob_\mathsf{min}$ was found. 
In case $\rob_\mathsf{min}$ is negative, $\inp_2$ is a counterexample for $\seqSys$ being \ffairC.

\fMinimiser implements a sound
\fContract-unfairness detection as stated in \Cref{prop:fMinimiserIsSound}.
However, it is not complete, i.e., it is not generally the case that $\seqSys$ is \ffairC for $\Is$ if $\rob$ is positive. 
It may happen that there is a counterexample, but \fMinimiser did not succeed in finding it before the timeout.
This is analogue to results obtained for model-agnostic \robustCleannessNDet analysis~\cite{DBLP:journals/tomacs/BiewerDH21}.

\begin{proposition}
\label{prop:fMinimiserIsSound}
Let $\seqSys: \Inputs \rightarrow \Outputs$ be a deterministic sequential program, $\fContract = \langle d_\Inputs, d_\Outputs, f \rangle$ a \fairCon, and $\Is$ a set of actual inputs. Further, let $(\rob_\mathsf{min}, \inp_1, \inp_2)$ be the result of $\fMinimiser(\seqSys, \fContract, \Is)$.
If $\rob_\mathsf{min}$ is negative, then $\seqSys$ is not \ffairC for $\Is$ w.r.t. \fContract.
\end{proposition}

Moreover, \fMinimiser circumvents major restrictions of the Lipschitz-fairness:
\begin{description}
  \item[$\dIn$-$\dOut$-relation:]
\FfairnessC defines constraints between input and output distances by means of a function $f$, which allows to express also complex fairness constraints.
For a more elaborate discussion, see \Cref{app:fbetterL}.
  \item[Input relevance:] 
  \FfairnessC explicitly distinguishes between actual and synthetic inputs.
  This way, \ffairnessC acknowledges a possible obstacle of the fairness theory when it comes to a real-world usage of the analysis, 
  namely that 
  only some elements of the system's input domain might be plausible and that usually only few of them become actual inputs that have to be monitored for unfairness.
  \item[Monitorability:] \fMinimiser demonstrates that \ffairnessC is monitorable.
  It resolves the quantification over $\Inputs$ using the above concepts from probabilistic falsification using the robustness estimate function $F$ as defined above.
\end{description}

\paragraph{Towards \ffairnessC in the loop}

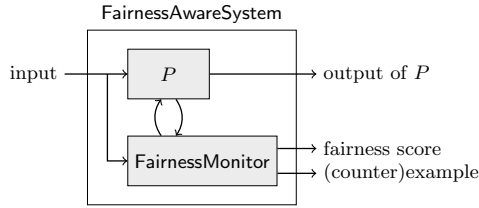
\begin{figure}
    \centering
    \begin{minipage}{6.75cm}
      \resizebox{1\linewidth}{!}{
        \tikzstyle{rectang} = [rectangle,text centered, draw=black, fill=gray!15, minimum width=1.3cm, minimum height=0.8cm]
\tikzstyle{invisible} = [inner sep=0pt, outer sep=0pt]

\begin{tikzpicture}[node distance=1cm] \small

\node (progr) [rectang] {$P$};
\node (minF)  [rectang, below=1.4cm of progr.west, anchor=west] {\fMinimiser};
\node (input) [left=1cm of progr.west] {input}; %{$j$};
\node (outputP) [right=1.7cm of progr.east, anchor=west] {output of $P$}; %{$P(j)$};
\node (outputFS) [below=1.2cm of outputP.west, anchor=west] {fairness score}; %{$fs$};
\node (outputCE) [below=1.6cm of outputP.west, anchor=west] {(counter)example}; %{$(i, P(i))$};

\node (progrAnchorLeft) [invisible, left=0.1cm of progr.south] {};
\node (progrAnchorRight) [invisible, right=0.1cm of progr.south] {};
\node (auxInput) [invisible, left=0.3cm of progr.west] {};
\node (auxOutput) [invisible, right=0.3cm of progr.east] {};

\draw [<-, semithick] (progrAnchorLeft.center) to [out=235,in=125] (progrAnchorLeft.center |- minF.north);% node[midway, left] {}; %{$\cdot$};
\draw [->, semithick] (progrAnchorRight.center)  to [out=305,in=55] (progrAnchorRight.center |- minF.north); %node[midway, right] {}; %{$P(\cdot)$};

%\draw [<-, semithick] (progrAnchorLeft.center) -- (progrAnchorLeft.center |- minF.north) node[midway, left] {};
%\draw [->, semithick] (progrAnchorRight.center) -- (progrAnchorRight.center |- minF.north) node[midway, right, align=left] {search for\\ counterexample};

\draw[->, semithick] (input)--(progr);
\draw[->, semithick] (auxInput.center) |- (minF.west);

\draw [<-, semithick] (outputFS.west) -- (outputFS.west -| minF.east) {};
\draw [<-, semithick] (outputCE.west) -- (outputCE.west -| minF.east) {};

\draw[->, semithick] (progr) -- (outputP);
%\draw [->, semithick] (auxOutput.center) -- (auxOutput.center |- minF.north) {};

\node (box) [rectangle, draw, inner sep=0.3cm, fit=(progr) (minF) (auxInput)] {};
\node [above=0cm of box] {\fWrapper};
\end{tikzpicture}
      }
      \end{minipage}
      \caption{Schematic visualisation of \fWrapper}%
      \label{fig:flow}
    \end{figure}

    \begin{algorithm}[t]
        \caption{\fWrapper}
      \label{algo:fWrapper}
      \begin{algorithmic}[1]
      \item[\textbf{Parameters:}] \Aisystem $\seqSys: \Inputs \rightarrow \Outputs$, \FairCon $\fContract = \langle d_\Inputs, d_\Outputs, f \rangle$
      \REQUIRE Input $\realInp \in \Inputs$
      \ENSURE Tuple of the system output, normalised fairness score, and synthetic values witnessing the fairness score 
      \STATE $(\rob_\mathsf{min}, \realInp, \synInp) \leftarrow \fMinimiser(\seqSys,\fContract,\set{\realInp})$
      \RETURN $\left(\seqSys(\realInp), \rob_\mathsf{min} \div f(d_\Inputs(\realInp, \synInp)), \left(\synInp, \seqSys(\synInp)\right)\right)$  
      \end{algorithmic}
      \end{algorithm}

    If a high-risk \aisystem is in operation, a human in the loop must oversee the correct and fair functioning of the outputs of the system. To do this, the human needs real-time fairness information.
    \Cref{fig:flow} shows how this can be achieved by coupling the system $\seqSys$ and the \fMinimiser in \Cref{algo:fMinimiser} in a new system called \fWrapper.
    \fWrapper is sketched in \Cref{algo:fWrapper}.
    Intuitively, the \fWrapper is a higher-order program that is parameterised with the original program $\seqSys$ and the \fairCon~\fContract.
    When instantiated with these parameters, the program takes arbitrary (actual) inputs $\realInp$ from $\Inputs$.
    In the first step, it does a fairness analysis using \fMinimiser
    with arguments \seqSys, \fContract, and $\set{\realInp}$.
    To make fairness scores comparable, \fWrapper normalises the fairness score $\rob$ received from \fMinimiser by dividing\footnote{\label{ftnote:divisionByZero}For $f$ that can return $0$, there may be a $0 \div 0$ division. The result of this division should be defined depending on the concrete context; 
    reasonable values range from the extreme scores $0$ (to indicate that the score is on the edge to becoming `unfair') to $1$ (to indicate that more fairness is impossible).} it by the output distance limit $f(d_\Inputs(\realInp, \synInp))$.
    For fair outputs, the score will be between 0 (almost unfair) and 1 (as fair as possible).\footnote{Fairness may be a vague concept that cannot be dichotomised. By its choice of the \fairCon parameters, our approach nevertheless specifies a (non-arbitrary) cut-off point at 0; but it does so for purely instrumental and non-ontological reasons.}
    Outputs that are not \ffairC are accompanied by a negative score representing how much the limit $f(d_\Inputs(\realInp, \synInp))$ is exceeded.
    A fairness score of $-n$ means that the output distance of $\seqSys(\realInp)$ and $\seqSys(\synInp)$ is $n+1$ times as high as that limit.
    Finally, \fWrapper returns the triple with $\seqSys$'s output for $\realInp$, the normalised fairness score, and the synthetic input with its output witnessing the fairness score.

\paragraph{Interpretation of monitoring results}
Especially when \fWrapper finds a violation of \ffairnessC, the suitable interpretation and appropriate response to the normalised fairness score proves to be a non-trivial matter that requires 
expertise.

\begin{example}
Instead of using $\seqSys$ from \Cref{ex:hrWomanIntro} on its own, \hrwoman now uses \fWrapper with a suitable fairness contract. 
(Which fairness contracts are suitable is an open research problem, see \emph{\limitsAndChallenges} in \Cref{sec:conclusion}.)
and thereby receive a fairness score along with $\seqSys$'s verdict on each applicant. If the fairness score is negative, she can also take into account the information on the  synthetic counterpart returned by \fWrapper. Among the \numberofapplicants applicants for the PhD program, the monitoring assigns a negative fairness score to three candidates: \alexa, who received a low score, \evgenji, who was scored very highly, and \john, who got an average score. According to their scoring, \alexa would be desk-rejected, while \evgenji and \john would be considered further.

\alexa's synthetic counterpart, let's call him \syntbad, is ranked much higher than \alexa. In fact, he is ranked so high that \syntbad would not be desk-rejected. \hrwoman compares \alexa and \syntbad and finds that they only differ in one respect: \syntbad's graduate university is the one in the official ranking that is immediately \emph{below} the one that \alexa attended. \hrwoman does some research and finds that \alexa's institution is predominantly attended by People of Colour, while this is not the case for \syntbad's institution. Therefore, \fWrapper helped \hrwoman not only to find an unfair treatment of \alexa, but also to uncover a case of potential racial discrimination.

\john's counterpart, \synthiaB, is ranked much lower than him. \hrwoman manually inspects \john's previous institution (an infamous online university), his GPA of 1.8, and his test result with only 13\%. She finds that this very much suggests that \john will not be a successful PhD candidate and desk-rejects him. Therefore, \hrwoman has successfully used \fWrapper to detect a fault in scoring system $\seqSys$ whereby \john would have been treated unfairly in a way that would have been to his advantage. 

\evgenji received a top score, but his synthetic counterpart, \syna, received only an average one. \hrwoman suspects that \evgenji was ranked too highly given his graduate institution, GPA, and test score. However, as he would not have been desk-rejected either way, nothing changes for \evgenji, and the unfairness he was subject to, is not of effect to him. 

The cases of \john and \evgenji share similarities with the configuration in (b) in \Cref{fig:example}, the one of \alexa with (a), and the ones of all other \numberofapplicantsminusthree candidates with (c).
\end{example}

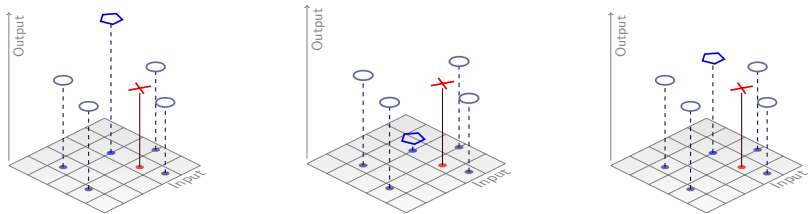
\begin{figure}[t]
    \centering
    \hfill
    \begin{subfigure}[t]{0.3\linewidth}
    {\resizebox{\linewidth}{!}{
      \dimendef\prevdepth=0

%styles of actual input
\tikzstyle{acT} = [cross, minimum size=0.3cm, rotate=20, very thick, draw, red!90!black!100,  minimum size=0.3cm]
\tikzstyle{acB} = [circle, ultra thin, fill, red!90!black!70, inner sep= 0pt, minimum size=1.5mm]
\tikzstyle{acLine} = [red!20!black!100]

%styles of synthetic inputs
\tikzstyle{syT} = [circle, very thick, draw, blue!40!black!60,  minimum size=0.3cm]
\tikzstyle{syB} = [circle, ultra thin, fill, blue!40!black!60, inner sep= 0pt, minimum size=1.5mm]
\tikzstyle{syLine} = [dashed, blue!20!black!100]

%styles of counterexamples
\tikzstyle{ceT} = [regular polygon,regular polygon sides=5, very thick, draw, blue!90!black!100, minimum size=0.2cm]
\tikzstyle{ceB} = [circle, ultra thin, fill, blue!90!black!70, inner sep= 0pt, minimum size=1.5mm]
\tikzstyle{ceLine} = [dashed, blue!20!black!100]

\begin{tikzpicture}[every node/.style={minimum size=0.2cm},on grid, scale=0.8]

\begin{scope}[every node/.append style={
    yslant=0.5,xslant=-1},yslant=0.5,xslant=-1]
 \shade[bottom color=black!2, top color=black!10] (3,3) rectangle +(-3,-3);
  \draw[step=0.6cm, gray] (0,0) grid (3,3);

  \node[align=right, anchor=east, text=gray] (input) at (3,-0.25) {\sffamily Input};
  \node (aux1) at (0,3) {};
  \node[above=4cm of aux1] (aux2) {};
  \node (aux3) at (0,0) {};
    
  %\node at (0,0) {(0,0)};
  %\node at (3,0) {(3,0)};
  %\node at (3,3) {(3,3)};
  %\node at (0,3) {(0,3)};
  \node[acB] (aB) at (2,0.9) {};
  \node[acT, above=2 of aB] (aT) {};
  \draw[acLine] (aB.center) -- (aT);
  
  \node[syB] (s1B) at (2.8,1.2) {};
  \node[syT, above=2.1 of s1B] (s1T) {};
  \draw[syLine] (s1B.center) -- (s1T);
  
  \node[syB] (s2B) at (2.2,0.3) {};
  \node[syT, above=1.8 of s2B] (s2T) {};
  \draw[syLine] (s2B.center) -- (s2T);
 
  \node[syB] (s3B) at (0.5,1) {};
  \node[syT, above=2.1 of s3B] (s3T) {};
  \draw[syLine] (s3B.center) -- (s3T);
  
  \node[syB] (s4B) at (0.8,2.1) {};
  \node[syT, above=2.2 of s4B] (s4T) {};
  \draw[syLine] (s4B.center) -- (s4T); 
  
  \node[ceB] (s5B) at (2,1.8) {};
  \node[ceT, above=3.4 of s5B] (s5T) {};
  \draw[ceLine] (s5B.center) -- (s5T);
 
\end{scope}
 
  \draw[->, gray] (aux1) -- (aux2) node[pos=0.84
  , rotate=90, below] {\sffamily Output};
  %\node[below=0.5cm of aux3] {(1)};
  
\end{tikzpicture}
    }}
    \subcaption{case of unfairness where input is treated worse than relevant counterpart}
    \end{subfigure}
    \hfill
    \begin{subfigure}[t]{0.3\linewidth}
    {\resizebox{\linewidth}{!}{\dimendef\prevdepth=0

    %styles of actual input
\tikzstyle{acT} = [cross, minimum size=0.3cm, rotate=20, very thick, draw, red!90!black!100,  minimum size=0.3cm]
\tikzstyle{acB} = [circle, ultra thin, fill, red!90!black!70, inner sep= 0pt, minimum size=1.5mm]
\tikzstyle{acLine} = [red!20!black!100]

%styles of synthetic inputs
\tikzstyle{syT} = [circle, very thick, draw, blue!40!black!60,  minimum size=0.3cm]
\tikzstyle{syB} = [circle, ultra thin, fill, blue!40!black!60, inner sep= 0pt, minimum size=1.5mm]
\tikzstyle{syLine} = [dashed, blue!20!black!100]

%styles of counterexamples
\tikzstyle{ceT} = [regular polygon,regular polygon sides=5, very thick, draw, blue!90!black!100, minimum size=0.2cm]
\tikzstyle{ceB} = [circle, ultra thin, fill, blue!90!black!70, inner sep= 0pt, minimum size=1.5mm]
\tikzstyle{ceLine} = [dashed, blue!20!black!100]

    \begin{tikzpicture}[every node/.style={minimum size=0.2cm},on grid, scale=0.8]
    
    \begin{scope}[every node/.append style={
        yslant=0.5,xslant=-1},yslant=0.5,xslant=-1]
     \shade[bottom color=black!2, top color=black!10] (3,3) rectangle +(-3,-3);
      \draw[step=0.6cm, gray] (0,0) grid (3,3);

      \node[align=right, anchor=east, text=gray] (input) at (3,-0.25) {\sffamily Input};
      \node (aux1) at (0,3) {};
      \node[above=4cm of aux1] (aux2) {};
      \node (aux3) at (0,0) {};
        
      %\node at (0,0) {(0,0)};
      %\node at (3,0) {(3,0)};
      %\node at (3,3) {(3,3)};
      %\node at (0,3) {(0,3)};
      
      \node[acB] (aB) at (2,0.9) {};
  \node[acT, above=2 of aB] (aT) {};
  \draw[acLine] (aB.center) -- (aT);
  
  \node[syB] (s1B) at (2.8,1.2) {};
  \node[syT, above=2.1 of s1B] (s1T) {};
  \draw[syLine] (s1B.center) -- (s1T);
  
  \node[syB] (s2B) at (2.2,0.3) {};
  \node[syT, above=1.8 of s2B] (s2T) {};
  \draw[syLine] (s2B.center) -- (s2T);
 
  \node[syB] (s3B) at (0.5,1) {};
  \node[syT, above=2.1 of s3B] (s3T) {};
  \draw[syLine] (s3B.center) -- (s3T);
  
  \node[syB] (s4B) at (0.8,2.1) {};
  \node[syT, above=2.2 of s4B] (s4T) {};
  \draw[syLine] (s4B.center) -- (s4T); 
      
      \node[ceB] (s5B) at (2,1.8) {};
      \node[ceT, above=0.3 of s5B] (s5T) {};
      \draw[ceLine] (s5B.center) -- (s5T);
     
    \end{scope}
     
      \draw[->, gray] (aux1) -- (aux2) node[pos=0.84
      , rotate=90, below] {\sffamily Output};
      %\node[below=0.5cm of aux3] {(1)};
      
    \end{tikzpicture}}}
    \subcaption{case of unfairness where input is treated better than relevant counterpart}
    \end{subfigure}
    \hfill
    \begin{subfigure}[t]{0.3\linewidth}
    {\resizebox{\linewidth}{!}{
      \dimendef\prevdepth=0

%styles of actual input
\tikzstyle{acT} = [cross, minimum size=0.3cm, rotate=20, very thick, draw, red!90!black!100,  minimum size=0.3cm]
\tikzstyle{acB} = [circle, ultra thin, fill, red!90!black!70, inner sep= 0pt, minimum size=1.5mm]
\tikzstyle{acLine} = [red!20!black!100]

%styles of synthetic inputs
\tikzstyle{syT} = [circle, very thick, draw, blue!40!black!60,  minimum size=0.3cm]
\tikzstyle{syB} = [circle, ultra thin, fill, blue!40!black!60, inner sep= 0pt, minimum size=1.5mm]
\tikzstyle{syLine} = [dashed, blue!20!black!100]

%styles of counterexamples
\tikzstyle{ceT} = [regular polygon,regular polygon sides=5, very thick, draw, blue!90!black!100, minimum size=0.2cm]
\tikzstyle{ceB} = [circle, ultra thin, fill, blue!90!black!70, inner sep= 0pt, minimum size=1.5mm]
\tikzstyle{ceLine} = [dashed, blue!20!black!100]

\begin{tikzpicture}[every node/.style={minimum size=0.2cm},on grid, scale=0.8]

\begin{scope}[every node/.append style={
    yslant=0.5,xslant=-1},yslant=0.5,xslant=-1]
 \shade[bottom color=black!2, top color=black!10] (3,3) rectangle +(-3,-3);
  \draw[step=0.6cm, gray] (0,0) grid (3,3);

  \node[align=right, anchor=east, text=gray] (input) at (3,-0.25) {\sffamily Input};
  \node (aux1) at (0,3) {};
  \node[above=4cm of aux1] (aux2) {};
  \node (aux3) at (0,0) {};
    
  %\node at (0,0) {(0,0)};
  %\node at (3,0) {(3,0)};
  %\node at (3,3) {(3,3)};
  %\node at (0,3) {(0,3)};
  
  \node[acB] (aB) at (2,0.9) {};
  \node[acT, above=2 of aB] (aT) {};
  \draw[acLine] (aB.center) -- (aT);
  
  \node[syB] (s1B) at (2.8,1.2) {};
  \node[syT, above=2.1 of s1B] (s1T) {};
  \draw[syLine] (s1B.center) -- (s1T);
  
  \node[syB] (s2B) at (2.2,0.3) {};
  \node[syT, above=1.8 of s2B] (s2T) {};
  \draw[syLine] (s2B.center) -- (s2T);
 
  \node[syB] (s3B) at (0.5,1) {};
  \node[syT, above=2.1 of s3B] (s3T) {};
  \draw[syLine] (s3B.center) -- (s3T);
  
  \node[syB] (s4B) at (0.8,2.1) {};
  \node[syT, above=2.2 of s4B] (s4T) {};
  \draw[syLine] (s4B.center) -- (s4T); 
  
  \node[ceB] (s5B) at (2,1.8) {};
  \node[ceT, above=2.4 of s5B] (s5T) {};
  \draw[ceLine] (s5B.center) -- (s5T);
 
\end{scope}
 
  \draw[->, gray] (aux1) -- (aux2) node[pos=0.84
  , rotate=90, below] {\sffamily Output};
  %\node[below=0.5cm of aux3] {(1)};
  
\end{tikzpicture}
    }}
    \subcaption{case of no detected unfairness}
    \end{subfigure}
    \hfill
    \caption{Exemplary illustration of configurations of an input (red cross) and its synthetic counterparts (grey circles) and the synthetic counterpart with the minimal fairness score (blue polygon); with a two-dimensional input space (grid) and a one-dimensional output.}
    \label{fig:example}
\end{figure}

If our monitor finds only a few problematic cases in a (sufficiently large and diverse) set of inputs, our monitoring helps \hrwoman from our running example by drawing her attention to cases that require special attention. Thereby, individuals who are judged by the system have a better chance of being treated fairly, since even rare instances of unfair treatment are detected. If, on the other hand, the number of problematic cases found is large, or \hrwoman finds especially concerning cases or patterns,
this can point to larger issues within the system. 
In these cases, \hrwoman should take appropriate steps and make sure that the system is no longer used until clarity is established why so many violations or concerning patterns are found. If the system is found to be systematically unfair, it should arguably be removed from the decision process. A possible conclusion could also be that the system is unsuitable for certain use cases, e.g., for the use on individuals from a particular group. Accordingly, it might not have to be removed altogether but only needs to be restricted such that problematic use cases are avoided. 
In any case, significant findings should also be fed back to developers or deployers of the potentially problematic system.
A fairness monitoring such as in \fWrapper or a fairness analysis as in \fMinimiser could also be useful to developers, regulating authorities, watchdog organisations, or forensic analysts as it helps them to check the individual fairness of a system in a controlled environment.

\section{Interdisciplinary Assessment of Fairness Monitoring}\label{sec:assessment}
Regulations for car related emissions are in force for a considerable amount of time, thus, its legal interpretation is mostly clear.
In case of human oversight of AI systems, the AI act is new and parts of it are legally ambiguous.
This raises the question of whether our approach meets requirements that go beyond pre-theoretical deliberations.
Even though comprehensive analyses would go far beyond the scope of this paper, we will nevertheless assess some key normative aspects in philosophical and legal terms, and also briefly turn to the related empirical aspects, especially from psychology.

\subsection{Psychological assessment}

Fairness monitoring promises various advantages in terms of human-system interaction in application contexts -- provided it is extended by an adequate user interface -- which call for empirical tests and studies.
We will only discuss a possible benefit that closely aligns with the current draft of the AI Act: our approach may support effective human oversight. Two central aspects of effective oversight are situation awareness and warranted trust. Our method highlights unfairness in outputs which can be expected to increase users' situation awareness (i.e., \enquote{the perception of the elements in the environment within a volume of time and space, the comprehension of their meaning and the projection of their status in the near future}~\cite[p.~36]{Endsley1995Toward}), which is a variable central for effective oversight~\cite{Endsley2017}. In the minimal case, this allows users to realise that something requires their attention and that they should check the outputs for plausibility and adequacy. In the optimal case and after some experience with the monitor, it may even allow users to predict instances where a system will produce potentially unfair outputs. In any case, the monitoring should enable them to understand limitations of the system and to feed back their findings to developers who can improve the system. 
This leads us to warranted trust, which includes that users are able to adequately judge when  
to rely on system outputs and when to reject them~\cite{Lee2004,jacovi2021formalizing}. Building warranted trust strongly depends on users being able to assess system trustworthiness in the given context of use \cite{schlicker2021towards,Lee2004}. 
According to their theoretical model on trust in automation, Lee and See \cite{Lee2004} propose that trustworthiness relates to different facets of which performance (e.g., whether the system performs reliably with high accuracy) and process (e.g., knowing how the system operates and whether the system's decision-processes help to fulfil the trustor's goals) are especially relevant in our case. 
Specifically, fairness monitoring should enable users to more accurately judge system performance (e.g., by revealing possible issues with system outputs) and system processes (e.g., whether the system's decision logic was appropriate). In line with Lee and See's propositions, this should provide a foundation for users to be better able to judge system trustworthiness and should thus be a promising means to promote warranted trust.
In consequence, our monitoring provides a needed addition to high-risk use contexts of AI because it offers information enabling humans to more adequately use AI-based systems in the sense of possibly better human-system decision performance and with respect to user duties as described in the AI Act. 

\subsection{Philosophical assessment}

More effective oversight promises more informed decision-making. This, in turn, enables morally better decisions and outcomes, since humans can morally ameliorate outcomes in terms of fairness 
and can see to it that moral values are promoted.
Also, fairness monitoring  helps in safeguarding fundamental democratic values if it is applied to potentially unfair systems which are used in certain societal institutions of a high-risk character such as courts or parliaments. It could, for example, make AI-aided court decisions more transparent and promote equality before the law. 
However, since our approach requires finding context-appropriate and morally permissible parameters for \fContract, 
moral requirements arise to enable the finding of such parameters. This not only affects, e.g., developers of such systems, but also those who are in a position to enforce that adequate parameters are chosen, such as governmental authorities, supervising institutions or certifiers.

Apart from that, various parties have arguably a legitimate interest in adequately ascribing moral responsibility for the outcomes of certain decisions to human deciders~\cite{Baumetal-BAUFRT-2}
-- regardless of whether the decision making process is supported by a system. Adequately ascribing moral responsibility is not always possible, though. One precondition for moral responsibility is that the agent had sufficient epistemic access to the consequences of their doing~\cite{sep-moral-responsibility, sep-computing-responsibility}, i.e., that they have enough and sufficiently well justified beliefs about the results of their decision. 
Someone overseeing a university selection process (like \hrwoman) should, for example, have sufficiently well justified beliefs that, at the very least, their decisions do not result in more unfairness in the world. If the admission process is supported by a black-box  \aisystem, though, \hrwoman cannot be expected to have any such beliefs since she lacks insight in the fairness of the system. Therefore, adequate responsibility ascription is usually not possible in this scenario. Our monitoring alleviates this problem by providing the decider with better epistemic access to the fairness of the system. 

\fWrapper helps in making \hrwoman's role in the decision process significant and not only that of a mere button-pusher. \fWrapper makes it possible for her to fulfil some of the responsibilities and duties plausibly associated with her role. For example, she can now be realistically expected to not only detect, but resolve at least some cases of apparent unfairness competently (although she may need additional information to do so). In this respect, she should not be \enquote*{automated away} (cf.~\cite{Matthias2004-MATTRG}).

\subsection{Legal assessment}
A central legislative debate of our time is how to counter the risks AI systems can  pose to the health and safety or fundamental rights of natural persons. 
Protective measures must be taken at various levels: First, before being permitted on the market, it must be ensured \textit{ex ante} that such high-risk AI-systems are in conformity with mandatory requirements\footnote{The specific risks set by AI-systems may also give reason to consider an adaptation and expansion of European legal frameworks such that an even broader prohibition of discrimination (cf.\ Appendix \ref{app:legal:discrimination}) is set into place.} regarding safety and human rights.
This means in particular that the selection of the properties which a system should exhibit requires a positive normative choice 
and should not simply replicate biases present in the status quo \cite{wachter2020bias}. In addition, AI-systems must be designed and developed in such a way that natural persons can oversee their functioning. For this purpose, it is necessary for the provider to identify appropriate human oversight measures before its placing on the market or putting into service. In particular, such measures should guarantee that the natural persons to whom human oversight has been assigned have the necessary competence, training and authority to carry out that role \cite[recital 48]{eu-0106-2021}\cite[Art. 14 (5)]{ai-act-amendments}.

Second, during runtime, the proper functioning of high-risk AI systems, which have been placed on the market lawfully, must be ensured. 
To achieve this goal, a bundle of different measures is needed, ranging from legal obligations to implement and perform meaningful oversight mechanisms to user training and awareness in order to counteract \enquote*{automation bias}.
Furthermore, the AI Act proposal requires deployers to inform the  provider or distributor and suspend the use of the system when they have identified any serious incidents or any malfunctioning \cite[Art.\ 29(4)]{eu-0106-2021,ai-act-amendments}.

Third, and \emph{ex post}, providers must act and take the necessary corrective actions as soon as they become aware, e.g. through information provided by the deployer, that the high-risk system does not (or no longer) meet the legal requirements \cite[Art. 16(g)]{eu-0106-2021,ai-act-amendments}. To this end, they must establish and document a system of monitoring that is proportionate to the type of AI technology and the risks of the high-risk AI system \cite[Art. 61(1)]{eu-0106-2021,ai-act-amendments}.

Fairness monitoring can be helpful in all three of the above respects. Therefore, we argue that there is even a legal obligation to use technical measures such as the method presented in this paper if this is the only way to ensure effective human oversight.

\section{Conclusion \& Future Work}\label{sec:conclusion}

This articles brings together software doping theory and probabilistic falsification techniques.
To this end, it proposes a suitable HyperSTL semantics and characterises \robustCleannessNDet and \fCleannessNDet as HyperSTL formulas and, for the special case of finite standard behaviour, STL formulas.
Software doping techniques have been extensively applied to the tampered diesel emission cleaning systems; this article continues this path of research by demonstrating how testing of real cars can become more effective.
For the first time, we apply software doping techniques to high-risk (AI) systems.
We propose a runtime fairness monitor to promote effective human oversight of high-risk systems.
The development of this monitor is complemented by an interdisciplinary evaluation from a psychological, philosophical, and legal perspective.

\paragraph{\limitsAndChallenges}

A challenge to those employing \robustCleannessNDet or \fCleannessNDet analysis is the selection of suitable parameters, especially $d_\Inputs$, $d_\Outputs$, and $f$ or $\inpbound$ and $\outpbound$. Because of their high degree of context sensitivity, there are no paradigmatic candidates for them that one can default to.  Instead, they have to be carefully selected with the concrete system, the structure of input data and the situation of use in mind.

Reasonable choices for \robustCleannessNDet analysis of diesel emissions have been proposed in recent work~\cite{DBLP:journals/tomacs/BiewerDH21,DBLP:journals/lmcs/BiewerDFGHHM22}.
With respect to individual fairness analysis, potential systems to which \fWrapper or \fMinimiser can be applied to are still too diverse to give recommendations for the contract parameters.
Obviously, further technical limitations include that $f$, $d_\Inputs$, and $d_\Outputs$ must be computable.

With a particular regard to fairness analysis, we identify also non-technical limitations.
As seen in \Cref{fig:pipeline}, our fairness monitoring aims to uncover a particular kind of unfairness, namely individual unfairness that originates from within the \aisystem. 
This excludes all kinds of group unfairness as well as unfairness from sources other than the system. 
Another limitation is the human's competence to interpret the system outputs. 
Even though this is not a limitation that is inherent to our approach, it nevertheless will arguably be relevant in some practical cases, and an implementation of the monitoring always has to happen with the human in mind. For example, the design of the tool should avoid creating the false impression that the system is proven to be fair for an individual if no counterexample has been found. Interpretations like this could lead to inflated judgements of system trustworthiness and eventually to overtrusting system outputs \cite{schlicker2021towards, schlicker2022trustworthiness}. 
Also, it might be reasonable to limit access to the monitoring results: if individuals who are processed by the system have full access to their fairness analysis, they could use this to \enquote*{game} the system, i.e. they could use the synthetic inputs to slightly modify their own input such that they receive a better outcome. 
While more transparency for the user is generally desirable, this has to be kept in mind to avoid introducing new unfairness on a meta-level.

\paragraph{Future Work}
The probabilistic falsification technique we use in this article can be seen as a modular framework that consists of several interchangeable components.
One of these components is the optimisation technique used to find the input with minimal robustness value.
\Cref{algo:falsification} uses a simulated annealing technique~\cite{chib1995understanding,DBLP:books/lib/Rubinstein81}, but other techniques have been proposed for temporal logic falsification, too~\cite{DBLP:conf/hybrid/SankaranarayananF12,5675195}.
We want to further look into such alternative optimisation techniques and to evaluate if they offer benefits w.r.t. cleanness falsification.

Finally, the fairness monitoring approach has been presented using a toy example. It is not claimed to be readily applicable to real-life scenarios. 
Besides the future work that has already been mentioned throughout the paper, we are planning on 
various extensions of our approach, and are working on an implementation that will allow us to integrate the monitoring into a real system. Moreover, we plan to test the possible benefits and shortcomings of the approach in user studies where decision-makers are tasked to make hiring decisions with and without the fairness monitoring approach.
Further 
work will encompass activities such as the improvement and embedding of the algorithm \fWrapper into a proper tool that can be used by non-computer-scientists, and the extension of the monitoring technique to cover more types of unfairness. 
For example, logging the output of the fairness monitor could be used to identify groups that are especially likely to be treated unfairly by the system: 
The individual fairness verdicts provided by \fWrapper and \fMinimiser may also be logged and considered for further fairness assessments or other means of quality assurance of system $P$.
Statistical analysis might unveil that individuals of certain groups are treated unfairly more frequently than individuals from other groups. Depending on the distinguishing features of the evaluated group, this can uncover problems in $P$, especially if protected attributes, such as gender, race, age, etc, are taken into account. Thereby, system fairness can be assessed for protected attributes without including them in the input of $P$, which should generally be avoided, and even without disclosing them to the human in the loop.
By evaluating the monitoring logs from sufficiently many diverse runs of \fWrapper, our local method can be lifted such that it resembles a global method for many practical applications, i.e. we can make statistical statements about the general fairness of $P$. Such an evaluation can also be used to extract prototypes and counterexamples in the spirit of Been et al.~\cite{10.5555/3157096.3157352} illustrating the \emph{tendency} to judge unfairly. 
This is an interesting combination of individual and group fairness that we want to look into further.
Other insights from the research on  reactive systems  \cite{DBLP:conf/esop/DArgenioBBFH17,DBLP:journals/tomacs/BiewerDH21,DBLP:journals/lmcs/BiewerDFGHHM22} can potentially be used to further enrich the monitoring. 
Finally, various disciplines have to join forces to resolve highly interdisciplinary questions such as what constitutes reasonable and adequate choices for $f$, $d_\Inputs$, and $d_\Outputs$ in given contexts of application.

\backmatter

\section*{Declarations}

\bmhead{Funding}
This work is partially funded by DFG grant 389792660 as part of TRR~248 -- CPEC (see \url{https://perspicuous-computing.science}) and by 
VolkswagenStiftung as part of grants AZ 98514, 98513 and 98512 EIS -- Explainable Intelligent Systems (see \url{https://explainable-intelligent.systems}).

\bmhead{Competing interests}
The authors have no competing interests to declare that are relevant to the content of this article.

\bmhead{Data Availability}
The datasets analysed during the current study are available in a Zenodo repository~\cite{biewer_sebastian_2023_8058770} (\url{https://zenodo.org/record/8058770}).

\bmhead{Ethics approval}
Not applicable

\bmhead{Consent to participate}
Not applicable

\bmhead{Consent for publication}
Not applicable

\bmhead{Code availability}
Not applicable

\bmhead{Authors' contributions}
Not applicable

\begin{appendices}

\section{Technical Appendix}\label{app:fbetterL}

This appendix 
illustrates that \ffairnessC is 
more expressive than Lipschitz-fairness and why this is useful. 
For this, we use as a toy example a very simple, hypothetical HR scoring system that aggregates five scores given to the candidates. 
We remark that the whole scenario, the implementation of the system, the choice of distance functions and $f$, is likely not applicable for real-life situations; everything is picked so that our explanations are understandable.

Suppose that certain qualities and characteristics of the applicants are pre-scored by other systems on a scale from $0$ to $\SI{100}{\percent}$, where $0$ means that the candidate is utterly unsuitable for the job in a certain regard, while a scoring of $\SI{100}{\percent}$ means that the candidate is perfect for the job in this regard. 
In particular, we will assume that the following marks are given to each applicant: an \emph{education mark} for how well they are academically suitable for the job, an \emph{experience mark} for how well their previous work experience fits the job, a \emph{personality mark} for their personal and social skills, a \emph{mental ability mark} for what is colloquially referred to as an applicant's general intelligence, and, finally, a \emph{skill mark} that tracks the special skills that applicants have which might be beneficial for the job, such as their knowledge of foreign languages.

The system $P$ that is of interest for us in this example is the one that aggregates all of these marks and gives out an overall score of how well the candidate is suited for the job. 
The human responsible for the hiring process 
can use this in her hiring decision, e.g., she can focus on the top-scoring candidates and choose among them. 

Let $\markSet = [0, 1] \subseteq \RR$ be the reals between $0$ and $1$.
Each of the five marks mentioned above is a real number from set $\markSet$.
The input domain $\Inputs = \markSet^5$ for the sketched HR system is a tuple of five marks.
The output of the system is the overall suitability score of an applicant, which is also a value from \markSet.
The distance between two inputs is defined as the euclidean distance, normalised to a value between $0$ and $1$, i.e., 
\begin{align*}
     & d_\Inputs\big((\eduMark_1, \expMark_1, \personMark_1, \intelliMark_1, \skillMark_1), (\eduMark_2, \expMark_2, \personMark_2, \intelliMark_2, \skillMark_2) \big) = \\
     & \qquad  \sqrt{\frac{(\eduMark_1 - \eduMark_2)^2 + (\expMark_1 - \expMark_2)^2 + (\personMark_1 - \personMark_2)^2 + (\intelliMark_1 - \intelliMark_2)^2 + (\skillMark_1 - \skillMark_2)^2}{5}},
\end{align*}
where $\eduMark$ represents the education mark, $\expMark$ the experience mark, $\personMark$ the personality mark, $\intelliMark$ the mental ability mark, and $\skillMark$ the skill mark of an applicant.
The distance between two outputs $d_\Outputs(\outp_1, \outp_2) = \abs{\outp_1 - \outp_2}$ is the absolute difference between the overall scores $\outp$ and $\outp'$.
Note that also output distances are values between $0$ and $1$.

Our scoring system is a function $P: \markSet^5 \rightarrow \markSet$.
We will assume here that $P$ is defined as the sum of five subscoring systems, one for each of the five input marks, computing a value between $0$ and $0.2$.
Then, 
\[ P((\eduMark, \expMark, \personMark, \intelliMark, \skillMark)) \coloneqq P_\eduMark(\eduMark) + P_\expMark(\expMark) + P_\personMark(\personMark) + P_\intelliMark(\intelliMark) + P_\skillMark(\skillMark).\]

Let $P_\eduMark$, $P_\expMark$, $P_\personMark$ and  $P_\intelliMark$ be defined according to the plot shown in Fig.~\ref{fig:subscoring} a). 
With an increasing mark, these subscores increases up to an input mark of $0.8$, whereafter the applicant becomes overqualified and the subscore slowly decreases.
$P_\skillMark$ is depicted in Fig.~\ref{fig:subscoring} b): The skill mark is less important, however a minimum amount of skills is required for the job.
Hence, there is a jump of the skill score at an skill mark of roughly $0.19$.
Let \john be an applicant with $\eduMark = \expMark = \personMark = \intelliMark = 0.5$ and a skill mark of $\skillMark = 0.2$, which maps to a skill score on the plateau after the jump.
The subscores for education, experience, personality and mental ability mark are $0.12$ each.
The skill score computed for \john is $0.05$.
Hence, \john's overall score is $P(\john) = 4 \cdot 0.12 + 0.05 = 0.53$.
Let \synthiaA be a synthetic applicant with the same marks as \john, except for the skill mark, which is $0.19$ in \synthiaA's case.
As depicted in Fig.~\ref{fig:subscoring} b), the skill subscore for skill mark $0.19$ is $0.02$ -- \synthiaA is at the plateau right before the jump of the skill score.
Her overall score is $P(\synthiaA) = 4 \cdot 0.12 + 0.02 = 0.50$.
The input distance between \john and \synthiaA is $d_\Inputs(\john,\synthiaA) = \sqrt{\frac{0.01^2}{5}} \approx 0.0045$ and the output distance is $d_\Outputs(\john,\synthiaA) = \abs{0.53-0.5} = 0.03$.
It is easy to see that if we use Lipschitz-fairness, the Lipschitz constant $L$ must be at least $L = 6.7$ to allow the small jump in the skill subscoring function.
We argue that small jumps like those in the skill subscore are normal behaviour and, hence, fair.
Assume for the remainder of this example that we use Lipschitz-fairness with $L=6.7$. 

\begin{figure}
  \includegraphics[width=\textwidth]{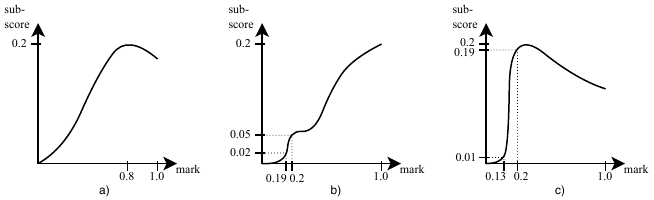}
  \caption{Visualisation of subscoring functions mapping marks to subscores}%
  \label{fig:subscoring}
\end{figure}

Consider now a slightly modified variant $P'$ of $P$.
$P'$ is as $P$ but uses a different subscoring function $P'_\skillMark$ for the skill score.
Fig.~\ref{fig:subscoring} c) shows the skill subscoring function for $P'$.
$P'_\skillMark$ has a jump at skill mark $0.13$ that is significantly larger than that in $P_\skillMark$.
We assume in this example that such a big jump is unfair. This assumption is warranted since, for many applications, such a small change in technical skills which has an immense impact on the skill subscore is not reasonable.  
Considering applicant \john, his skill mark still maps to a very high skill score of $0.19$.
Let \synthiaB be a third (potentially synthetic) applicant with $\eduMark = \expMark = \personMark = \intelliMark = 0.5$ (as for \john and \synthiaA) and $\skillMark = 0.13$. 
Her skill mark maps to a very small skill score of $0.01$.
The overall scores are $P'(\john) = 4 \cdot 0.12 + 0.19 = 0.67$ and $P'(\synthiaB) = 4 \cdot 0.12 + 0.01 = 0.49$.
The input distance is $d_\Inputs(\john,\synthiaB) = 0.0313$ and the output distance is $d_\Outputs(\john, \synthiaB) = 0.18$.
Applying the Lipschitz condition to $P'$ and $d_\Inputs(\john,\synthiaB)$, it easy to see that $d_\Outputs(\john, \synthiaB)$ may become as large as $0.21$.
Hence, $P'$ is classified as fair w.r.t. the Lipschitz condition.
We see that a problem of the Lipschitz condition is that it is not possible to allow small jumps and at the same time disallow large jumps with equal increasing rate.
This is because the distance of the inputs can only be used to multiply it with the Lipschitz constant.

\begin{wrapfigure}[5]{R}{0.52\textwidth}
  \vspace{-1.0cm}
  \begin{minipage}{0.36\textwidth}
    \begin{align*}
      f(d) = \label{eq:f}
      \begin{cases}
        0.001 + 8d, & \text{ for }d\in [0.0, 0.01]\\
        0.001 + 4d, & \text{ for }d\in (0.01, 0.1]\\
        0.001 + 2d, & \text{ for }d\in (0.1, 1.0]\\
      \end{cases}
    \end{align*} 
  \end{minipage}
\end{wrapfigure}
\FfairnessC is different in this regard. 
Function $f$ receives the input distance and can freely define a bound on output distances based on the input distance.
Indeed, the concrete $f$ on the right overcomes the problem observed in the example. 
It uses the input distance for a case distinction on the magnitude of the input distance.
For input distances up to $0.01$, $f$ effectively applies Lipschitz-fairness with $L=8$ to allow small jumps.
For input distances between $0.01$ and $0.1$, $f$ behaves like Lipschitz-fairness for $L=4$, and for larger input distances, it enforces $L=2$.
In all cases we add $0.001$ to the result to avoid $f$ becoming zero (see footnote~\ref{ftnote:divisionByZero} on page~\pageref{ftnote:divisionByZero} in the main paper).
Applying \ffairnessC with $\Contract = \langle d_\Inputs, d_\Outputs, f \rangle$ to $P$, the combination of \john and \synthiaA (and hence the small jump of the skill score function) is not highlighted by \fWrapper, i.e., it is correctly detected as \ffairC. 
Applied to $P'$, however, \john and \synthiaB fall into the second case in the definition of $f$, but, as the emulated Lipschitz condition with $L=4$ is violated, \fWrapper likely finds a negative fairness score, i.e., $P'$ is not \ffairC w.r.t. \john.
We remark that we propose this $f$ for purely illustrative purposes. 
For real-world examples, $f$ should be more sophisticated. % (and ideally should not contain jumps).
Finding a suitable $f$ can be a non-trivial task which hinges on various aspects that are crucial for the fairness evaluation in a given context.
Clearly, the $P$ and $f$ provided in this illustration are toy examples that are probably inappropriate for real-world usage.

\subsection{Proofs}
\label{sec:appendix:proofs}

In this section, we will provide proofs for most of the propositions and theorems in the main paper.
First, we show the correctness of the HyperSTL characterisations of \robustCleannessNDet and \fCleannessNDet.

We first provide a lemma, which  destructs the globally ($\G$) and weak until ($\W$) operators such that the timing constraints encoded by these operators becomes explicit.

\begin{lemma}
  \label{lemma:testing:stlDerivedOps}
  Let $\sigma: \timeSet \to X$ be a trace with $\timeSet = \NN$ or $\timeSet = \RRpos$ and let $\phi$ and  $\psi$ be STL formulas. Then the following equivalences hold.
  \begin{enumerate}
    \item \label{lemma:testing:stlDerivedOps:globally} $\sigma, 0 \models \G \phi$ if and only if $ \forall \varTime \geq 0.\; \sigma, \varTime \models \phi$,
    \item \label{lemma:testing:stlDerivedOps:wUntil} if $\timeSet = \NN$, then $\sigma, 0 \models \phi \W \psi $ if and only if $ \forall \varTime \geq 0.\; (\forall \varTime' \leq \varTime.\; \sigma, \varTime' \models \neg \psi) \lImp \sigma, \varTime \models \phi$.
  \end{enumerate}
\end{lemma}
\begin{proof}
  We prove the two statements separately.
  \begin{enumerate}
    \item Using the definition of the derived operators $\G$ and $\F$, we get that 
    $\sigma, 0 \models \G \phi$ holds if and only if
    $\sigma, 0 \models \neg (\top \U \neg \phi)$ holds.
    Using the (Boolean) semantics of STL, we get that  this is equivalent to
    $\neg (\exists \varTime \geq 0.\; \sigma, \varTime \models \neg \phi \land \forall \varTime' < \varTime.\; \sigma, \varTime' \models \top)$. 
    After simple logical operations, we get that this is equivalent to $\forall \varTime \geq 0.\; \sigma, \varTime \models \phi$ as required.

    \item Using \ref{lemma:testing:stlDerivedOps:globally}, the definition of $\W$, the (Boolean) semantics of STL, and considering that $\timeSet = \NN$, we get that $\sigma, 0 \models \phi \W \psi $ if and only if $\exists \varTime \in \NN.\; \sigma, \varTime \models \psi \land \forall \varTime' < \varTime.\; \sigma, \varTime' \models \phi$ or $\forall \varTime \in \NN.\; \sigma, \varTime \models \phi$.
    We denote this proposition as $V$.
    It is easy to see that the right operand of the equivalence to prove can be rewritten to
    $\forall \varTime \in \NN.\; (\exists \varTime' \leq \varTime.\; \sigma, \varTime' \models \psi) \lor \sigma, \varTime \models \phi$.
    We denote this proposition as $W$ and must show that $V \lImp W$ and $W \lImp V$.
    To prove that $V$ implies $W$, we distinguish two cases.
    \begin{itemize}
      \item For the first case, assume that the left operand of the disjunction in $V$ holds, i.e., there is some $\varTime \in \NN$, such that $\sigma, \varTime \models \psi \land \forall \varTime' < \varTime.\; \sigma, \varTime' \models \phi$.
      To show $W$, let $\varTime_0 \in \NN$ be arbitrary.
      If $\varTime \leq \varTime_0$, then there exists $\varTime' \leq \varTime_0$ (namely $\varTime' = \varTime$) such that $\sigma, \varTime' \models \psi$; hence $W$ holds.
      If $\varTime > \varTime_0$, then we know from $\forall \varTime' < \varTime.\; \sigma, \varTime' \models \phi$ that $\sigma, \varTime_0 \models \phi$ is true; hence, $W$ holds.
      \item For the second case, assume that the right operand of the disjunction in $V$ holds, i.e., $\forall \varTime \in \NN.\; \sigma, \varTime \models \phi$.
      Then, obviously $W$ holds.
    \end{itemize}
    To prove that $W$ implies $V$, let $\mathit{PV} = \set{\varTime \in \NN \mid \sigma, \varTime \models \psi}$ be the set of all time points at which $\psi$ holds.
    If $\mathit{PV}$ is the empty set, it follows immediately from $W$ that $\forall \varTime \in \NN.\; \sigma, \varTime \models \phi$ and that, hence, $V$ holds.
    If $\mathit{PV}$ is not empty, let $\varTime = \min \mathsf{PV}$ be the smallest time in  $\mathit{PV}$ (the minimum always exists, because $\timeSet = \NN$).
    Then, obviously, $\exists \varTime \in \NN.\; \sigma, \varTime \models \psi$.
    To show that $V$ holds, it suffices to show that $\forall \varTime' < \varTime.\; \sigma, \varTime' \models \phi$.
    This follows from $W$, because $\varTime$ is the smallest time at which $\sigma, \varTime \models \psi$ holds and, therefore, for every $\varTime' < \varTime$ it does not hold that $\sigma, \varTime' \models \psi$.
  \end{enumerate}
\end{proof}

\Cref{lemma:hyperstl:translateFormulaIntoLogic} is specific to the HyperSTL formula~(\ref{eq:hyperstl:urob}); it converts it into a first-order logic formula.

\begin{lemma}\label{lemma:hyperstl:translateFormulaIntoLogic}
  Let $\SysSTL \subseteq (\NN \to X)$ be a discrete-time system and let $\Std \subseteq \SysSTL$ be a set of standard traces.
  Also, let $\Std_\pi$ be a quantifier-free HyperSTL subformula, such that $\SysSTL, \set{\pi \coloneqq w}, 0 \models \Std_\pi$ if and only if $\w \in \Std$.
  Then, $\SysSTL, \emptyset, 0 \models \psi_{\text{u-rob}}$ if and only if
  \begin{align*}
    &\forall \w \in \Std.\; \forall \w'\in\SysSTL.\; \exists \w''\in\Std.\ (\forall t \geq 0.\; \mathsf{eq}(\mapInp{\w}[t], \mapInp{\w''}[t]) \leq 0) \land {} \\
    &\quad\  \forall t \geq 0.\; (\forall t' \leq t.\; \dIn(\mapInp{\w''}[t'], \mapInp{\w'}[t'])-\inpbound \leq 0) \lImp \dOut(\mapOut{\w''}[t], \mapOut{\w'}[t])-\outpbound \leq 0. 
  \end{align*}
\end{lemma}
\begin{proof}
  Using \Cref{lemma:testing:stlDerivedOps}.\ref{lemma:testing:stlDerivedOps:globally}, \Cref{lemma:testing:stlDerivedOps}.\ref{lemma:testing:stlDerivedOps:wUntil}, and \Cref{def:hyperstl:boolSem}, we get that 
  {\small\begin{align*}
    &\SysSTL, \emptyset, 0 \models  \AAA{\pi}\AAA{\pi'}\EEE{\pi''}
    {\Std_{\pi}}\\
    & \qquad \  \limp \Big({\Std_{\pi''}} \land {\G(\mathsf{eq}(\mapInp{\pi}, \mapInp{\pi''}) \leq 0)}\land {} \\
    & \qquad\qquad\qquad\qquad\qquad\ \ 
    \big((\dOut(\mapOut{\pi''},\mapOut{\pi'})-\outpbound\leq 0)\W(\dIn(\mapInp{\pi''},\mapInp{\pi'})-\inpbound > 0) \big)\Big)
  \end{align*}}
   holds if and only if 
   {\small\begin{align*}
    &\forall \w\in\SysSTL.\;\forall \w'\in\SysSTL.\; \exists \w''\in\SysSTL.\ (\SysSTL, \Pi, 0 \models \Std_{\pi})\\
    & \quad  \limp \Big((\SysSTL, \Pi, 0 \models{\Std_{\pi''}}) \land {(\forall t \geq 0.\; (\SysSTL, \Pi, t \models\mathsf{eq}(\mapInp{\pi}, \mapInp{\pi''}) \leq 0))}\land {} \\
    &  \qquad\qquad 
    \big( \forall t \geq 0.\; (\forall t' \leq t.\; (\SysSTL, \Pi, t' \models \neg \dIn(\mapInp{\pi''},\mapInp{\pi'})-\inpbound > 0)) \\
    & \qquad\qquad\qquad\qquad\qquad\qquad\qquad\qquad\quad\ \  {} \lImp (\SysSTL, \Pi, t \models \dOut(\mapOut{\pi''},\mapOut{\pi'})-\outpbound\leq 0) \big)\Big)
  \end{align*}} holds for $\Pi = \set{\pi \coloneqq \w, \pi' \coloneqq \w', \pi'' \coloneqq \w''}$.
  Using the the constraint under which $\Std_\pi$ must be modelled, and by further applying \Cref{def:hyperstl:boolSem} and basic logical operations, we get that the above proposition is equivalent to
  {\small\begin{align*}
    &\forall \w\in\SysSTL.\;\forall \w'\in\SysSTL.\; \exists \w''\in\SysSTL.\ \w\in\Std\\
    & \  \limp \Big(\w''\in\Std \land {(\forall t \geq 0.\; \mathsf{eq}(\mapInp{\w}[t], \mapInp{\w''}[t]) \leq 0)}\land {} \\
    &  \quad \ \ 
    \big( \forall t \geq 0.\; (\forall t' \leq t.\; \dIn(\mapInp{\w''}[t'],\mapInp{\w'}[t'])-\inpbound \leq 0) \lImp \dOut(\mapOut{\w''},\mapOut{\w'})-\outpbound\leq 0 \big)\Big).
  \end{align*}}
  Finally, after carefully reordering premises, we get that the above holds if and only if 
  {\small\begin{align*}
    &\forall \w\in\Std.\;\forall \w'\in\SysSTL.\; \exists \w''\in\Std.\ {(\forall t \geq 0.\; \mathsf{eq}(\mapInp{\w}[t], \mapInp{\w''}[t]) \leq 0)}\land {} \\
    &  \qquad
    \forall t \geq 0.\; (\forall t' \leq t.\; \dIn(\mapInp{\w''}[t'],\mapInp{\w'}[t'])-\inpbound \leq 0) \lImp \dOut(\mapOut{\w''},\mapOut{\w'})-\outpbound\leq 0.
  \end{align*}}
\end{proof}

We omit the lemma analogue to \Cref{lemma:hyperstl:translateFormulaIntoLogic} that reformulates formula~(\ref{eq:hyperstl:lrob}) as a first-order characterisation.
The proof for \Cref{prop:hyperstl:robCorrectness} further transforms the first-order characterisations of formulas~(\ref{eq:hyperstl:lrob}) and~(\ref{eq:hyperstl:urob}) 
to show that they indeed match the definitions of 
\robustLowCleannessNDet and \robustUpCleannessNDet.

\propHyperstlRobCorrectness*
\begin{proof}
  We prove the correctness for \robustLowCleannessNDet and \robustUpCleannessNDet separately and begin with \robustUpCleannessNDet.
  Using \Cref{lemma:hyperstl:translateFormulaIntoLogic}, we get that 
  {\small\begin{align*}
    &\mixedIOSys, \emptyset, 0 \models  \AAA{\pi_1}\AAA{\pi_2}\EEE{\pi'_1}
    {\Std_{\pi_1}}\\
    & \qquad \  \limp \Big({\Std_{\pi'_1}} \land {\G(\mathsf{eq}(\mapInp{\pi_1}, \mapInp{\pi'_1}) \leq 0)}\land {} \\
    & \qquad\qquad\qquad\qquad\qquad\ \ 
    \big((\dOut(\mapOut{\pi'_1},\mapOut{\pi_2})-\outpbound\leq 0)\W(\dIn(\mapInp{\pi'_1},\mapInp{\pi_2})-\inpbound > 0) \big)\Big)
  \end{align*}}
  holds if and only if 
  {\small\begin{align*}
    &\forall \w_1\in\Std.\;\forall \w_2\in\mixedIOSys.\; \exists \w'_1\in\Std.\ {(\forall t \geq 0.\; \mathsf{eq}(\mapInp{\w_1}[t], \mapInp{\w'_1}[t]) \leq 0)}\land {} \\
    &  \ 
    \forall t \geq 0.\; (\forall t' \leq t.\; \dIn(\mapInp{\w'_1}[t'],\mapInp{\w_2}[t'])-\inpbound \leq 0) \lImp \dOut(\mapOut{\w'_1},\mapOut{\w_2})-\outpbound\leq 0.
  \end{align*}}
  After applying simple logical operations and using that  $\mathsf{eq}(\inp_1, \inp_2) = 0$ if and only if $\inp_1 = \inp_2$, we get that this is equivalent to
  {\small\begin{align*}
    &\forall \w_1\in\Std.\;\forall \w_2\in\mixedIOSys.\; \exists \w'_1\in\Std \text{ with } \mapInp{\w_1} = \mapInp{\w'_1}.\\
    &  \ 
    \big( \forall t \geq 0.\; (\forall t' \leq t.\; \dIn(\mapInp{\w'_1}[t'],\mapInp{\w_2}[t']) \leq \inpbound) \lImp  \dOut(\mapOut{\w'_1}[t],\mapOut{\w_2}[t]) \leq \outpbound \big),
  \end{align*}}
  which, since we assumed $\Std \subseteq \mixedIOSys$, is equivalent to the definition of  \robustUpCleannessNDet for mixed-IO systems.

  The proof for \robustLowCleannessNDet is analogue.
\end{proof}

We recapitulate the proposition similar to \Cref{prop:hyperstl:robCorrectness} for \fCleannessNDet.

\propHyperstlFCorrectness*

The proof for \Cref{prop:hyperstl:fCorrectness} is conceptually similar to the one for \Cref{prop:hyperstl:robCorrectness}.
The only difference is that instead of the reasoning about the $\W$ construct, the globally enforced relation between output distances and the result of $f$ must be proven equivalent in the HyperSTL formulas and \fCleannessNDet.
We omit the proofs here.

\paragraph{Correctness of STL characterisations}
Next, we show the correctness of the STL characterisations, i.e., we will prove the correctness of \Cref{thm:hyperstl:urobCorrectnessSTL,thm:hyperstl:ufCorrectnessSTL}.
We do so by first establishing a connection between the HyperSTL and the STL characterisations.

\begin{restatable}{lemma}{propHyperstlSTLrobCleannessCorrectness}
    \label{prop:hyperstl:STLrobCleannessCorrectness}
    Let $\SysSTL \subseteq (\NN \to X)$ be a discrete-time system 
    and let $\Std = \set{w_1, \dots, w_\stdcnt} \subseteq \SysSTL$ be a finite set of standard traces.
    Also, let $\Std_\pi$ be a quantifier-free HyperSTL subformula, such that $\SysSTL, \set{\pi \coloneqq w}, 0 \models \Std_\pi$ if and only if $\w \in \Std$.
    Then, $\SysSTL, \emptyset, 0 \models \psi_{\textsf{u-rob}}$ if and only if $(\SysSTL \circ \StdSet) \models \varphi_{\textsf{u-rob}}$ (with $\varphi_{\textsf{u-rob}}$ from \Cref{thm:hyperstl:urobCorrectnessSTL}). 
  \end{restatable}

\begin{proof}
  Using \Cref{lemma:hyperstl:translateFormulaIntoLogic} we get that 
  {\small\begin{align*}
    &\SysSTL, \emptyset, 0 \models  \AAA{\pi'}\AAA{\pi''}\EEE{\pi'''}
    {\Std_{\pi'}}\\
    & \qquad \  \limp \Big({\Std_{\pi'''}} \land {\G(\mathsf{eq}(\mapInp{\pi'}, \mapInp{\pi'''}) \leq 0)}\land {} \\
    & \qquad\qquad\qquad\qquad\qquad\ \ 
    \big((\dOut(\mapOut{\pi'''},\mapOut{\pi''})-\outpbound\leq 0)\W(\dIn(\mapInp{\pi'''},\mapInp{\pi''})-\inpbound > 0) \big)\Big)
  \end{align*}}
  holds if and only if 
  {\small\begin{align*}
    &\forall \w'\in\Std.\;\forall \w''\in\SysSTL.\; \exists \w'''\in\Std.\ {(\forall t \geq 0.\; \mathsf{eq}(\mapInp{\w'}[t], \mapInp{\w'''}[t]) \leq 0)}\land {} \\
    &  \ 
    \forall t \geq 0.\; (\forall t' \leq t.\; \dIn(\mapInp{\w'''}[t'],\mapInp{\w''}[t'])-\inpbound \leq 0) \lImp \dOut(\mapOut{\w'''},\mapOut{\w''})-\outpbound\leq 0.
  \end{align*}}
  Since $\Std =  \set{w_1, \dots, w_\stdcnt}$, we can replace the universal and existential quantifiers over $\Std$ by a conjunction, respectively disjunction, over the standard traces~\cite{rosen2012discrete}. 
  We instantiate the universal quantifier for $\w''$ with $\w$ and get that
  {\small\begin{align*}
    &\bigwedge_{1 \leq a \leq \stdcnt} \ \bigvee_{1 \leq b \leq \stdcnt} {(\forall t \geq 0.\; \mathsf{eq}(\mapInp{\w_a}[t], \mapInp{\w_b}[t]) \leq 0)}\land {} \\
    &  \ 
    \forall t \geq 0.\; (\forall t' \leq t.\; \dIn(\mapInp{\w_b}[t'],\mapInp{\w}[t'])-\inpbound \leq 0) \lImp \dOut(\mapOut{\w_b},\mapOut{\w})-\outpbound\leq 0.
  \end{align*}}
  From the Boolean semantics of STL and by replacing all traces $\w$, respectively $\w_k$, by the corresponding $\w_+$-projections, we get the equivalent proposition
  {\small\begin{align*}
    &\bigwedge_{1 \leq a \leq \stdcnt} \ \bigvee_{1 \leq b \leq \stdcnt} {(\forall t \geq 0.\; (\w_+, t \models \mathsf{eq}(\mapInp{\w_a}, \mapInp{\w_b}) \leq 0))}\land {} \\
    &  \ 
    \forall t \geq 0.\; (\forall t' \leq t.\; (\w_+, t' \models \neg \dIn(\mapInp{\w_b},\mapInp{\w})-\inpbound > 0)) \lImp (\w_+, t \models \dOut(\mapOut{\w_b},\mapOut{\w})-\outpbound\leq 0).
  \end{align*}}
  With the Boolean semantics of STL and Lemma \ref{lemma:testing:stlDerivedOps}.\ref{lemma:testing:stlDerivedOps:globally} and \ref{lemma:testing:stlDerivedOps}.\ref{lemma:testing:stlDerivedOps:wUntil} we get the equivalent statement that
  {\small\begin{align*}
    &\w_+, 0 \models \bigwedge_{1 \leq a \leq \stdcnt} \ \bigvee_{1 \leq b \leq \stdcnt} {(\G (\mathsf{eq}(\mapInp{\w_a}, \mapInp{\w_b}) \leq 0))}\land {} \\
    &  \ \big( (\dOut(\mapOut{\w_b},\mapOut{\w})-\outpbound\leq 0) \W (\dIn(\mapInp{\w_b},\mapInp{\w})-\inpbound > 0) \big).
  \end{align*}}
\end{proof}

We are now able to prove \Cref{thm:hyperstl:urobCorrectnessSTL}.

\thmHyperstlUrobCorrectnessSTL*
\begin{proof}
  The theorem follows from \Cref{prop:hyperstl:robCorrectness} and \Cref{prop:hyperstl:STLrobCleannessCorrectness}.
\end{proof}

To prove \Cref{thm:hyperstl:ufCorrectnessSTL}, we establish the following lemma, which is analogue to \Cref{prop:hyperstl:STLrobCleannessCorrectness}, up to \fUpCleannessNDet replacing \robustUpCleannessNDet.

  \begin{lemma}\label{prop:hyperstl:STLfCleannessCorrectness}
    Let $\SysSTL \subseteq (\timeSet \to X)$ be a system 
    and let $\Std = \set{w_1, \dots, w_\stdcnt} \subseteq \SysSTL$ be a finite set of standard traces.
    Also, let $\Std_\pi$ be a quantifier-free HyperSTL subformula, such that $\SysSTL, \set{\pi \coloneqq w}, 0 \models \Std_\pi$ if and only if $\w \in \Std$.
    Then, $\SysSTL, \emptyset, 0 \models \psi_{\textsf{u-fun}}$ if and only if $(\SysSTL \circ \StdSet) \models \varphi_{\textsf{u-fun}}$ (with $\varphi_{\textsf{u-fun}}$ from \Cref{thm:hyperstl:ufCorrectnessSTL}). 
\end{lemma}

The proof for \Cref{prop:hyperstl:STLfCleannessCorrectness} is, up to the different reasoning for $\G(\dOut(\mapOut{\w_b},\mapOut{\w}) - f(\dIn(\mapInp{\w_b},\mapInp{\w})) \leq 0)$ instead of $(d_\Outputs(\mapOut{\w_b}, \mapOut{\w})-\outpbound \leq 0)\W(d_\Inputs(\mapInp{\w_b}, \mapInp{\w}) - \inpbound > 0)$, identical to that of \Cref{prop:hyperstl:STLrobCleannessCorrectness}. 
We omit it here.

\thmHyperstlUfCorrectnessSTL*
\begin{proof}
  The theorem follows from \Cref{prop:hyperstl:fCorrectness} and \Cref{prop:hyperstl:STLfCleannessCorrectness}.
\end{proof}

\section{Fairness Pipeline}
\label{app:pipeline_full}

As explained in \Cref{sec:background} in the main paper,
it is important to recognise that there are many sources of unfairness \cite{barocas2016big}. 
\Cref{app:pipeline_full} shows a more detailed version of \Cref{fig:pipeline} in the main paper.
Not every technical measure is able to detect every kind of unfairness and eliminating one source of unfairness might not be sufficient to eliminate all unfairness. 
\begin{figure}[t]
    \centering
    \resizebox{0.95\linewidth}{!}{
      \tikzstyle{rectang} = [rectangle, rounded corners, minimum width=2cm, minimum height=1cm,text centered, draw=black, fill=gray!15, font=\sffamily]
\tikzstyle{round} = [ellipse, minimum width=2cm, minimum height=1cm,text centered, draw=black, fill=gray!15, font=\sffamily]
\tikzstyle{myplain} = [draw=none, fill=none, minimum width=2cm, text width=2.5cm, text depth=1.2cm, text centered]
\tikzstyle{arrow} = [thick,->,>=stealth]

\begin{tikzpicture}[node distance=3cm] \small
\node (world)       [rectang]                   {world};
\node (input)       [round, right of=world]     {input data};
\node (system)      [rectang, right of=input]   {system};
\node (training)    [round, dashed, above=0.4cm of system]    {training data};
\node (output)      [round, right of=system]    {output};
\node (decision)    [rectang, right of=output]  {decision};

\draw [arrow] (world) -- (input);
\draw [arrow] (input) -- (system);
\draw [arrow] (system) -- (output);
\draw [arrow] (output) -- (decision);

\draw [arrow, dashed] (training) -- (system);
\draw [arrow, dashed] (world) -- (training.west);

\draw [arrow] (decision.south) -- ++(0,-0.3) -|  (world.south);

\node (worldEx)     [myplain, below=1.2cm of world.center] {unfairness in the world};
\node (inputEx)     [myplain, text width=2.75cm, below=1.2cm of input.center] {unfairness introduced through the input data or in its collection, representation or selection};
\node (systemEx)    [myplain, text width=2cm, text depth=0.8cm, below=1.2cm of system.center] {unfairness introduced by the system itself};
\node (outputEx)    [myplain, text width=2.75cm, below=1.2cm of output.center] {unfairness introduced by the human interpretation of the output};
\node (decisionEx)  [myplain, below=1.2cm of decision.center] {unfairness introduced by the human decision};

%\node (rfmBox) [draw, thick, blue!40!black!50, densely dotted, rounded corners, fit=(training) (system) (systemEx), inner sep = 2mm] {};
%\node (rfmLabel) [align=right, text=blue!40!black!50, font=\sffamily\small, rotate=90, text width=2cm, anchor=north, below=1cm of rfmBox.north east] {our fairness\\ monitoring};

\end{tikzpicture}
    }

    \vspace{3em}

    \caption{Sketch of different origins of unfairness in a decision process supported by a system; dashed elements are inapplicable to systems that are not learning-based. }%

    \label{fig:pipelinefull}
\end{figure}
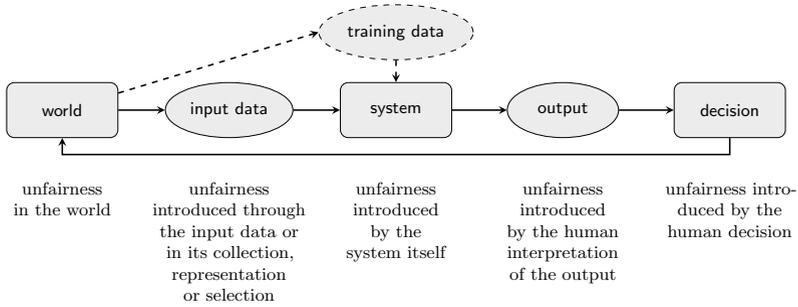

\begin{description}
\item[World] There can be unfairness in the world that leads to individuals already having worse (or better) starting conditions than others and subsequently have a lower (or higher) chance that the final decision is made in their favour. For example, an individual could be systematically excluded from certain societal resources (e.g., girls who are excluded from education in Afghanistan under the Taliban) which puts these individuals at a disadvantage.    
\item[Input data] The input data or its collection, representation or selection could be problematic and lead to unfairness \cite{lahoti2019}. If, for example, crucial information is left out in the input data or data is aggregated in unsuitable ways, individuals could face an outcome that is unwarranted by the factual situation.
\item[System (and training data)]  The system itself can introduce new unfairness. 
Among other things, this can come about by erroneous algorithms or (in the case of a trained model) by problematic training data, e.g., if a certain group of individuals  is not properly represented \cite{tay2021conceptual}. 
\item[Output] The human decider can fail to interpret the output properly \cite{hoff2015trust, langer2021trust}, which can lead to further unfairness. They could, for example, lack knowledge of the limitations of the system or fail to take into account that the system output is subject to some systematic uncertainty.
\item[Decision]  The human decider can make an unfair decision even in the face of a fair system output and an adequate interpretation thereof, for example if they have conscious or subconscious bias against certain groups \cite{bertrand2004emily}. 
\end{description}
Unfairness in any part of the chain can arguably perpetuate or reinforce unfairness in the world.

In the main paper, we propose a runtime monitoring technique that aims to uncover individual unfairness introduced by the system. 
By focusing on the system and its input-output relation only, we can say that the system is unfair without having to say anything about the degree of fairness with which an individual is treated in other respects in the decision process. 
It especially allows us to say that a system output is unfair, even though the outcome of the overall decision process is not. It may, for example, be that 
the system unfairness is \enquote*{cancelled out} by something else that is hidden from the system: an applicant with a stellar-looking CV might be treated unfairly by the system because of their age, but not hiring them is not unfair because they are known to have forged their diploma. Cases like this, however, do not make the unfairness introduced by the system any less problematic.

\section{Legal Appendix}
\label{app:legal}

\subsection{EU Anti-Discrimination Law}\label{app:legal:discrimination} Antidiscrimination is a principle deeply rooted in EU law. It is enshrined in Art.\ 21 of the Charter of Fundamental Rights (CFR) \cite{CFR}, which prohibits \enquote{[a]ny discrimination based on any ground such as sex, race, colour, ethnic or social origin, genetic features, language, religion or belief, political or any other opinion, membership of a national minority, property, birth, disability, age or sexual orientation} as well as  \enquote{any discrimination based on grounds of nationality}.  According to Art.\ 51 CFR, the addressees of this fundamental right are the EU and its institutions, bodies, offices and agencies as well as the Member States, insofar as they implement Union law. They are directly bound by Art.\ 21 above all in their legislative activities, but also in their executive and judicial measures. In contrast, private individuals are not directly bound by Art.\ 21 CFR, but they may be bound by regulations implementing this provision. However, according to recent European Court of Justice (ECJ) case law, Art.\ 21 CFR is directly applicable as a result of Directives, such as Directive 2000/78/EC \cite{EUdir2000/78/EC} establishing a general framework for equal treatment in employment and occupation \cite[§ 76]{ecj-c-414-16}.
Apart from this, while Art.\ 21 CFR stipulates a general prohibition of any unjustified discrimination, the more specific secondary legislation applicable to private actors only prohibits discrimination only in certain sensitive areas and only with regard to certain protected attributes. Correspondingly, private actors may not discriminate against certain persons---to name just a few---in employment relationships \cite{EUdir2000/78/EC},
in cases of abuse of a dominant market position \cite[Art.\ 102]{AEUV}
or also in so-called mass transactions, i.e., contracts that are typically concluded without regard to the person on comparable terms in a large number of cases \cite{EUdir2004/113/EC}.
In contrast, discriminating in other legal relationships or on other grounds such as local origin (as opposed to ethnic origin), or a person's financial situation is not generally prohibited. The rationale behind these  \enquote{discriminatory standards of anti-discrimination law} \cite{thuesingZfA, hartmann, schwab} is the principle of private (or personal) autonomy, and more specifically freedom of contract as one of its manifestations, which govern legal transactions between private individuals \cite{looschelders}. According to this principle, individuals are free to shape their legal relationships according to their own preferences and ideas, however irrational or socially unacceptable they may be. In essence, this also includes a right to discriminate against others. This freedom to autonomously form legal relations is only constrained where this is stipulated by anti-discrimination legislation for policy reasons. 

When using an AI-system to recruit candidates, developers and deployers have to make sure that the system with its parameters comply with these legal requirements set by anti-discrimination law. This means in particular that the selection of the properties which a classifier should exhibit requires a positive normative choice and should not simply replicate biases present in the status quo \cite{wachter2020bias}. However, the risks associated with deploying such systems in an HR context (such as  a malfunctioning remaining undetected due to the system's opacity, a huge practical relevance of biased outputs due to the systems' scalability or the human operator's tendency of over-relying on the output produced by the AI system ( \enquote{automation bias})), raise the question whether it can still be deemed normatively acceptable that the EU legal framework turns a blind eye on certain forms of discrimination. Furthermore, the principle of private autonomy as rationale for justifying the freedom to discriminate against others is only valid with regard to human's wilful actions, but not to algorithm-generated output. We are not advocating for abolishing the existing balance between private autonomy (freedom to contract) and prohibition to discriminate. So humans should still be permitted to differentiate on grounds that are not caught by anti-discrimination law. However, there is no reason to grant the  \enquote{right to discriminate} also to a non-human system that has merely "learned" this discrimination. In this respect, it seems justified to apply different standards for algorithms with regard to the prohibition of discrimination than for human decisions. With regard to an AI system's decision metrics, therefore, it should be considered to expand the secondary legal framework to include a broad prohibition of discrimination. This would not mean that all discrimination would be unlawful, since objectively justified unequal treatment is, after all, permissible, but it would shift the focus to the question of objective justification \cite{ep-2020-2012}. Another legal challenge that will become even more pressing with the advent of technical decision systems is how to detect and prove prohibited discrimination. This is because the prohibition of discrimination resulting from various legal regulations in certain, especially sensitive, areas, such as human resources, presupposes that a difference in treatment is recognised in the first place. The recognition of discrimination is therefore not only in the interest of the decision-maker, who is threatened with sanctions in the event of a violation of the prohibition of discrimination. Rather, it is also essential for the discriminated party to prove the discrimination. For as far as a legal claim follows from a prohibited discrimination, the principle applies that the person who invokes the legal claim must prove the facts giving rise to the claim. Especially when complex algorithms are used, however, it is likely to be extremely difficult to prove corresponding circumstantial evidence. According to the case law of the ECJ, however, the burden of proof is reversed if the party who has prima facie been discriminated against would otherwise have no effective means of enforcing the prohibition of discrimination \cite{ecj-c-127-92, ecj-c-400-93}. Monitoring, as described here, would therefore be a suitable means of providing the  \enquote{prima facie} evidence necessary for shifting the burden of proof.
 
 \subsection{Discrimination and the GDPR}\label{app:legal:discriminationgdpr}

There has recently been discussion if and to which extent data protection law contains obligations for non-discriminating data processing or whether the scope of protection of data protection law is thereby overstretched. There is no explicit prohibition of discrimination in the General Data Protection Regulation (GDPR). According to Article 1 (2), however, the GDPR is intended to protect the fundamental rights and freedoms of natural persons. This is aimed in particular at their right to protection of personal data (Article 8 CFR), but not exclusively so. Thus, the broad and non-restrictive reference to fundamental rights also encompasses all other fundamental rights, including the right to non-discrimination (Article 21 CFR) \cite{KOM-GDPR-proposal}. This is reflected, for example, in the higher level of protection for data with an increased potential for discrimination, the so-called special categories of personal data under Article 9 GDPR. The GDPR can also be interpreted as granting a “preventive protection against discrimination”, namely when discrimination is made impossible from the outset, in that the data-processing agencies cannot gain knowledge of characteristics susceptible to discrimination in the first place, i.e., when any respective data processing is forbidden \cite{KuhlingBuchner}. Any processing of personal data must furthermore comply with the processing principles set out in Article 5 GDPR, including the fairness principle (‘personal data shall be processed fairly’) set out in Article 5(1)(a). While formerly transparency obligations were read into this principle while the Data Protection Directive was into effect, the regulatory content of the fairness principle is highly disputed since it was split off into a separate processing principle. But due to the fact that discriminatory data processing can hardly be described as fair, a prohibition of discrimination can be linked to the fairness principle \cite{hacker2018,MalgieriFacct2020}. However, the concrete scope of the fairness principle clearly goes beyond the understanding of fairness in the context of technical systems on which this paper is based. 

Specifically for the HR context, there are discrimination-sensitive regulations in the GDPR. Article 9 GDPR makes the processing of special categories of data, i.e., sensitive data and data susceptible to discrimination, subject to particularly strict authorisation criteria, which should in practice rarely be present in recruitment situations. On the one hand, processing for recruitment purposes, i.e., prior to the establishment of an employment relationship, is rarely necessary in order to exercise certain rights and obligations under employment law (Art.\ 9(2)(b) GDPR), and on the other hand, explicit consent (Art.\ 9(2)(a) GDPR) will often lack the necessary voluntariness due to the specifics of job application situations and the power imbalances inherent in them. The prohibition of processing sensitive data may be problematic in cases where the link to sensitive data is strictly necessary to detect discriminatory effects. For high-risk systems, Art.\ 10 V AI Regulation Proposal therefore provides for a new permissive clause: 'To the extent that it is strictly necessary for the purposes of ensuring bias monitoring, detection and correction, ... the providers of such systems may process special categories of personal data' while ensuring appropriate safeguards for the fundamental rights of natural persons.

With regard to the processing of non-sensitive personal data, however, the opening clause in Art 88(1) GDPR allows Member States to adopt more specific rules for processing for recruitment purposes, whereby, according to paragraph 2, suitable and specific measures must be ensured to safeguard the fundamental rights of the data subject. These requirements can be met by state-of-the-art monitoring tools. The national regulations cannot be discussed in depth here. For Germany, for example, Section 26 of the Federal Data Protection Act (BDSG) stipulates that personal data may only be processed for recruitment purposes if this is necessary, i.e., if the data processing is required for the decision on recruitment. In any case, data processing may not be necessary if the characteristics depicted in the data may not be taken into account in the hiring decision, for example due to anti-discrimination law \cite{Riesenhuber}.

\label{sec:background:related:aiact}

% \section{Section title of first appendix}\label{secA1}

% An appendix contains supplementary information that is not an essential part of the text itself but which may be helpful in providing a more comprehensive understanding of the research problem or it is information that is too cumbersome to be included in the body of the paper.

% %%=============================================%%
% %% For submissions to Nature Portfolio Journals %%
% %% please use the heading ``Extended Data''.   %%
% %%=============================================%%

% %%=============================================================%%
% %% Sample for another appendix section			       %%
% %%=============================================================%%

% %% \section{Example of another appendix section}\label{secA2}%
% %% Appendices may be used for helpful, supporting or essential material that would otherwise 
% %% clutter, break up or be distracting to the text. Appendices can consist of sections, figures, 
% %% tables and equations etc.

\end{appendices}


\bibliography{bibliography}% common bib file

\begin{thebibliography}{138}
\providecommand{\natexlab}[1]{#1}
\providecommand{\url}[1]{{#1}}
\providecommand{\urlprefix}{URL }
\providecommand{\doi}[1]{\url{https://doi.org/#1}}
\providecommand{\eprint}[2][]{\url{#2}}
 \bibcommenthead

\bibitem[{Abbas et~al(2013)Abbas, Fainekos, Sankaranarayanan, Ivancic, and
  Gupta}]{DBLP:journals/tecs/AbbasFSIG13}
Abbas H, Fainekos GE, Sankaranarayanan S, et~al (2013) Probabilistic temporal
  logic falsification of cyber-physical systems. {ACM} Trans Embed Comput Syst
  12(2s):95:1--95:30. \doi{10.1145/2465787.2465797}

\bibitem[{Alves and Rossi(1978)}]{alves1978should}
Alves WM, Rossi PH (1978) Who should get what? fairness judgments of the
  distribution of earnings. American journal of Sociology 84(3):541--564

\bibitem[{Angwin et~al(2016)Angwin, Larson, Mattu, and
  Kirchner}]{compaspropubstory}
Angwin J, Larson J, Mattu S, et~al (2016) {Machine Bias}.
  \urlprefix\url{https://www.propublica.org/article/machine-bias-risk-assessments-in-criminal-sentencing}

\bibitem[{Annapureddy and Fainekos(2010)}]{5675195}
Annapureddy YSR, Fainekos GE (2010) Ant colonies for temporal logic
  falsification of hybrid systems. In: IECON 2010 - 36th Annual Conference on
  IEEE Industrial Electronics Society, pp 91--96,
  \doi{10.1109/IECON.2010.5675195}

\bibitem[{Arrieta et~al(2020)Arrieta, D{\'\i}az-Rodr{\'\i}guez, Del~Ser,
  Bennetot, Tabik, Barbado, Garc{\'\i}a, Gil-L{\'o}pez, Molina, Benjamins
  et~al}]{arrieta2020explainable}
Arrieta AB, D{\'\i}az-Rodr{\'\i}guez N, Del~Ser J, et~al (2020) Explainable
  artificial intelligence ({XAI}): Concepts, taxonomies, opportunities and
  challenges toward responsible {AI}. Information Fusion 58:82--115

\bibitem[{Artistotle(1998{\natexlab{a}})}]{Aristot_ne}
Artistotle (1998{\natexlab{a}}) The Nicomachean Ethics. Oxford worlds classics,
  Oxford University Press, Oxford, translation by W.D. Ross. Edition by John L.
  Ackrill, and James O. Urmson.

\bibitem[{Artistotle(1998{\natexlab{b}})}]{Aristot_pol}
Artistotle (1998{\natexlab{b}}) Politics. Oxford worlds classics, Oxford
  University Press, Oxford, translation by Ernest Barker. Edition by R. F.
  Stalley.

\bibitem[{Barocas and Selbst(2016)}]{barocas2016big}
Barocas S, Selbst AD (2016) Big data's disparate impact. Calif L Rev 104:671

\bibitem[{Barthe et~al(2011)Barthe, D'Argenio, and Rezk}]{BartheDR11:mscs}
Barthe G, D'Argenio PR, Rezk T (2011) Secure information flow by
  self-composition. Mathematical Structures in Computer Science
  21(6):1207--1252. \urlprefix\url{http://dx.doi.org/10.1017/S0960129511000193}

\bibitem[{Barthe et~al(2016)Barthe, D'Argenio, Finkbeiner, and
  Hermanns}]{DBLP:conf/isola/BartheDFH16}
Barthe G, D'Argenio PR, Finkbeiner B, et~al (2016) Facets of software doping.
  In: Margaria T, Steffen B (eds) Leveraging Applications of Formal Methods,
  Verification and Validation: Discussion, Dissemination, Applications - 7th
  International Symposium, ISoLA 2016, Imperial, Corfu, Greece, October 10-14,
  2016, Proceedings, Part {II}, pp 601--608,
  \urlprefix\url{https://doi.org/10.1007/978-3-319-47169-3\_46}

\bibitem[{Bathaee(2017)}]{bathaee2017artificial}
Bathaee Y (2017) The artificial intelligence black box and the failure of
  intent and causation. Harv JL \& Tech 31:889

\bibitem[{Baum(2016)}]{Baum16:isola}
Baum K (2016) What the hack is wrong with software doping? In: Margaria T,
  Steffen B (eds) Leveraging Applications of Formal Methods, Verification and
  Validation: Discussion, Dissemination, Applications - 7th International
  Symposium, ISoLA 2016, Imperial, Corfu, Greece, October 10-14, 2016,
  Proceedings, Part {II}, pp 633--647, \doi{10.1007/978-3-319-47169-3\_49},
  \urlprefix\url{https://doi.org/10.1007/978-3-319-47169-3\_49}

\bibitem[{Baum et~al(2022)Baum, Mantel, Schmidt, and
  Speith}]{Baumetal-BAUFRT-2}
Baum K, Mantel S, Schmidt E, et~al (2022) From responsibility to reason-giving
  explainable artificial intelligence. Philosophy \& Technology 35(1):12.
  \doi{10.1007/s13347-022-00510-w},
  \urlprefix\url{https://doi.org/10.1007/s13347-022-00510-w}

\bibitem[{Biewer(2023{\natexlab{a}})}]{biewer_sebastian_2023_8058770}
Biewer S (2023{\natexlab{a}}) Real driving emissions tests records.
  \doi{10.5281/zenodo.8058770},
  \urlprefix\url{https://doi.org/10.5281/zenodo.8058770}

\bibitem[{Biewer(2023{\natexlab{b}})}]{diss:Biewer}
Biewer S (2023{\natexlab{b}}) Software doping -- theory and detection.
  Dissertation (forthcoming)

\bibitem[{Biewer and Hermanns(2022)}]{DBLP:conf/fase/BiewerH22}
Biewer S, Hermanns H (2022) On the detection of doped software by
  falsification. In: Johnsen EB, Wimmer M (eds) Fundamental Approaches to
  Software Engineering - 25th International Conference, {FASE} 2022, Held as
  Part of the European Joint Conferences on Theory and Practice of Software,
  {ETAPS} 2022, Munich, Germany, April 2-7, 2022, Proceedings, Lecture Notes in
  Computer Science, vol 13241. Springer, pp 71--91,
  \doi{10.1007/978-3-030-99429-7\_4},
  \urlprefix\url{https://doi.org/10.1007/978-3-030-99429-7\_4}

\bibitem[{Biewer et~al(2019)Biewer, D'Argenio, and
  Hermanns}]{DBLP:conf/qest/BiewerDH19}
Biewer S, D'Argenio PR, Hermanns H (2019) Doping tests for cyber-physical
  systems. In: Parker D, Wolf V (eds) Quantitative Evaluation of Systems, 16th
  International Conference, {QEST} 2019, Glasgow, UK, September 10-12, 2019,
  Proceedings, Lecture Notes in Computer Science, vol 11785. Springer, pp
  313--331, \doi{10.1007/978-3-030-30281-8\_18},
  \urlprefix\url{https://doi.org/10.1007/978-3-030-30281-8\_18}

\bibitem[{Biewer et~al(2021{\natexlab{a}})Biewer, D'Argenio, and
  Hermanns}]{DBLP:journals/tomacs/BiewerDH21}
Biewer S, D'Argenio PR, Hermanns H (2021{\natexlab{a}}) Doping tests for
  cyber-physical systems. {ACM} Trans Model Comput Simul 31(3):16:1--16:27.
  \doi{10.1145/3449354}, \urlprefix\url{https://doi.org/10.1145/3449354}

\bibitem[{Biewer et~al(2021{\natexlab{b}})Biewer, Finkbeiner, Hermanns,
  K{\"{o}}hl, Schnitzer, and Schwenger}]{DBLP:conf/tacas/BiewerFHKSS21}
Biewer S, Finkbeiner B, Hermanns H, et~al (2021{\natexlab{b}}) \rtlola on
  board: Testing real driving emissions on your phone. In: Groote JF, Larsen KG
  (eds) Tools and Algorithms for the Construction and Analysis of Systems -
  27th International Conference, {TACAS} 2021, Held as Part of the European
  Joint Conferences on Theory and Practice of Software, {ETAPS} 2021,
  Luxembourg City, Luxembourg, March 27 - April 1, 2021, Proceedings, Part
  {II}, Lecture Notes in Computer Science, vol 12652. Springer, pp 365--372,
  \doi{10.1007/978-3-030-72013-1\_20},
  \urlprefix\url{https://doi.org/10.1007/978-3-030-72013-1\_20}

\bibitem[{Biewer et~al(2022)Biewer, Dimitrova, Fries, Gazda, Heinze, Hermanns,
  and Mousavi}]{DBLP:journals/lmcs/BiewerDFGHHM22}
Biewer S, Dimitrova R, Fries M, et~al (2022) Conformance relations and
  hyperproperties for doping detection in time and space. Log Methods Comput
  Sci 18(1). \doi{10.46298/lmcs-18(1:14)2022},
  \urlprefix\url{https://doi.org/10.46298/lmcs-18(1:14)2022}

\bibitem[{Biewer et~al(2023)Biewer, Finkbeiner, Hermanns, K{\"{o}}hl,
  Schnitzer, and Schwenger}]{C3OTHER:STTT-toappear}
Biewer S, Finkbeiner B, Hermanns H, et~al (2023) On the road with rtlola. Int J
  Softw Tools Technol Transf 25(2):205--218. \doi{10.1007/s10009-022-00689-5},
  \urlprefix\url{https://doi.org/10.1007/s10009-022-00689-5}

\bibitem[{Binns(2020)}]{10.1145/3351095.3372864}
Binns R (2020) On the apparent conflict between individual and group fairness.
  In: Proceedings of the 2020 Conference on Fairness, Accountability, and
  Transparency. Association for Computing Machinery, New York, NY, USA, FAT*
  '20, p 514–524, \doi{10.1145/3351095.3372864},
  \urlprefix\url{https://doi.org/10.1145/3351095.3372864}

\bibitem[{Bloem et~al(2014)Bloem, Chatterjee, Greimel, Henzinger, Hofferek,
  Jobstmann, K{\"{o}}nighofer, and
  K{\"{o}}nighofer}]{DBLP:journals/acta/BloemCGHHJKK14}
Bloem R, Chatterjee K, Greimel K, et~al (2014) Synthesizing robust systems.
  Acta Informatica 51(3-4):193--220. \doi{10.1007/s00236-013-0191-5},
  \urlprefix\url{https://doi.org/10.1007/s00236-013-0191-5}

\bibitem[{Borgesius(2020)}]{Borgesius2020}
Borgesius FJZ (2020) Strengthening legal protection against discrimination by
  algorithms and artificial intelligence. The International Journal of Human
  Rights 24(10):1572--1593. \doi{10.1080/13642987.2020.1743976},
  \urlprefix\url{https://doi.org/10.1080/13642987.2020.1743976}

\bibitem[{Buchner(2020)}]{KuhlingBuchner}
Buchner B (2020) {DS-GVO Art. 1 Gegenstand und Ziele Rn. 14}. In: Buchner JK
  (ed) Datenschutz-Grundverordnung, Bundesdatenschutzgesetz. C.H. Beck, Munich

\bibitem[{Burke(2020)}]{texasuniversity}
Burke L (2020) {The Death and Life of an Admissions Algorithm}.
  \urlprefix\url{https://www.insidehighered.com/admissions/article/2020/12/14/u-texas-will-stop-using-controversial-algorithm-evaluate-phd}

\bibitem[{Chazette et~al(2021)Chazette, Brunotte, and Speith}]{9604587}
Chazette L, Brunotte W, Speith T (2021) Exploring explainability: A definition,
  a model, and a knowledge catalogue. In: 2021 IEEE 29th International
  Requirements Engineering Conference (RE), pp 197--208,
  \doi{10.1109/RE51729.2021.00025}

\bibitem[{Chib and Greenberg(1995)}]{chib1995understanding}
Chib S, Greenberg E (1995) Understanding the metropolis-hastings algorithm. The
  american statistician 49(4):327--335. \doi{10.1080/00031305.1995.10476177}

\bibitem[{Chouldechova(2017)}]{DBLP:journals/bigdata/Chouldechova17}
Chouldechova A (2017) Fair prediction with disparate impact: {A} study of bias
  in recidivism prediction instruments. Big Data 5(2):153--163.
  \doi{10.1089/big.2016.0047},
  \urlprefix\url{https://doi.org/10.1089/big.2016.0047}

\bibitem[{Clarkson et~al(2014)Clarkson, Finkbeiner, Koleini, Micinski, Rabe,
  and S{\'{a}}nchez}]{ClarksonFKMRS14:post}
Clarkson MR, Finkbeiner B, Koleini M, et~al (2014) Temporal logics for
  hyperproperties. In: Principles of Security and Trust - Third International
  Conference, {POST} 2014, Held as Part of the European Joint Conferences on
  Theory and Practice of Software, {ETAPS} 2014, Grenoble, France, April 5-13,
  2014, Proceedings, LNCS, vol 8414. Springer, pp 265--284,
  \urlprefix\url{https://doi.org/10.1007/978-3-642-54792-8\_15}

\bibitem[{D'Argenio et~al(2017)D'Argenio, Barthe, Biewer, Finkbeiner, and
  Hermanns}]{DBLP:conf/esop/DArgenioBBFH17}
D'Argenio PR, Barthe G, Biewer S, et~al (2017) Is your software on dope? -
  formal analysis of surreptitiously ``enhanced'' programs. In: Yang H (ed)
  Programming Languages and Systems - 26th European Symposium on Programming,
  {ESOP} 2017, Held as Part of the European Joint Conferences on Theory and
  Practice of Software, {ETAPS} 2017, Uppsala, Sweden, April 22-29, 2017,
  Proceedings, Lecture Notes in Computer Science, vol 10201. Springer, pp
  83--110, \doi{10.1007/978-3-662-54434-1\_4},
  \urlprefix\url{https://doi.org/10.1007/978-3-662-54434-1\_4}

\bibitem[{Donz{\'{e}} et~al(2013)Donz{\'{e}}, Ferr{\`{e}}re, and
  Maler}]{DBLP:conf/cav/DonzeFM13}
Donz{\'{e}} A, Ferr{\`{e}}re T, Maler O (2013) Efficient robust monitoring for
  {STL}. In: Sharygina N, Veith H (eds) Computer Aided Verification - 25th
  International Conference, {CAV} 2013, Saint Petersburg, Russia, July 13-19,
  2013. Proceedings, Lecture Notes in Computer Science, vol 8044. Springer, pp
  264--279, \doi{10.1007/978-3-642-39799-8\_19}

\bibitem[{Dressel and Farid(2018)}]{dressel2018accuracy}
Dressel J, Farid H (2018) The accuracy, fairness, and limits of predicting
  recidivism. Science advances 4(1):eaao5580

\bibitem[{Dwork et~al(2012)Dwork, Hardt, Pitassi, Reingold, and
  Zemel}]{dwork2012fairness}
Dwork C, Hardt M, Pitassi T, et~al (2012) Fairness through awareness. In:
  Proceedings of the 3rd innovations in theoretical computer science
  conference, pp 214--226

\bibitem[{Dworkin(1981)}]{10.2307/2265047}
Dworkin R (1981) What is equality? part 2: Equality of resources. Philosophy \&
  Public Affairs 10(4):283--345.
  \urlprefix\url{http://www.jstor.org/stable/2265047}

\bibitem[{Endsley(1995)}]{Endsley1995Toward}
Endsley MR (1995) Toward a theory of situation awareness in dynamic systems.
  Human Factors 37(1):32--64. \doi{10.1518/001872095779049543}

\bibitem[{Endsley(2017)}]{Endsley2017}
Endsley MR (2017) From here to autonomy: Lessons learned from
  human–automation research. Human Factors 59(1):5--27.
  \doi{10.1177/0018720816681350},
  \urlprefix\url{https://doi.org/10.1177/0018720816681350}, pMID: 28146676,
  {\href{https://arxiv.org/abs/https://doi.org/10.1177/0018720816681350}{{https://arxiv.org/abs/https://doi.org/10.1177/0018720816681350}}}

\bibitem[{{European Commission}(2011)}]{KOM-GDPR-proposal}
{European Commission} (2011) Proposal for a regulation of the european
  parliament and of the council on the protection of individuals with regard to
  the processing of personal data and on the free movement of such data
  (general data protection regulation) /* com/2012/011 final.
  \url{https://eur-lex.europa.eu/legal-content/EN/TXT/?uri=celex\%3A52012PC0011}

\bibitem[{{European Commission}(2021)}]{eu-0106-2021}
{European Commission} (2021) Laying down harmonised rules on artificial
  intelligence (artificial intelligence act) and amending certain union
  legislative acts (proposal for a regulation) no 0106/2021.
  \url{https://eur-lex.europa.eu/legal-content/EN/TXT/?uri=CELEX\%3A52021PC0206}

\bibitem[{{European Commission}(2023)}]{ai-act-amendments}
{European Commission} (2023) Amendments adopted by the european parliament on
  14 june 2023 on the proposal for a regulation of the european parliament and
  of the council on laying down harmonised rules on artificial intelligence
  (artificial intelligence act) and amending certain union legislative acts.
  \url{https://www.europarl.europa.eu/doceo/document/TA-9-2023-0236\_EN.html}

\bibitem[{{European Court of Justice}(1993)}]{ecj-c-127-92}
{European Court of Justice} (1993) C-127/92 - enderby ecli:eu:c:1993:859.
  \url{https://curia.europa.eu/juris/liste.jsf?language=en&num=C-127/92}

\bibitem[{{European Court of Justice}(1995)}]{ecj-c-400-93}
{European Court of Justice} (1995) C-400/93 - royal copenhagen
  ecli:eu:c:195:155.
  \url{https://curia.europa.eu/juris/liste.jsf?language=en&num=C-400/93}

\bibitem[{{European Court of Justice}(2014)}]{ecj-c-356-12}
{European Court of Justice} (2014) C-356/12 - glatzel ecli:eu:c:2014:350.
  \url{https://curia.europa.eu/juris/liste.jsf?language=en&num=C-356/12}

\bibitem[{{European Court of Justice}(2018)}]{ecj-c-414-16}
{European Court of Justice} (2018) C-414/16 - egenberger ecli:eu:c:2018:257.
  \url{https://curia.europa.eu/juris/liste.jsf?language=en&num=C-414/16}

\bibitem[{{European Parliament}(2020)}]{ep-2020-2012}
{European Parliament} (2020) European parliament resolution of 20 october 2020
  with recommendations to the commission on a framework of ethical aspects of
  artificial intelligence, robotics and related technologies.
  \url{https://www.europarl.europa.eu/doceo/document/TA-9-2020-0275_EN.html}

\bibitem[{{European Union}(2016{\natexlab{a}})}]{CFR}
{European Union} (2016{\natexlab{a}}) Charter of fundamental rights of the
  european union.
  \url{https://eur-lex.europa.eu/legal-content/EN/TXT/?uri=CELEX\%3A12012P\%2FTXT}

\bibitem[{{European Union}(2016{\natexlab{b}})}]{AEUV}
{European Union} (2016{\natexlab{b}}) Consolidated version of the treaty on the
  functioning of the european union.
  \url{https://eur-lex.europa.eu/legal-content/EN/TXT/?uri=CELEX\%3A12016ME\%2FTXT}

\bibitem[{Fainekos and Pappas(2009)}]{DBLP:journals/tcs/FainekosP09}
Fainekos GE, Pappas GJ (2009) Robustness of temporal logic specifications for
  continuous-time signals. Theor Comput Sci 410(42):4262--4291.
  \doi{10.1016/j.tcs.2009.06.021}

\bibitem[{Ferrer et~al(2021)Ferrer, Nuenen, Such, Coté, and Criado}]{9445793}
Ferrer X, Nuenen Tv, Such JM, et~al (2021) Bias and discrimination in {AI}: A
  cross-disciplinary perspective. IEEE Technology and Society Magazine
  40(2):72--80. \doi{10.1109/MTS.2021.3056293}

\bibitem[{Finkbeiner et~al(2015)Finkbeiner, Rabe, and
  S{\'{a}}nchez}]{FinkbeinerRS15:cav}
Finkbeiner B, Rabe MN, S{\'{a}}nchez C (2015) Algorithms for model checking
  {HyperLTL} and {HyperCTL$^*$}. In: {CAV} 2015, LNCS, vol 9206. Springer, pp
  30--48, \urlprefix\url{http://dx.doi.org/10.1007/978-3-319-21690-4\_3}

\bibitem[{Friedler et~al(2021)Friedler, Scheidegger, and
  Venkatasubramanian}]{10.1145/3433949}
Friedler SA, Scheidegger C, Venkatasubramanian S (2021) The (im)possibility of
  fairness: Different value systems require different mechanisms for fair
  decision making. Commun ACM 64(4):136–143. \doi{10.1145/3433949},
  \urlprefix\url{https://doi.org/10.1145/3433949}

\bibitem[{Gazda and Mousavi(2020)}]{Gazda2019}
Gazda M, Mousavi MR (2020) Logical characterisation of hybrid conformance. In:
  Czumaj A, Dawar A, Merelli E (eds) 47th International Colloquium on Automata,
  Languages, and Programming, {ICALP} 2020, July 8-11, 2020, Saarbr{\"{u}}cken,
  Germany (Virtual Conference), LIPIcs, vol 168. Schloss Dagstuhl -
  Leibniz-Zentrum f{\"{u}}r Informatik, pp 130:1--130:18,
  \doi{10.4230/LIPIcs.ICALP.2020.130},
  \urlprefix\url{https://doi.org/10.4230/LIPIcs.ICALP.2020.130}

\bibitem[{Gunning(2016)}]{DarpaXAI2016}
Gunning D (2016) Explainable artificial intelligence ({XAI}) (darpa-baa-16-53).
  Tech. rep., Arlington, VA, USA

\bibitem[{Guryan and Charles(2013)}]{10.2307/42919257}
Guryan J, Charles KK (2013) Taste-based or statistical discrimination: The
  economics of discrimination returns to its roots. The Economic Journal
  123(572):F417--F432. \urlprefix\url{http://www.jstor.org/stable/42919257}

\bibitem[{Hacker(2018)}]{hacker2018}
Hacker P (2018) {Teaching Fairness to Artificial Intelligence: Existing and
  Novel Strategies Against Algorithmic Discrimination Under EU Law}. Common
  Market Law Review (55):1143--1186.
  \urlprefix\url{https://ssrn.com/abstract=3164973}

\bibitem[{Hartmann(2006)}]{hartmann}
Hartmann F (2006) Diskriminierung durch {A}ntidiskriminierungsrecht?
  {M}öglichkeiten und {G}renzen eines postkategorialen
  {D}iskriminierungsschutzes in der {E}uropäischen {U}nion. EuZA -
  Europäische Zeitschrift für Arbeitsrecht p~24

\bibitem[{Heaven(2020)}]{predpol1}
Heaven WD (2020) {Predictive policing algorithms are racist. They need to be
  dismantled.}
  \urlprefix\url{{https://www.technologyreview.com/2020/07/17/1005396/predictive-policing-algorithms-racist-dismantled-machine-learning-bias-criminal-justice/}}

\bibitem[{{High-Level Expert Group on Artificial Intelligence}(2019)}]{aihleg}
{High-Level Expert Group on Artificial Intelligence} (2019) {Ethics Guidelines
  for Trustworthy AI}.
  \urlprefix\url{https://digital-strategy.ec.europa.eu/en/library/ethics-guidelines-trustworthy-ai}

\bibitem[{Hough et~al(2001)Hough, Oswald, and Ployhart}]{hough2001determinants}
Hough LM, Oswald FL, Ployhart RE (2001) Determinants, detection and
  amelioration of adverse impact in personnel selection procedures: Issues,
  evidence and lessons learned. International Journal of Selection and
  Assessment 9(1-2):152--194

\bibitem[{Ilvento(2019)}]{ilvento2019metric}
Ilvento C (2019) Metric learning for individual fairness. arXiv preprint
  arXiv:190600250

\bibitem[{Jacovi et~al(2021)Jacovi, Marasovi{\'c}, Miller, and
  Goldberg}]{jacovi2021formalizing}
Jacovi A, Marasovi{\'c} A, Miller T, et~al (2021) Formalizing trust in
  artificial intelligence: Prerequisites, causes and goals of human trust in
  {AI}. In: Proceedings of the 2021 ACM Conference on Fairness, Accountability,
  and Transparency, pp 624--635

\bibitem[{Jewson and Mason(1986)}]{jewson1986modes}
Jewson N, Mason D (1986) Modes of discrimination in the recruitment process:
  formalisation, fairness and efficiency. Sociology 20(1):43--63

\bibitem[{John et~al(2020)John, Vijaykeerthy, and Saha}]{john}
John PG, Vijaykeerthy D, Saha D (2020) Verifying individual fairness in machine
  learning models. In: Adams RP, Gogate V (eds) Proceedings of the Thirty-Sixth
  Conference on Uncertainty in Artificial Intelligence, {UAI} 2020, virtual
  online, August 3-6, 2020, Proceedings of Machine Learning Research, vol 124.
  {AUAI} Press, pp 749--758,
  \urlprefix\url{http://proceedings.mlr.press/v124/george-john20a.html}

\bibitem[{K{\"{a}}stner et~al(2021)K{\"{a}}stner, Langer, Lazar,
  Schom{\"{a}}cker, Speith, and Sterz}]{DBLP:conf/re/KastnerLLSSS21}
K{\"{a}}stner L, Langer M, Lazar V, et~al (2021) On the relation of trust and
  explainability: Why to engineer for trustworthiness. In: Yue T, Mirakhorli M
  (eds) 29th {IEEE} International Requirements Engineering Conference
  Workshops, {RE} 2021 Workshops, Notre Dame, IN, USA, September 20-24, 2021.
  {IEEE}, pp 169--175, \doi{10.1109/REW53955.2021.00031},
  \urlprefix\url{https://doi.org/10.1109/REW53955.2021.00031}

\bibitem[{Kim et~al(2016)Kim, Khanna, and Koyejo}]{10.5555/3157096.3157352}
Kim B, Khanna R, Koyejo O (2016) Examples are not enough, learn to criticize!
  criticism for interpretability. In: Proceedings of the 30th International
  Conference on Neural Information Processing Systems. Curran Associates Inc.,
  Red Hook, NY, USA, NIPS'16, p 2288–2296

\bibitem[{K{\"{o}}hl et~al(2018)K{\"{o}}hl, Hermanns, and
  Biewer}]{DBLP:conf/rv/KohlHB18}
K{\"{o}}hl MA, Hermanns H, Biewer S (2018) Efficient monitoring of real driving
  emissions. In: Colombo C, Leucker M (eds) Runtime Verification - 18th
  International Conference, {RV} 2018, Limassol, Cyprus, November 10-13, 2018,
  Proceedings, Lecture Notes in Computer Science, vol 11237. Springer, pp
  299--315, \doi{10.1007/978-3-030-03769-7_17}

\bibitem[{Lai and Tan(2019)}]{lai2019human}
Lai V, Tan C (2019) On human predictions with explanations and predictions of
  machine learning models: A case study on deception detection. In: Proceedings
  of the conference on fairness, accountability, and transparency, pp 29--38

\bibitem[{Langer et~al(2021{\natexlab{a}})Langer, Baum, Hartmann, Hessel,
  Speith, and Wahl}]{DBLP:conf/re/LangerBHHSW21}
Langer M, Baum K, Hartmann K, et~al (2021{\natexlab{a}}) Explainability
  auditing for intelligent systems: {A} rationale for multi-disciplinary
  perspectives. In: Yue T, Mirakhorli M (eds) 29th {IEEE} International
  Requirements Engineering Conference Workshops, {RE} 2021 Workshops, Notre
  Dame, IN, USA, September 20-24, 2021. {IEEE}, pp 164--168,
  \doi{10.1109/REW53955.2021.00030},
  \urlprefix\url{https://doi.org/10.1109/REW53955.2021.00030}

\bibitem[{Langer et~al(2021{\natexlab{b}})Langer, Oster, Speith, Hermanns,
  K{\"{a}}stner, Schmidt, Sesing, and Baum}]{DBLP:journals/ai/LangerOSHKSSB21}
Langer M, Oster D, Speith T, et~al (2021{\natexlab{b}}) What do we want from
  explainable artificial intelligence ({XAI})? - {A} stakeholder perspective on
  {XAI} and a conceptual model guiding interdisciplinary {XAI} research. Artif
  Intell 296:103,473. \doi{10.1016/j.artint.2021.103473},
  \urlprefix\url{https://doi.org/10.1016/j.artint.2021.103473}

\bibitem[{Larson et~al(2016)Larson, Mattu, Kirchner, and
  Angwin}]{compaspropubanalysis}
Larson J, Mattu S, Kirchner L, et~al (2016) {How We Analyzed the COMPAS
  Recidivism Algorithm}.
  \urlprefix\url{https://www.propublica.org/article/how-we-analyzed-the-compas-recidivism-algorithm}

\bibitem[{Lee and See(2004)}]{Lee2004}
Lee JD, See KA (2004) Trust in automation: Designing for appropriate reliance.
  Human factors 46(1):50--80

\bibitem[{Linardatos et~al(2021)Linardatos, Papastefanopoulos, and
  Kotsiantis}]{e23010018}
Linardatos P, Papastefanopoulos V, Kotsiantis S (2021) Explainable {AI}: A
  review of machine learning interpretability methods. Entropy 23(1).
  \doi{10.3390/e23010018},
  \urlprefix\url{https://www.mdpi.com/1099-4300/23/1/18}

\bibitem[{Looschelders(2012)}]{looschelders}
Looschelders D (2012) Diskriminierung und {S}chutz vor {D}iskriminierung im
  {P}rivatrecht. JZ - JuristenZeitung p 105

\bibitem[{Maler and Nickovic(2004)}]{DBLP:conf/formats/MalerN04}
Maler O, Nickovic D (2004) Monitoring temporal properties of continuous
  signals. In: Lakhnech Y, Yovine S (eds) Formal Techniques, Modelling and
  Analysis of Timed and Fault-Tolerant Systems, Joint International Conferences
  on Formal Modelling and Analysis of Timed Systems, {FORMATS} 2004 and Formal
  Techniques in Real-Time and Fault-Tolerant Systems, {FTRTFT} 2004, Grenoble,
  France, September 22-24, 2004, Proceedings, Lecture Notes in Computer
  Science, vol 3253. Springer, pp 152--166, \doi{10.1007/978-3-540-30206-3\_12}

\bibitem[{Malgieri(2020)}]{MalgieriFacct2020}
Malgieri G (2020) What “fairness” means? {A} linguistic and contextual
  interpretation from the {GDPR}. In: FAT* '20: Conference on Fairness,
  Accountability, and Transparency, Barcelona, Spain, January 27-30, 2020.
  {ACM}, pp 154--166, \doi{10.1145/3351095.3372868},
  \urlprefix\url{https://doi.org/10.1145/3351095.3372868}

\bibitem[{Mathews(????)}]{wired}
Mathews M (????) {A}re {Y}ou {R}eady for {S}oftware-{D}efined {E}verything?
  {Wired},
  \urlprefix\url{https://www.wired.com/insights/2013/05/are-you-ready-for-software-defined-everything/},
  {O}nline; accessed: 2023-06-23

\bibitem[{Matthias(2004)}]{Matthias2004-MATTRG}
Matthias A (2004) The responsibility gap: Ascribing responsibility for the
  actions of learning automata. Ethics and Information Technology
  6(3):175--183. \doi{10.1007/s10676-004-3422-1}

\bibitem[{Mecacci and de~Sio(2020)}]{Mecacci2020-MECMHC}
Mecacci G, de~Sio FS (2020) Meaningful human control as reason-responsiveness:
  The case of dual-mode vehicles. Ethics and Information Technology
  22(2):103--115. \doi{10.1007/s10676-019-09519-w}

\bibitem[{Mehrabi et~al(2021)Mehrabi, Morstatter, Saxena, Lerman, and
  Galstyan}]{mehrabi}
Mehrabi N, Morstatter F, Saxena N, et~al (2021) A survey on bias and fairness
  in machine learning. ACM Computing Surveys (CSUR) 54(6):1--35

\bibitem[{Meinke and Sindhu(2011)}]{DBLP:conf/tap/MeinkeS11}
Meinke K, Sindhu MA (2011) Incremental learning-based testing for reactive
  systems. In: Gogolla M, Wolff B (eds) Tests and Proofs - 5th International
  Conference, TAP@TOOLS 2011, Zurich, Switzerland, June 30 - July 1, 2011.
  Proceedings, Lecture Notes in Computer Science, vol 6706. Springer, pp
  134--151, \doi{10.1007/978-3-642-21768-5\_11}

\bibitem[{Methnani et~al(2021)Methnani, Aler~Tubella, Dignum, and
  Theodorou}]{10.3389/frai.2021.737072}
Methnani L, Aler~Tubella A, Dignum V, et~al (2021) Let me take over: Variable
  autonomy for meaningful human control. Frontiers in Artificial Intelligence
  4. \doi{10.3389/frai.2021.737072},
  \urlprefix\url{https://www.frontiersin.org/article/10.3389/frai.2021.737072}

\bibitem[{Meurrens(2021)}]{visa}
Meurrens S (2021) {The Increasing Role of {AI} in Visa Processing}.
  \urlprefix\url{https://canadianimmigrant.ca/immigrate/immigration-law/the-increasing-role-of-ai-in-visa-processing}

\bibitem[{Mittelstadt et~al(2016)Mittelstadt, Allo, Taddeo, Wachter, and
  Floridi}]{doi:10.1177/2053951716679679}
Mittelstadt BD, Allo P, Taddeo M, et~al (2016) The ethics of algorithms:
  Mapping the debate. Big Data \& Society 3(2):2053951716679,679.
  \doi{10.1177/2053951716679679},
  \urlprefix\url{https://doi.org/10.1177/2053951716679679}

\bibitem[{Molnar et~al(2020)Molnar, Casalicchio, and
  Bischl}]{DBLP:conf/pkdd/MolnarCB20}
Molnar C, Casalicchio G, Bischl B (2020) Interpretable machine learning - {A}
  brief history, state-of-the-art and challenges. In: Koprinska I, Kamp M,
  Appice A, et~al (eds) {ECML} {PKDD} 2020 Workshops - Workshops of the
  European Conference on Machine Learning and Knowledge Discovery in Databases
  {(ECML} {PKDD} 2020): SoGood 2020, {PDFL} 2020, {MLCS} 2020, {NFMCP} 2020,
  {DINA} 2020, {EDML} 2020, {XKDD} 2020 and {INRA} 2020, Ghent, Belgium,
  September 14-18, 2020, Proceedings, Communications in Computer and
  Information Science, vol 1323. Springer, pp 417--431,
  \doi{10.1007/978-3-030-65965-3\_28},
  \urlprefix\url{https://doi.org/10.1007/978-3-030-65965-3\_28}

\bibitem[{Mukherjee et~al(2020)Mukherjee, Yurochkin, Banerjee, and
  Sun}]{pmlr-v119-mukherjee20a}
Mukherjee D, Yurochkin M, Banerjee M, et~al (2020) Two simple ways to learn
  individual fairness metrics from data. In: III HD, Singh A (eds) Proceedings
  of the 37th International Conference on Machine Learning, Proceedings of
  Machine Learning Research, vol 119. PMLR, pp 7097--7107,
  \urlprefix\url{https://proceedings.mlr.press/v119/mukherjee20a.html}

\bibitem[{Nghiem et~al(2010)Nghiem, Sankaranarayanan, Fainekos, Ivancic, Gupta,
  and Pappas}]{DBLP:conf/hybrid/NghiemSFIGP10}
Nghiem T, Sankaranarayanan S, Fainekos GE, et~al (2010) Monte-carlo techniques
  for falsification of temporal properties of non-linear hybrid systems. In:
  Johansson KH, Yi W (eds) Proceedings of the 13th {ACM} International
  Conference on Hybrid Systems: Computation and Control, {HSCC} 2010,
  Stockholm, Sweden, April 12-15, 2010. {ACM}, pp 211--220,
  \doi{10.1145/1755952.1755983}

\bibitem[{Nguyen et~al(2017)Nguyen, Kapinski, Jin, Deshmukh, and
  Johnson}]{DBLP:conf/memocode/NguyenKJDJ17}
Nguyen LV, Kapinski J, Jin X, et~al (2017) Hyperproperties of real-valued
  signals. In: Talpin J, Derler P, Schneider K (eds) Proceedings of the 15th
  {ACM-IEEE} International Conference on Formal Methods and Models for System
  Design, {MEMOCODE} 2017, Vienna, Austria, September 29 - October 02, 2017.
  {ACM}, pp 104--113, \doi{10.1145/3127041.3127058}

\bibitem[{Noorman(2020)}]{sep-computing-responsibility}
Noorman M (2020) {Computing and Moral Responsibility}. In: Zalta EN (ed) The
  {Stanford} Encyclopedia of Philosophy, {S}pring 2020 edn. Metaphysics
  Research Lab, Stanford University

\bibitem[{Nunes and Jannach(2017)}]{nunes2017systematic}
Nunes I, Jannach D (2017) A systematic review and taxonomy of explanations in
  decision support and recommender systems. User Modeling and User-Adapted
  Interaction 27(3):393--444

\bibitem[{O'Neil(2016{\natexlab{a}})}]{jobs}
O'Neil C (2016{\natexlab{a}}) {How algorithms rule our working lives}.
  \urlprefix\url{https://www.theguardian.com/science/2016/sep/01/how-algorithms-rule-our-working-lives},
  {O}nline; accessed: 2023-06-23

\bibitem[{O'Neil(2016{\natexlab{b}})}]{10.5555/3002861}
O'Neil C (2016{\natexlab{b}}) Weapons of Math Destruction: How Big Data
  Increases Inequality and Threatens Democracy. Crown Publishing Group, USA

\bibitem[{Orcale(2019)}]{oraclehr}
Orcale (2019) {AI} in human resources: The time is now.
  \urlprefix\url{https://www.oracle.com/a/ocom/docs/applications/hcm/oracle-ai-in-hr-wp.pdf}

\bibitem[{{Organisation for Economic Co-operation and Development
  (OECD)}(2021)}]{6677}
{Organisation for Economic Co-operation and Development (OECD)} (2021)
  Artificial intelligence, machine learning and big data in finance:
  Opportunities, challenges and implications for policy makers.
  \urlprefix\url{https://www.oecd.org/finance/financial-markets/Artificial-intelligence-machine-learning-big-data-in-finance.pdf}

\bibitem[{Pessach and Shmueli(2022)}]{10.1145/3494672}
Pessach D, Shmueli E (2022) A review on fairness in machine learning. ACM
  Comput Surv 55(3). \doi{10.1145/3494672},
  \urlprefix\url{https://doi.org/10.1145/3494672}

\bibitem[{Pnueli(1977)}]{DBLP:conf/focs/Pnueli77}
Pnueli A (1977) The temporal logic of programs. In: 18th Annual Symposium on
  Foundations of Computer Science, Providence, Rhode Island, USA, 31 October -
  1 November 1977. {IEEE} Computer Society, pp 46--57,
  \doi{10.1109/SFCS.1977.32}

\bibitem[{Rawls(1985)}]{10.2307/2265349}
Rawls J (1985) Justice as fairness: Political not metaphysical. Philosophy \&
  Public Affairs 14(3):223--251.
  \urlprefix\url{http://www.jstor.org/stable/2265349}

\bibitem[{Rawls(1999)}]{rawls1999theory}
Rawls J (1999) A theory of justice: Revised edition. Harvard university press

\bibitem[{Rawls(2001)}]{rawls2001justice}
Rawls J (2001) Justice as fairness: A restatement. Harvard University Press

\bibitem[{Ribeiro et~al(2016{\natexlab{a}})Ribeiro, Singh, and
  Guestrin}]{DBLP:journals/corr/RibeiroSG16a}
Ribeiro MT, Singh S, Guestrin C (2016{\natexlab{a}}) Model-agnostic
  interpretability of machine learning. CoRR abs/1606.05386.
  \urlprefix\url{http://arxiv.org/abs/1606.05386},
  {\href{https://arxiv.org/abs/1606.05386}{{https://arxiv.org/abs/1606.05386}}}

\bibitem[{Ribeiro et~al(2016{\natexlab{b}})Ribeiro, Singh, and
  Guestrin}]{10.1145/2939672.2939778}
Ribeiro MT, Singh S, Guestrin C (2016{\natexlab{b}}) ``{W}hy should {I} trust
  you?'': Explaining the predictions of any classifier. In: Proceedings of the
  22nd ACM SIGKDD International Conference on Knowledge Discovery and Data
  Mining. Association for Computing Machinery, New York, NY, USA, KDD '16, p
  1135–1144, \doi{10.1145/2939672.2939778},
  \urlprefix\url{https://doi.org/10.1145/2939672.2939778}

\bibitem[{Riesenhuber(2021)}]{Riesenhuber}
Riesenhuber K (2021) {BDSG § 26 Datenverarbeitung für Zwecke des
  Beschäftigungsverhältnisses Rn. 79f}. In: Wolff SBA (ed) BeckOK
  Datenschutzrecht. C.H. Beck, Munich

\bibitem[{Rockafellar and Wets(2009)}]{rockafellar2009variational}
Rockafellar RT, Wets RJB (2009) Variational analysis, vol 317. Springer Science
  \& Business Media

\bibitem[{Rosen and Krithivasan(2012)}]{rosen2012discrete}
Rosen KH, Krithivasan K (2012) Discrete mathematics and its applications: with
  combinatorics and graph theory. Tata McGraw-Hill Education

\bibitem[{Rowe(2022)}]{Rowe2022-ROWCAR-4}
Rowe T (2022) Can a risk of harm itself be a harm? Analysis 81(4):694--701.
  \doi{10.1093/analys/anab033}

\bibitem[{Rubinstein(1981)}]{DBLP:books/lib/Rubinstein81}
Rubinstein RY (1981) Simulation and the Monte Carlo method. Wiley series in
  probability and mathematical statistics, Wiley,
  \urlprefix\url{https://www.worldcat.org/oclc/07275104}

\bibitem[{Sankaranarayanan and
  Fainekos(2012)}]{DBLP:conf/hybrid/SankaranarayananF12}
Sankaranarayanan S, Fainekos G (2012) Falsification of temporal properties of
  hybrid systems using the cross-entropy method. In: Dang T, Mitchell IM (eds)
  Hybrid Systems: Computation and Control (part of {CPS} Week 2012), HSCC'12,
  Beijing, China, April 17-19, 2012. {ACM}, pp 125--134,
  \doi{10.1145/2185632.2185653},
  \urlprefix\url{https://doi.org/10.1145/2185632.2185653}

\bibitem[{Sanneman and Shah(2020)}]{sanneman2020situation}
Sanneman L, Shah JA (2020) A situation awareness-based framework for design and
  evaluation of explainable {AI}. In: International Workshop on Explainable,
  Transparent Autonomous Agents and Multi-Agent Systems, Springer, pp 94--110

\bibitem[{Schlicker and Langer(2021)}]{schlicker2021towards}
Schlicker N, Langer M (2021) Towards warranted trust: A model on the relation
  between actual and perceived system trustworthiness. In: Mensch und Computer
  2021. p 325--329

\bibitem[{Schlicker et~al(2021)Schlicker, Langer, {\"{O}}tting, Baum,
  K{\"{o}}nig, and Wallach}]{DBLP:journals/chb/SchlickerLOBKW21}
Schlicker N, Langer M, {\"{O}}tting SK, et~al (2021) What to expect from
  opening up 'black boxes'? comparing perceptions of justice between human and
  automated agents. Comput Hum Behav 122:106,837.
  \doi{10.1016/j.chb.2021.106837},
  \urlprefix\url{https://doi.org/10.1016/j.chb.2021.106837}

\bibitem[{Schlicker et~al(2022)Schlicker, Uhde, Baum, Hirsch, and
  Langer}]{schlicker2022trustworthiness}
Schlicker N, Uhde A, Baum K, et~al (2022) Calibrated trust as a result of
  accurate trustworthiness assessment -- introducing the trustworthiness
  assessment model. PsyArXiv Preprints \doi{10.31234/osf.io/qhwvx}

\bibitem[{Schwab(2006)}]{schwab}
Schwab D (2006) Schranken der {V}ertragsfreiheit durch die
  {A}ntidiskriminierungsrichtlinien und ihre {U}msetzung in {D}eutschland.
  DNotZ - Deutsche Notar-Zeitschrift p 649

\bibitem[{Santoni~de Sio and van~den Hoven(2018)}]{10.3389/frobt.2018.00015}
Santoni~de Sio F, van~den Hoven J (2018) Meaningful human control over
  autonomous systems: A philosophical account. Frontiers in Robotics and AI 5.
  \doi{10.3389/frobt.2018.00015},
  \urlprefix\url{https://www.frontiersin.org/article/10.3389/frobt.2018.00015}

\bibitem[{Smith and Vogell(2021)}]{tenantscreen}
Smith E, Vogell H (2021) {How Your Shadow Credit Score Could Decide Whether You
  Get an Apartment }.
  \urlprefix\url{https://www.propublica.org/article/how-your-shadow-credit-score-could-decide-whether-you-get-an-apartment},
  {O}nline; accessed: 2023-06-23

\bibitem[{Speith(2022)}]{10.1145/3531146.3534639}
Speith T (2022) A review of taxonomies of explainable artificial intelligence
  ({XAI}) methods. In: 2022 ACM Conference on Fairness, Accountability, and
  Transparency. Association for Computing Machinery, New York, NY, USA, FAccT
  '22, p 2239–2250, \doi{10.1145/3531146.3534639},
  \urlprefix\url{https://doi.org/10.1145/3531146.3534639}

\bibitem[{Sterz et~al(2021)Sterz, Baum, Lauber{-}R{\"{o}}nsberg, and
  Hermanns}]{DBLP:conf/re/SterzBLH21}
Sterz S, Baum K, Lauber{-}R{\"{o}}nsberg A, et~al (2021) Towards perspicuity
  requirements. In: Yue T, Mirakhorli M (eds) 29th {IEEE} International
  Requirements Engineering Conference Workshops, {RE} 2021 Workshops, Notre
  Dame, IN, USA, September 20-24, 2021. {IEEE}, pp 159--163,
  \doi{10.1109/REW53955.2021.00029},
  \urlprefix\url{https://doi.org/10.1109/REW53955.2021.00029}

\bibitem[{Tabuada et~al(2012)Tabuada, Balkan, Caliskan, Shoukry, and
  Majumdar}]{DBLP:conf/emsoft/TabuadaBCSM12}
Tabuada P, Balkan A, Caliskan SY, et~al (2012) Input-output robustness for
  discrete systems. In: Proceedings of the 12th International Conference on
  Embedded Software, {EMSOFT} 2012, part of the Eighth Embedded Systems Week,
  ESWeek 2012, Tampere, Finland, October 7-12, 2012. {ACM}, pp 217--226,
  \urlprefix\url{http://doi.acm.org/10.1145/2380356.2380396}

\bibitem[{Talbert(2019)}]{sep-moral-responsibility}
Talbert M (2019) {Moral Responsibility}. In: Zalta EN (ed) The {Stanford}
  Encyclopedia of Philosophy, {W}inter 2019 edn. Metaphysics Research Lab,
  Stanford University

\bibitem[{Tay et~al(2022)Tay, Woo, Hickman, Booth, and
  D’Mello}]{tay2021conceptual}
Tay L, Woo SE, Hickman L, et~al (2022) A conceptual framework for investigating
  and mitigating machine-learning measurement bias (mlmb) in psychological
  assessment. Advances in Methods and Practices in Psychological Science 5(1).
  \doi{10.1177/25152459211061337},
  \urlprefix\url{https://doi.org/10.1177/25152459211061337}

\bibitem[{{Technavio}(2022)}]{SDE}
{Technavio} (2022) Software defined everything ({SDE}) market by end-user and
  geography - forecast and analysis 2022-2026.
  \urlprefix\url{https://www.technavio.com/report/software-defined-everything-sde-market-industry-analysis},
  {O}nline; accessed: 2023-06-23

\bibitem[{{The Council of the European Union}(2000)}]{EUdir2000/78/EC}
{The Council of the European Union} (2000) Council directive 2000/78/{EC} of 27
  november 2000 establishing a general framework for equal treatment in
  employment and occupation.
  \urlprefix\url{https://eur-lex.europa.eu/legal-content/EN/TXT/?uri=CELEX:32000L0078}

\bibitem[{{The Council of the European Union}(2004)}]{EUdir2004/113/EC}
{The Council of the European Union} (2004) Council directive 2004/113/{EC} of
  13 december 2004 implementing the principle of equal treatment between men
  and women in the access to and supply of goods and services.
  \urlprefix\url{https://eur-lex.europa.eu/legal-content/EN/TXT/?uri=celex\%3A32004L0113}

\bibitem[{{The European Parliament and the Council of the European
  Union}(2017)}]{LEX:32017R1151}
{The European Parliament and the Council of the European Union} (2017)
  {Commission Regulation (EU) 2017/1151}.
  \urlprefix\url{http://data.europa.eu/eli/reg/2017/1151/oj}

\bibitem[{Thüsing(2013)}]{Thuesing}
Thüsing G (2013) European Labour Law, § 3 Protection against discrimination.
  C.H. Beck

\bibitem[{Thüsing(2019)}]{thuesingZfA}
Thüsing G (2019) Das künftige {A}nti-{D}iskriminierungsrecht als
  {H}erausforderung für {W}issenschaft und {P}raxis. ZfA - Zeitschrift für
  Arbeitsrecht p 241

\bibitem[{Tutuianu et~al(2015)Tutuianu, Bonnel, Ciuffo, Haniu, Ichikawa,
  Marotta, Pavlovic, and Steven}]{TUTUIANU201561}
Tutuianu M, Bonnel P, Ciuffo B, et~al (2015) Development of the world-wide
  harmonized light duty test cycle ({WLTC}) and a possible pathway for its
  introduction in the european legislation. Transportation Research Part D:
  Transport and Environment 40(Supplement C):61 -- 75.
  \doi{10.1016/j.trd.2015.07.011}

\bibitem[{{United Nations}(2013)}]{nedc}
{United Nations} (2013) {UN Vehicle Regulations - 1958 Agreement, Revision 2,
  Addendum 100, Regulation No. 101, Revision 3 ---
  E/ECE/324/Rev.2/Add.100/Rev.3}.
  \urlprefix\url{http://www.unece.org/trans/main/wp29/wp29regs101-120.html}

\bibitem[{{United Nations Educational, Scientific and Cultural Organization
  (UNESCO)}(2021)}]{aiunesco}
{United Nations Educational, Scientific and Cultural Organization (UNESCO)}
  (2021) {Recommendation on the ethics of artificial intelligence}.
  \urlprefix\url{https://unesdoc.unesco.org/ark:/48223/pf0000380455}

\bibitem[{Volpato and Tretmans(2015)}]{DBLP:journals/eceasst/VolpatoT15}
Volpato M, Tretmans J (2015) Approximate active learning of nondeterministic
  input output transition systems. Electron Commun Eur Assoc Softw Sci Technol
  72. \doi{10.14279/tuj.eceasst.72.1008}

\bibitem[{Wachter et~al(2020)Wachter, Mittelstadt, and
  Russell}]{wachter2020bias}
Wachter S, Mittelstadt B, Russell C (2020) Bias preservation in machine
  learning: the legality of fairness metrics under eu non-discrimination law. W
  Va L Rev 123:735. \doi{10.2139/ssrn.3792772},
  \urlprefix\url{http://dx.doi.org/10.2139/ssrn.3792772}

\bibitem[{{Washington State}(2020)}]{washstatefacerecbill}
{Washington State} (2020) {Certification of Enrollment: Engrossed Substitute
  Senate Bill 6280 ('Washington State Facial Recognition Law')}.
  \urlprefix\url{https://lawfilesext.leg.wa.gov/biennium/2019-20/Pdf/Bills/Senate\%20Passed\%20Legislature/6280-S.PL.pdf?q=20210513071229}

\bibitem[{Waters and Miikkulainen(2014)}]{Waters_Miikkulainen_2014}
Waters A, Miikkulainen R (2014) Grade: Machine learning support for graduate
  admissions. AI Magazine 35(1):64. \doi{10.1609/aimag.v35i1.2504},
  \urlprefix\url{https://ojs.aaai.org/index.php/aimagazine/article/view/2504}

\bibitem[{Zehlike et~al(2021)Zehlike, Yang, and Stoyanovich}]{zehlike}
Zehlike M, Yang K, Stoyanovich J (2021) Fairness in ranking: {A} survey. CoRR
  abs/2103.14000. \urlprefix\url{https://arxiv.org/abs/2103.14000},
  {\href{https://arxiv.org/abs/2103.14000}{{https://arxiv.org/abs/2103.14000}}}

\bibitem[{Zemel et~al(2013)Zemel, Wu, Swersky, Pitassi, and
  Dwork}]{zemel2013learning}
Zemel R, Wu Y, Swersky K, et~al (2013) Learning fair representations. In:
  International conference on machine learning, PMLR, pp 325--333

\bibitem[{Ziegert and Hanges(2005)}]{ziegert2005employment}
Ziegert JC, Hanges PJ (2005) Employment discrimination: the role of implicit
  attitudes, motivation, and a climate for racial bias. Journal of applied
  psychology 90(3):553

\bibitem[{{\noop{ZZZ}}Bertrand and Mullainathan(2004)}]{bertrand2004emily}
{\noop{ZZZ}}Bertrand M, Mullainathan S (2004) Are emily and greg more
  employable than lakisha and jamal? a field experiment on labor market
  discrimination. American economic review 94(4):991--1013

\bibitem[{{\noop{ZZZ}}Hoff and Bashir(2015)}]{hoff2015trust}
{\noop{ZZZ}}Hoff KA, Bashir M (2015) Trust in automation: Integrating empirical
  evidence on factors that influence trust. Human factors 57(3):407--434

\bibitem[{{\noop{ZZZ}}Lahoti et~al(2019){\noop{ZZZ}}Lahoti, Gummadi, and
  Weikum}]{lahoti2019}
{\noop{ZZZ}}Lahoti P, Gummadi KP, Weikum G (2019) ifair: Learning individually
  fair data representations for algorithmic decision making. In: 2019 ieee 35th
  international conference on data engineering (icde), IEEE, pp 1334--1345

\bibitem[{{\noop{ZZZ}}Langer et~al(2022){\noop{ZZZ}}Langer, K{\"o}nig, Back,
  and Hemsing}]{langer2021trust}
{\noop{ZZZ}}Langer M, K{\"o}nig CJ, Back C, et~al (2022) Trust in artificial
  intelligence: Comparing trust processes between human and automated trustees
  in light of unfair bias. Journal of Business and Psychology

\end{thebibliography}


\end{document}